\documentclass[11pt]{article}
\usepackage[a4paper,hmargin=1.0in,vmargin=1.0in]{geometry}
\usepackage{amsmath,amsthm,amssymb,sectsty,setspace}
\usepackage[round]{natbib}
\usepackage[usenames,dvipsnames]{xcolor}
\usepackage[linktocpage=true,pagebackref=true,colorlinks,linkcolor=BrickRed,citecolor=blue,bookmarks,bookmarksopen,bookmarksnumbered]{hyperref}

\theoremstyle{plain}
\newtheorem{theorem}{Theorem}[section]
\newtheorem{lemma}[theorem]{Lemma}
\newtheorem{claim}[theorem]{Claim}

\newtheorem{observation}[theorem]{Observation}
\newcommand{\qedsymb}{\hfill{\rule{2mm}{2mm}}}
\renewenvironment{proof}{\begin{trivlist} \item[\hspace{\labelsep}{\bf \noindent Proof.\/}] }{\qedsymb\end{trivlist}}

\newcommand{\MyAbove}[2]{\genfrac{}{}{0pt}{}{#1}{#2}}
\newcommand{\bs}[1]{\boldsymbol{#1}}
\newcommand{\eps}{\epsilon}
\newcommand{\opt}{\mathrm{OPT}}
\newcommand{\objfunc}{\Phi}
\newcommand{\newobj}{\Psi}
\newcommand{\mymiddle}{\mathrm{mid}}
\newcommand{\mysparse}{\mathrm{sparse}}
\newcommand{\myheavy}{\mathrm{heavy}}
\newcommand{\mylight}{\mathrm{light}}
\newcommand{\mytop}{\mathrm{top}}
\newcommand{\mycore}{\mathrm{core}}
\newcommand{\mycap}{\mathrm{capacity}}
\newcommand{\mycross}{\mathrm{cross}}
\newcommand{\mytime}{\mathrm{Time}}
\newcommand{\myspace}{\mathrm{spaced}}
\newcommand{\myfix}{\mathrm{FixCrossing}}
\newcommand{\mythin}{\mathrm{Thin}}
\newcommand{\mybulky}{\mathrm{Bulky}}
\newcommand{\myextra}{\mathrm{Extra}}
\newcommand{\mybest}{\mathrm{Best}}

\DeclareMathOperator*{\argmin}{\arg\!\min}
\DeclareMathOperator*{\argmax}{\arg\!\max}

\sectionfont{\large} \subsectionfont{\normalsize}
\allowdisplaybreaks
\makeindex

\begin{document}

\begin{titlepage}

\title{Approximation Algorithms for \\
The Generalized Incremental Knapsack Problem}
\author{%
Yuri Faenza\thanks{Department of Industrial Engineering and Operations Research, Columbia University, 500 W.\ 120th Street, New York, NY 10027. Email: \{yf2414,lz2573\}@columbia.edu.}%
\and Danny Segev\thanks{Department of Statistics and Operations Research, School of Mathematical Sciences, Tel Aviv University, Tel Aviv 69978, Israel. Email: segev.danny@gmail.com.}%
\and Lingyi Zhang\footnotemark[1]}
\date{}
\maketitle

\thispagestyle{empty}

\begin{abstract}
We introduce and study a discrete multi-period extension of the classical knapsack problem, dubbed \emph{generalized incremental knapsack}. In this setting, we are given a set of $n$ items, each associated with a non-negative weight, and $T$ time periods with non-decreasing capacities $W_1 \leq \dots \leq W_T$. When item $i$ is inserted at time $t$, we gain a profit of $p_{it}$; however, this item remains in the knapsack for all subsequent periods. The goal is to decide if and when to insert each item, subject to the time-dependent capacity constraints, with the objective of maximizing our total profit. Interestingly, this setting subsumes as special cases a number of recently-studied incremental knapsack problems, all known to be strongly NP-hard.

Our first contribution comes in the form of a polynomial-time $(\frac{1}{2}-\epsilon)$-approximation for the generalized incremental knapsack problem. This result is based on a reformulation as a single-machine sequencing problem, which is addressed by blending dynamic programming techniques and the classical Shmoys-Tardos algorithm for the generalized assignment problem. Combined with further enumeration-based self-reinforcing ideas and newly-revealed structural properties of nearly-optimal solutions, we turn our basic algorithm into a quasi-polynomial time approximation scheme (QPTAS). Hence, under widely believed complexity assumptions, this finding rules out the possibility that generalized incremental knapsack is APX-hard.
\end{abstract}


\end{titlepage}

\thispagestyle{empty}
\tableofcontents

\newpage
\setcounter{page}{1}

\section{Introduction} \label{sec:introduction}

In many scenarios, classical optimization models are too simplistic to faithfully capture applications arising in real-life environments. Much research has therefore been devoted to extend fundamental well-studied models to more  realistic, yet still algorithmically tractable settings. A very common extension along these lines introduces time-dependent components, adding a computationally-challenging layer on top of the inherent complexity of the underlying problem. For instance, \emph{maximum flow over time}, originally introduced in the seminal work of \cite{ford_fulkerson_1956}, has recently received a great deal of attention~\citep{skutella2009introduction, Gross2012, LinJ15}. Additional examples for such settings include time-expanded versions of various packing problems~\citep{Caprara2002, adjiashvili2014time, Epstein2019}, network scheduling over time~\citep{BOLAND201434, AkridaCGKS19}, adaptive routing over time~\citep{graf2020, ismaili2017}, and facility location over time~\citep{Farahani2009, Nickel2019}, just to mention a few. 

\paragraph{Incremental knapsack problems.} In this paper, we investigate a multi-period extension of the classical knapsack problem. To provide initial intuition for the inner-workings of our model, consider the problem faced by urban planners, who intend to build infrastructural facilities over the course of several years, under budget constraints. Once an infrastructure has been built, its construction cost cannot be recovered. With a quantification of each infrastructure's annual contribution to the well-being of the community once it is in place, the goal is to maximize the total benefit over the course of the planning horizon (hence, the mayor's chances of being re-elected). A host of additional applications, such as planning the incremental growth of highways and networks, community development, and memory allocation can be found within several of the undermentioned papers and the references therein.

Computational questions of this nature can be modeled via multi-period knapsack extensions, collectively dubbed as \emph{incremental knapsack} problems. In such settings, the input ingredients consist in a set of $n$ items with strictly positive weights $\{ w_i \}_{i \in [n]}$, a collection of $T$ time periods with non-decreasing capacities $W_1 \leq \cdots \leq W_T$, and a set of item-period profits, on which we further elaborate below. We say that a sequence of item sets ${\cal S}=(S_1,\dots,S_T)$ is a \emph{chain} when $S_1 \subseteq \cdots \subseteq S_{T} \subseteq [n]$; here, $S_t$ represents the subset of items inserted into the knapsack up to and including time period $t$. As such, the chain ${\cal S}$ is \emph{feasible} when $w(S_t)\leq W_t$ for every $t \in [T]$. Our fundamental assumption is that, for each item $i \in [n]$ and time $t \in [T]$, we are given a non-negative parameter $p_{it}$, corresponding to the profit we obtain when item $i$ is inserted at time $t$ (i.e., when $i \in S_t\setminus S_{t-1}$, with the convention that $S_0 = \emptyset$). Hence, the cumulative profit of any chain ${\cal S}=(S_1,\dots,S_T)$ over all time periods is captured by $\objfunc({\cal S}) = \sum_{t \in [T]} \sum_{i \in S_t \setminus S_{t-1}} p_{it}$. We refer to the resulting formulation as the \emph{generalized incremental knapsack} problem.

\paragraph{Directly-related settings.} To our knowledge, due to its double-dependency on both the item and time period in question, the above-mentioned profit structure makes generalized incremental knapsack the most inclusive incremental knapsack problem studied so far. Probably the simplest such problem is \emph{time-invariant incremental knapsack}, where each item $i$ is assumed to contribute a profit of $\phi_i$ to each period starting at its insertion time, corresponding to product-form profits, $p_{it}=(T+1-t)\cdot \phi_i$. Surprisingly, unlike the basic knapsack problem, \cite{bienstock2013approximation} showed that this extension is strongly NP-hard. On the positive side, \cite{Faenza2018APF} proposed a polynomial-time approximation scheme (PTAS) based on rounding fractional solutions to an appropriate disjunctive relaxation. In the broader \emph{incremental knapsack} problem, we have $p_{it}= \phi_i \cdot \sum_{\tau=t}^T \Delta_{\tau}$, where $\Delta_{\tau}\geq 0$ is a time-dependent scaling factor; in this context, \cite{AouadS2020} have very recently obtained a PTAS, leveraging approximate dynamic programming ideas. We refer the reader to a number of additional resources related to incremental knapsack problems \citep{Sharp07, Hartline08, ye2016, della2018approximating, della2019approximating} for a deeper look into these settings. 

In contrast, the flexibility of our item- and time-dependent profit structure allows us to capture a variety of situations. For instance, when an item $i$ gains a profit of $\phi_{i\tau}$ for each period $\tau$, starting at its insertion time, we can set $p_{it}=\sum_{\tau=t}^T \phi_{i\tau}$. If, moreover, the per-period profits $\phi_{i\tau}$ are discounted by a factor of $c_{\tau-t}$ after $\tau-t$ time units have elapsed since the insertion of item $i$, we set $p_{it}=\sum_{\tau=t}^T c_{\tau-t} \phi_{i\tau}$. More broadly, the generalized incremental knapsack problem allows the profits $p_{it}$ to be completely unrelated, and in particular, to possibly be non-monotone in $t$. To our knowledge, prior to the present paper, this problem was not known to admit any non-trivial approximation guarantees. Moreover, we are not aware of any way to leverage existing techniques in the above-mentioned papers for dealing with the broad generality of our profit structure. 

\subsection{Results and techniques}

\paragraph{Constant-factor approximation.} Our first contribution comes in the form of a polynomial-time constant-factor approximation for the generalized incremental knapsack problem, whose specifics are provided in Section~\ref{sec:2-approx}.

\begin{theorem}\label{thm:half-approx-intro}
For any fixed $\eps \in (0,\frac{1}{2})$, the generalized incremental knapsack problem can be approximated in polynomial time within factor $\frac{1}{2} - \eps$.
\end{theorem}

The starting point of our algorithm is a reformulation of generalized incremental knapsack as a single-machine sequencing problem, where feasible chains are mapped to item permutations $\pi$, with $\pi(i_1)<\pi(i_2)$ implying that the insertion time of item $i_1$ occurs no later than that of $i_2$, potentially leaving item $i_2$ out of the knapsack. Based on this reformulation, we partition any given permutation into ``heavy'' and ``light'' chains of items, depending on how their weights compare to the combined weight of previously-inserted items. Guided by this decomposition, our approach consists of devising two approximation schemes, one competing against the best-possible profit due to heavy contributions and the other against the analogous quantity due to light contributions. Technically speaking, for heavy chains, we make use of dynamic programming ideas, whereas for light chains, we further reformulate this setting as a highly-structured instance of the generalized assignment problem, which is solved to super-optimality via the Shmoys-Tardos algorithm~\citeyearpar{ShmoysT93}, and truncated to a feasible near-optimal solution. 

\paragraph{Quasi-PTAS.} As previously mentioned, special cases of the generalized incremental knapsack problem are known to be strongly NP-hard, admitting a PTAS under specific profit-structure assumptions. A natural question is whether one can design efficient algorithms with the same degree of accuracy for generalized incremental knapsack, without any such assumption. Towards this goal, our second main contribution establishes the following result.

\begin{theorem}\label{thm:QPTAS-intro}
The generalized incremental knapsack problem admits a quasi-polynomial time approximation scheme.
\end{theorem}

Hence, under widely believed complexity assumptions, this finding rules out the possibility that generalized incremental knapsack is APX-hard, thus making it substantially different from other knapsack extensions, such as the generalized assignment problem (see brief discussion in Section~\ref{subsec:related_work}). The main idea behind the above-mentioned algorithm lies in a  ``self-improving'' procedure, which combines certain ingredients of our constant-factor approach along with further guessing methods and structural modifications to convert a black-box $\alpha$-approximation into a $\frac{1}{2-\alpha}$-approximation, as accounted for in Section~\ref{sec:qptas-one}. When iteratively applied, these improvements lead to a $(1-\epsilon)$-approximation, albeit with a running time  exponential in $1/\eps$, $\log n$, and $\log(w_{\max} / w_{\min})$. In essence, the last term emerges from a dual manipulation of both chain-related and permutation-related representations, where the  dependency on the extremal weight ratio appears to be inevitable. To bypass this obstacle, in Section~\ref{sec:qptas-two} we employ our algorithm as a subroutine on a sequence of weakly-dependent subinstances, each with a polynomial $w_{\max} / w_{\min}$-value, obtained through a structural analysis of nearly-optimal solutions. The resulting approach is shown to be implementable in quasi-polynomial time for any given instance of the generalized incremental knapsack problem, thus proving Theorem~\ref{thm:QPTAS-intro}. 

\subsection{Related knapsack extensions} \label{subsec:related_work}

In the \emph{maximum generalized assignment} problem, we are given $n$ items and $m$ capacitated buckets. Assigning an item $j$ to a bucket $i$ takes $w_{ij}$ capacity units while generating a profit of $p_{ij}$. The goal is to compute a feasible item-to-bucket assignment whose overall profit is maximized. For the minimization variant of this problem, \cite{ShmoysT93} proposed an LP-based $2$-approximation, which was observed by \cite{ChekuriK05} to be easily adaptable to obtain a $1/2$-approximation for the maximization variant. Interestingly, these algorithmic ideas will be useful within one of the subroutines employed by our approach. \cite{Feige2006} attained a $(1 - 1/e + \delta)$-approximation, for some absolute constant $\delta > 0$, which is currently the best known performance guarantee for maximum generalized assignment. Earlier constant-factor approximations were obtained by \cite{Fleischer2006}, \cite{Nutov2006}, and \cite{Cohen2006}.

In the \emph{unsplittable flow on a path} problem, we are given an edge-capacitated path as well as a collection of tasks. Each task is characterized by its own subpath, profit, and demand. The goal is to select a subset of tasks of maximum total profit, under the constraint that the overall demand of the selected tasks along each edge resides within its capacity. The currently best polynomial-time approximation in this context is $5/3+ \eps$, for any fixed $\eps>0$, due to \cite{grandoni20185}, who improved on earlier constant-factor guarantees by \cite{Bonsma2014}, \cite{Anagnostopoulos2018}, and \cite{Calinescu2011}. In parallel, unsplittable flow on a path admits a quasi-PTAS, as shown by \cite{bansal2006quasi} and by \cite{batra2014new}. From a technical perspective, the methods involved are very different from those exploited in our paper, and it is unclear whether algorithmic ideas in one setting are migratable to the other. 

\section{A Polynomial-Time \texorpdfstring{\boldmath{$(\frac{1}{2} - \eps)$}}{}-Approximation}\label{sec:2-approx}
\
In this section, we present our first approximability result for the generalized incremental knapsack problem, showing that the optimal profit can be efficiently approached within a factor arbitrarily close to $\frac{ 1 }{ 2 }$. The specifics of this finding, along with its corresponding running time, are formally stated in the next theorem.

\begin{theorem}\label{thm:half-approx}
For any accuracy level $\eps \in (0,\frac{1}{2})$, the generalized incremental knapsack problem can be approximated within factor $\frac{1}{2} - \eps$. The running time of our algorithm is $O(n^{O(1/ \eps^2)} \cdot |{\cal I}|^{O(1)})$, where $|{\cal I}|=\Theta(n\log \|w\|_{\infty} + nT\log \|p\|_{\infty}+T\log \|W\|_{\infty} )$ stands for the input size.
\end{theorem}

\paragraph{Outline.} For simplicity of presentation, we start off Section~\ref{subsec:reformulation} by proposing an equivalent formulation of the generalized incremental knapsack problem as a single-machine sequencing problem. Given this reformulation, we explain in Section~\ref{subsec:decomp_outline} how the profit function can be decomposed into ``heavy'' and ``light'' item contributions. Somewhat informally, with respect to an unknown optimal sequencing solution, the marginal contribution of each item to the overall profit will be classified as being either heavy or light, depending on the item's weight and position on the timeline. Guided by this decomposition, our approach consists of devising two approximation schemes, one competing against the best-possible profit due to heavy contributions (Section~\ref{sec:heavy}) and the other against the analogous quantity due to light contributions (Section~\ref{sec:light}). The best of these algorithms will be shown to provide an approximation guarantee of $\frac{1}{2} - \eps$, thereby deriving Theorem~\ref{thm:half-approx}. It is worth pointing out that the techniques involved in competing against light contributions will be further utilized in Sections~\ref{sec:qptas-one} and~\ref{sec:qptas-two} to obtain an approximation scheme for general instances, albeit in quasi-polynomial time.

\subsection{An equivalent sequencing formulation}\label{subsec:reformulation}

In what follows, we present an equivalent sequencing reformulation for the generalized incremental knapsack problem. As explained in subsequent sections, the interchangeability between these formulations allows us to describe our algorithmic ideas and to analyze their performance guarantees with greater ease. For this purpose, we proceed by arguing that the generalized incremental knapsack problem can be rephrased as a sequencing problem on a single machine as follows:
\begin{itemize}
\item Let $\pi : [n] \to [n]$ be a permutation of the underlying items, where $\pi(i)$ stands for the position of item $i$.

\item By viewing the weight of each item as its processing time, we define the completion time of item $i$ with respect to $\pi$ as $C_{ \pi }( i ) = \sum_{j \in [n] : \pi(j) \leq \pi(i)} w_j$. Accordingly, the profit $\varphi_{ \pi }( i )$ of this item is given by the largest profit we can gain by inserting $i$ at a time period whose capacity occurs no earlier than $C_{ \pi }( i )$, namely, $\varphi_{ \pi }( i ) = \max \{ p_{i,t} : t \in [T+1] \text{ and } W_t \geq C_{ \pi }( i ) \}$, with the convention that $W_{T+1} = \infty$ and $p_{i, T+1} = 0$ for every item $i$.

\item The overall profit of the permutation $\pi$ is specified by $\newobj( \pi ) = \sum_{i \in [n]} \varphi_{ \pi }( i )$. Our objective is to compute a permutation whose profit is maximized.
\end{itemize}

The next lemma captures the equivalence between the item-introducing perspective of the generalized incremental knapsack problem and the sequencing perspective described above. 

\begin{lemma}\label{lem:reformulation}
Any feasible chain ${\cal S}$ can be mapped to a permutation $\pi_{\cal S}$ with $\newobj(\pi_{\cal S}) \geq \objfunc({\cal S})$. Conversely, any permutation $\pi$ of a subset of the items can be mapped to a feasible chain ${\cal S}_{\pi}$ with $\objfunc({\cal S}_{\pi}) = \newobj(\pi)$.
\end{lemma}
\begin{proof}
First, given a feasible chain ${\cal S}$, we construct the permutation $\pi_{\cal S}$ as follows:
\begin{itemize}
    \item For each $t \in [T]$, let $\pi^t$ be an arbitrary permutation of the items introduced in this period, $S_t\setminus S_{t-1}$.  In addition, let $\pi^{T+1}$ be an arbitrary permutation of the remaining items, i.e., those in $[n] \setminus S_T$.
    
    \item The permutation $\pi_{\cal S}$ is defined as the concatenation of $\pi^{1}, \ldots, \pi^{T+1}$ in this order. Namely, for $i \in S_t\setminus S_{t-1}$ with $t \in [T]$, we have $\pi_{\cal S}(i)=\pi^t(i)+|S_{t-1}|$, whereas for $i \in [n] \setminus S_T$, we have $\pi_{\cal S}(i) = \pi^{T+1}(i) + |S_T|$.
\end{itemize}
To prove that $\newobj(\pi_{\cal S}) \geq \objfunc({\cal S})$, it suffices to argue that $\varphi_{\pi_{\cal S}}(i) \geq p_{i, t_i}$ for every item $i \in S_T$, where $t_i$ stands for the insertion time of item $i$ with respect to the chain ${\cal S}$. To derive this relation, note that $C_{\pi_{\cal S}} (i) \leq w(S_{t_i}) \leq W_{t_i}$ for any such item, where the last inequality follows from the feasibility of ${\cal S}$. Therefore, $\varphi_{\pi_{\cal S}}(i) = \max \{ p_{i,t} : t \in [T+1] \text{ and } W_t \geq C_{ \pi_{\cal S} }( i ) \} \geq p_{i, t_i}$.

Conversely, given a permutation $\pi$ of any subset of items, we construct a chain ${\cal S}_{ \pi }$ that includes all items with a completion time of at most $W_T$. Specifically, the insertion time $t_i$ of each such item $i$ will be the time period that maximizes $p_{i,t_i}$ over the set $\{t \in [T] : W_t \geq C_\pi(i)\}$. As such, the chain ${\cal S}_{ \pi }$ is indeed feasible, since $w(S_t) \leq \sum_{i \in [n]: C_{\pi}(i) \leq W_t} w_i \leq W_t$ for every $t \in [T]$. To show that $\objfunc( {\cal S}_{ \pi } ) = \newobj( \pi )$, it remains to explain why $p_{i, t_i} = \varphi_{\pi} (i)$ for inserted items and why $\varphi_{\pi}(i) = 0$ for non-inserted ones. To this end, note that our choice for the insertion time $t_{i}$ follows the definition of $\varphi_{\pi}(i)$ to the letter, meaning that $p_{i, t_i} = \varphi_{\pi} (i)$. On the other hand, for any item $i$ we do not insert to ${\cal S}_{ \pi }$, one has $\varphi_{\pi}(i) = 0$, since $C_{\pi}(i)>W_T$.
\end{proof}

\subsection{Profit decomposition and high-level overview} \label{subsec:decomp_outline}

In what follows, we focus our attention on the sequencing formulation and present a decomposition of the profit function $\newobj$ into ``heavy'' and ``light'' contributions, collected over geometrically-increasing intervals. With the necessary definitions in place, we outline how a decomposition of this nature guides us in proposing two approximation schemes, to separately compete against heavy and light contributions. The main result of this section, as stated in Theorem~\ref{thm:half-approx}, will eventually be derived by taking the more profitable of these approaches.

For simplicity of presentation, we assume without loss of generality that $\eps \in (0,\frac{1}{2})$, and moreover, that $\frac{ 1 }{ \eps }$ is an integer. In addition, we assume that $w_{\min} = \min_{ i \in [n] } w_i = 3$; the latter property can easily be enforced through scaling all item weights $w_i$ and time period capacities $W_t$ by a factor of $\frac{3}{w_{\min}}$.

\paragraph{Profit decomposition.} We begin by geometrically partitioning the interval $[0, \sum_{i \in [n]} w_i]$ by powers of $1 + \eps$ into a collection of intervals ${\cal I}_0, \ldots, {\cal I}_{K}$, where $K = \lceil \log_{1 + \eps} (\sum_{i \in [n]} w_i) \rceil$. Specifically, ${\cal I}_0 = [0,1]$ and ${\cal I}_k = ((1 + \eps)^{k-1}, (1 + \eps)^k]$ for $ k \in  [K ]$. With this definition, the profit $\newobj( \pi ) = \sum_{i \in [n]} \varphi_{ \pi }( i )$ of any permutation $\pi$ can be expressed by summing item contributions according to the interval in which their completion times fall, i.e.,
\[ \newobj( \pi ) = \sum_{k \in [K]_0} \sum_{ \MyAbove{ i \in [n] : }{ C_{ \pi }( i ) \in {\cal I}_k } } \varphi_{ \pi }( i ) \ . \]
    
We say that item $i$ is $k$-heavy when $w_i \geq \eps^2 \cdot (1 + \eps)^k$; otherwise, this item is $k$-light. We denote the sets of $k$-heavy and $k$-light items by $H_k$ and $L_k$, respectively, noting that $H_0 \supseteq H_1 \supseteq \cdots \supseteq H_K$ and that $L_k = [n] \setminus H_k$ for every $k$. As a side note, one can easily verify that all items are $0$-heavy (i.e., $H_0 = [n]$), by recalling that $w_{\min{}}=3$ and $\epsilon < \frac{1}{2}$. Consequently, the profit $\newobj( \pi )$ can be refined by separating $k$-heavy and $k$-light items, namely,
\begin{equation}\label{eq:profit_decomp}
\newobj( \pi ) = \underbrace{ \sum_{k \in [K]_0} \sum_{ \MyAbove{ i \in H_k : }{ C_{ \pi }( i ) \in {\cal I}_k } } \varphi_{ \pi }( i ) }_{ \newobj_{ \myheavy } ( \pi ) } + \underbrace{ \sum_{k \in [K]_0} \sum_{ \MyAbove{ i \in L_k : }{ C_{ \pi }( i ) \in {\cal I}_k } } \varphi_{ \pi }( i ) }_{ \newobj_{ \mylight } ( \pi ) } \ . 
\end{equation}
As shown above, we designate the first and second terms in the above expression by $\newobj_{ \myheavy } ( \pi )$ and $\newobj_{ \mylight } ( \pi )$, respectively.

\paragraph{Overview.} Let $\pi^*$ be an optimal permutation, with $\newobj( \pi^* ) = \newobj_{ \myheavy } ( \pi^* ) + \newobj_{ \mylight } ( \pi^* )$. The remainder of this section is dedicated to presenting two approximation schemes that would separately compete against $\newobj_{ \myheavy } ( \pi^* )$ and $\newobj_{ \mylight } ( \pi^* )$:
\begin{itemize}
\item {\em Heavy contributions}: Section~\ref{sec:heavy} explains how dynamic programming ideas allow us to efficiently compute a permutation $\pi_{ \myheavy }: [n] \to [n]$ satisfying $\newobj( \pi_{ \myheavy } ) \geq (1 - \eps) \cdot \newobj_{ \myheavy } ( \pi^* )$. The required running time  will be $O( n^{ O(1/\eps^2) } \cdot |{\cal I}|)$.  

\item {\em Light contributions}: Section~\ref{sec:light} argues that the generalized assignment algorithm of \cite{ShmoysT93} can be leveraged to compute a permutation $\pi_{ \mylight } : [n] \to [n]$  satisfying $\newobj( \pi_{ \mylight } ) \geq (1 - \eps) \cdot \newobj_{ \mylight } ( \pi^* )$. This algorithm can be implemented in $O( ( \frac{ |{\cal I}| }{ \eps } )^{O(1)} )$ time.
\end{itemize}
Consequently, to establish the approximation guarantee stated in Theorem~\ref{thm:half-approx}, we pick the more profitable permutation out of $\pi_{ \myheavy }$ and $\pi_{ \mylight }$, to obtain a profit of
\begin{eqnarray*}
\max \left\{ \newobj \left( \pi_{ \myheavy } \right), \newobj \left( \pi_{ \mylight } \right) \right\} & \geq & \frac{ 1 }{ 2 } \cdot \left( \newobj \left( \pi_{ \myheavy } \right) + \newobj \left( \pi_{ \mylight } \right) \right) \\
& \geq & \frac{ 1 - \eps }{ 2 } \cdot \left( \newobj_{ \myheavy } \left( \pi^* \right) + \newobj_{ \mylight } \left( \pi^* \right) \right) \\
& = & \frac{ 1 - \eps }{ 2 } \cdot \newobj \left( \pi^* \right) \ .
\end{eqnarray*}

\subsection{Algorithm for heavy contributions} \label{sec:heavy}

In what follows, we present an $O( n^{ O(1/\eps^2) } \cdot |{\cal I}|)$-time dynamic programming approach for computing a permutation $\pi_{ \myheavy }$ with a profit of $\newobj( \pi_{ \myheavy } ) \geq (1 - \eps) \cdot \newobj_{ \myheavy } ( \pi^* )$. 

\subsubsection{Preliminaries} \label{subsec:heavy_prelim}

The intuition behind our algorithm begins with the observation that, in order to compete against $\newobj_{ \myheavy } ( \pi^* )$, we can safely eliminate items that are classified as light with respect to the interval in which their completion time falls. While the remaining items will be shifted back in the residual permutation, potentially being completed in a lower-index interval, each of them will still be heavy. To formalize these notions, for a subset of items $S \subseteq [n]$ and a permutation $\pi : S \to [|S|]$, we say that the pair $(S,\pi)$ is bulky if, for every $k \in [K]_0$, all items with a completion time in ${\cal I}_k$ are $k$-heavy, i.e., $\{ i \in S : C_{ \pi }( i ) \in {\cal I}_k \} \subseteq H_k$. The next claim shows that bulky pairs can attain a total profit of at least $\newobj_{ \myheavy } ( \pi^* )$.

\begin{lemma} \label{lem:good_bulky}
There exist a subset of items $S \subseteq [n]$ and a permutation $\pi : S \to [|S|]$ such that $(S,\pi)$ is bulky and $\sum_{i \in S} \varphi_{ \pi }( i ) \geq \newobj_{ \myheavy } ( \pi^* )$.
\end{lemma}
\begin{proof}
With respect to the optimal permutation $\pi^*$, we define a new permutation $\pi$ by eliminating, for every $k \in [K]_0$, all items $i \in L_k$ with $C_{ \pi^* }( i ) \in {\cal I}_k$. The subset $S$ will consist of the remaining items. To see why $(S,\pi)$ is bulky, note that $C_{ \pi }( i ) \leq C_{ \pi^* }( i )$ for any $i \in S$, meaning that each such item is still heavy with respect to the interval that contains $C_{ \pi }( i )$, since $H_0 \supseteq \cdots \supseteq H_K$. In terms of profit, the latter observation implies that, for every item $i \in S$,
\begin{eqnarray*}
\varphi_{ \pi }( i ) & = & \max \left\{ p_{i,t}: t \in [T+1] \text{  and } W_t \geq C_{\pi}(i) \right\} \\
& \geq & \max \left\{ p_{i,t}: t \in [T+1] \text{ and } W_t \geq C_{\pi^*}(i) \right\} \\
& = & \varphi_{ \pi^* }( i ) \ . 
\end{eqnarray*}
Summing the above inequality over all items in $S$, we have $\sum_{i \in S} \varphi_{\pi}(i) \geq \sum_{i \in S}\varphi_{\pi^*}(i) = \newobj_{ \myheavy } ( \pi^* )$, where the latter equality holds since every eliminated item does not contribute toward $\newobj_{ \myheavy } ( \pi^* )$ but rather  toward $\newobj_{ \mylight } ( \pi^* )$.
\end{proof}

\paragraph{Additional notation.} For a bulky pair $(S,\pi)$, we define its top index as $\mytop(S,\pi) = \max \{ k \in [K]_0 : \{ C_{ \pi }( i ) : i \in S \} \cap {\cal I}_k \neq \emptyset \}$, that is, the largest index of an interval that contains at least one completion time. In addition, we define $\mycore(S)$ as the set of $\min \{ \frac{ 1 }{ \eps^2 }, |S| \}$ heaviest items in $S$, breaking ties by adding to $\mycore(S)$ small-index items before large-index ones. Finally, the makespan of $(S,\pi)$ corresponds to the maximum completion time of an item in $S$ with respect to the permutation $\pi$; in our case, this measure identifies with $w(S)$.

\subsubsection{The continuous dynamic program}\label{subsec:continuous_dp}

The technical crux in restricting attention to bulky pairs will be exhibited through our dynamic programming formulation. As formally explained below, by focusing on the dual objective of makespan minimization, we prove the existence of a well-hidden optimal substructure within the sequencing problem.

\paragraph{States.} Each state $(k, \psi_k, {\cal Q}_k)$ of our dynamic program consists of the following parameters, whose precise role will be clarified once their corresponding value function is presented:
\begin{itemize}
\item The index of the current interval $k$, taking values in $[K]_0$.

\item The total profit $\psi_k$ collected thus far, due to items whose completion time falls in ${\cal I}_0, \ldots {\cal I}_k$. For the time being, $\psi_k$ will be treated as a continuous parameter, taking values in $[0, \sum_{i \in [n]} \max_{t \in [T]} p_{i,t}]$. 

\item The core ${\cal Q}_k$ of the set of items whose completion time falls in ${\cal I}_0, \ldots {\cal I}_k$. By definition of $\mycore(\cdot)$, this parameter is restricted to item sets of cardinality at most $\frac{ 1 }{ \eps^2 }$. 
\end{itemize}
It is important to emphasize that, since $\psi_k$ is a continuous parameter, the dynamic programming formulation below is still not algorithmic in nature, and should be viewed as a characterization of optimal solutions. In Section~\ref{subsec:discretize_DP}, we explain how to discretize the parameter $\psi_k$ to take polynomially-many values while incurring only an $\eps$-loss in profit. 

\paragraph{Value function.} The value function $F(k, \psi_k, {\cal Q}_k)$ represents the minimum makespan $w(S)$ that can be attained over all bulky pairs $(S,\pi)$ that satisfy the following conditions:
\begin{enumerate}
\item {\em Top index}: $\mytop(S,\pi) \leq k$.

\item {\em Total profit}: $\newobj(\pi) \geq \psi_k$. 

\item {\em Core}: $\mycore(S) = {\cal Q}_k$.
\end{enumerate}
For ease of presentation, we denote the collection of bulky pairs that meet conditions~1-3 by $\mybulky(k, \psi_k, {\cal Q}_k)$.  When the latter set is empty, we define $F(k, \psi_k, {\cal Q}_k) = \infty$. With these definitions, Lemma~\ref{lem:good_bulky} proves in retrospect the existence of a bulky pair $(S,\pi) \in \mybulky(K, \newobj_{ \myheavy } ( \pi^* ), \mycore(S))$ with $F(K, \newobj_{ \myheavy } ( \pi^* ), \mycore(S)) < \infty$. Therefore, had we been able to compute the maximal value $\psi^*$ that satisfies $F(K, \psi^*, {\cal Q}_K) < \infty$ for some core ${\cal Q}_K$ of at most $\frac{ 1 }{ \eps^2 }$ items, its corresponding bulky pair would have guaranteed a profit of at least $\psi^* \geq \newobj_{ \myheavy } ( \pi^* )$. 

\paragraph{Optimal substructure.} To this end, we proceed by unveiling the optimal substructure that allows us to compute the value function $F$ by means of dynamic programming. In order to gain ample intuition, suppose that $(S, \pi)$ is a bulky pair that attains $F(k, \psi_k, {\cal Q}_k)$. Then, we argue that, by eliminating from $S$ the set items $Q$ whose completion time falls within the interval ${\cal I}_k$, one obtains a bulky pair that attains $F(k-1, \psi_{k-1}, {\cal Q}_{k-1})$, where the residual profit $\psi_{k-1}$ is obtained by removing from $\psi_{k}$ the contribution of items in $Q$ and ${\cal Q}_{k-1}$ is an appropriately chosen core. 

Formally, suppose that $\mybulky(k, \psi_k, {\cal Q}_k) \neq \emptyset$, and let $(S,\pi)$ be a bulky pair that minimizes $w(S)$ over this set. Let $Q = \{ i \in S : C_{ \pi }( i ) \in {\cal I}_k \}$ be the set of items in $S$ whose completion time with respect to $\pi$ falls in the interval ${\cal I}_k$. Note that since $\mytop(S,\pi) \leq k$, completion times cannot fall in ${\cal I}_{k+1}, \ldots, {\cal I}_K$. We first argue that $|Q| \leq \frac{ 1 }{ \eps }$. To verify this claim, note that since $(S,\pi)$ is bulky, $Q \subseteq H_k$. As a result, every item in $Q$ has a weight of at least $\eps^2 \cdot (1 + \eps)^k$, while ${\cal I}_k = ((1 + \eps)^{k-1}, (1 + \eps)^k]$, meaning that we necessarily have $|Q| \leq \frac{ (1 + \eps)^k - (1 + \eps)^{k-1} }{ \eps^2 \cdot (1 + \eps)^k } \leq \frac{ 1 }{ \eps }$.

Now, let us define the pair $(\hat{S},\hat{\pi})$, where $\hat{S} = S \setminus Q$ and $\hat{\pi} : \hat{S} \to [|\hat{S}|]$ is the permutation where items in $\hat{S}$ follow their relative order in $\pi$, that is, for any pair of items $i_1$ and $i_2$, we have $\hat{\pi}( i_1 ) < \hat{\pi}( i_2 )$ if and only if ${\pi}( i_1 ) < {\pi}( i_2 )$. In addition, let $\psi_{k-1} = [ \psi_k - \sum_{i \in Q} \varphi_{ \pi }( i ) ]^+$ and ${\cal Q}_{k-1} = \mycore( \hat{S} )$, where $[ x ]^+ = \max \{ x, 0 \}$. These definitions directly ensure that $(\hat{S},\hat{\pi}) \in \mybulky(k-1, \psi_{k-1}, {\cal Q}_{k-1})$. Moreover, as we show next, $(\hat{S},\hat{\pi})$ forms an optimal solution with respect to the latter state. 

\begin{lemma}
$w( \hat{S} ) = F(k-1, \psi_{k-1}, {\cal Q}_{k-1})$.
\end{lemma}
\begin{proof}
Suppose there exists some bulky pair $(\tilde{S},\tilde{\pi}) \in \mybulky(k-1, \psi_{k-1}, {\cal Q}_{k-1})$ with $w( \tilde{S} ) < w( \hat{S} )$. We first claim that $\tilde{S} \cap Q = \emptyset$. To verify this property, had there been an item $i \in \tilde{S} \cap Q$, its weight would satisfy $w_i \geq \eps^2 \cdot (1 + \eps)^k$, since $Q \subseteq H_k$. On the other hand, since $\mytop(\tilde{S},\tilde{\pi}) \leq k-1$, the completion times of all items in $\tilde{S}$ with respect to $\tilde{\pi}$ reside within the union of ${\cal I}_0, \ldots, {\cal I}_{k-1}$, which is the interval $[0,(1 + \eps)^{k-1}]$, implying that $w( \tilde{S} ) \leq (1 + \eps)^{k-1}$. Therefore, since $\mycore( \tilde{S} )$ is the set of $\min \{ \frac{ 1 }{ \eps^2 }, |\tilde{S}| \}$ heaviest items in $\tilde{S}$, regardless of how ties are broken we must have $i \in \mycore( \tilde{S} )$. We have just arrived at a contradiction: Since $\mycore( \tilde{S} ) = {\cal Q}_{k-1} = \mycore( \hat{S} ) = \mycore( S \setminus Q )$, it follows that $i \notin Q$. 

Knowing that $\tilde{S} \cap Q = \emptyset$, we can extend the permutation $\tilde{\pi} : \tilde{S} \to [|\tilde{S}|]$ to $\tilde{S}^+ = \tilde{S} \cup Q$ by appending the set of items $Q$ in exactly the same order as they appear in $\pi$. Letting $\tilde{\pi}^+ : \tilde{S}^+ \to [|\tilde{S}^+|]$ be the resulting permutation, we next argue that $(\tilde{S}^+, \tilde{\pi}^+)$ is in fact a feasible solution to precisely the same subproblem with respect to which which $(S,\pi)$ is optimal. The proof of this structural result is provided in Appendix~\ref{app:proof_clm_feasible_subproblem}. 

\begin{claim} \label{clm:feasible_subproblem}
$(\tilde{S}^+, \tilde{\pi}^+) \in \mybulky(k, \psi_k, {\cal Q}_k)$.  
\end{claim}

We have just arrived at a contradiction to the fact that $(S,\pi)$ minimizes $w(S)$ over the set $\mybulky(k, \psi_k, {\cal Q}_k)$, by observing that $w( \tilde{S}^+ ) = w( \tilde{S} ) + w( Q ) < w( \hat{S} ) + w( Q ) = w(S)$.
\end{proof}

\paragraph{Recursive equations.} In light of this structural characterization, to obtain a recursive equation for $F(k, \psi_k, {\cal Q}_k)$, it suffices to ``guess'' the collection of items $Q$, their internal permutation $\pi_Q$, the residual profit requirement $\psi_{k-1}$, and the resulting core ${\cal Q}_{k-1}$. Formally, $F(k, \psi_k, {\cal Q}_k)$ is given by minimizing $F ( k-1, \psi_{k-1}, {\cal Q}_{k-1} ) + w( Q )$ over all choices of $Q$, $\pi_Q$, $\psi_{k-1}$, and ${\cal Q}_{k-1}$ that simultaneously satisfy the following conditions:
\begin{enumerate}
    \item Top index: $F ( k-1, \psi_{k-1}, {\cal Q}_{k-1} ) + w( Q ) \leq (1 + \eps)^k$. This constraint ensures that,  with the addition of $Q$, all items can still be packed within ${\cal I}_0, \ldots, {\cal I}_k$.
    
    \item Total profit: $\psi_{k-1} \geq [ \psi_k - \sum_{i \in Q} \varphi_{ \pi_Q }^{ \rightsquigarrow }( i ) ]^+$, where the term $\varphi_{ \pi_Q }^{ \rightsquigarrow }( i )$ denotes the profit of item $i$ with respect to the permutation $\pi_Q$, when its completion time is increased by $F ( k-1, \psi_{k-1}, {\cal Q}_{k-1} )$. This constraint guarantees that, by appending $\pi_Q$, we obtain a total profit of at least $\psi_k$.
    
    \item Core: ${\cal Q}_{k-1} \cap Q = \emptyset$, $\mycore( {\cal Q}_{k-1} \cup Q ) = {\cal Q}_k$, $Q \subseteq H_k$, and $|Q| \leq \frac{ 1 }{ \eps }$. These constraints ensure a correct core update as a result of adding the item set $Q$, where the latter set consists of at most $\frac{ 1 }{ \eps }$ items, each restricted to being $k$-heavy. To better understand the requirement $\mycore( {\cal Q}_{k-1} \cup Q ) = {\cal Q}_k$, note that the core resulting from the addition of $Q$ can be computed without a complete knowledge of all previously packed items, as all those outside the current core ${\cal Q}_{k-1}$ are irrelevant for this purpose (i.e., too light to be one of the $\frac{ 1 }{ \eps^2 }$ heaviest).
\end{enumerate}

\subsubsection{Discretization and final algorithm} \label{subsec:discretize_DP}

As previously mentioned, due to the continuity of the profit requirement $\psi_k$, it remains to propose an appropriate discretization of this parameter, so that we obtain a polynomially-sized state space with only negligible loss in profit.

\paragraph{The discrete program $\bs{\tilde{F}}$.} To this end, we alter the underlying state space of our dynamic program, by restricting the continuous parameter $\psi_k$ to a finite set of values, ${\cal D}_{\psi} = \{  d \cdot  \frac{ \eps p_{\max} }{ n } : d \in [\frac{n^2}{\eps}]_0 \}$. Here, $p_{\max}$ is the maximum profit attainable by any single item, i.e., $p_{\max} = \max \{ p_{it}: i \in [n], t \in [T], \text{ and } w_i \leq W_t \}$. We make use of $\tilde{F}(k, \psi_k, Q_k)$ to designate the value function $F$ restricted to the resulting set of states, and similarly, $\widetilde{\mybulky}(k, \psi_k, {\cal Q}_k)$ will stand for the collection of bulky pairs that meet conditions~1-3. As a side note, beyond the additional restriction on $\psi_k$, both $\tilde{F}$ and $\widetilde{\mybulky}$ are defined identically to $F$ and $\mybulky$. 

\paragraph{Analysis.} We remind the reader that, in Section~\ref{subsec:continuous_dp}, the quantity $\psi^*$ was defined as the maximal value satisfying $F(K, \psi^*, {\cal Q}_K) < \infty$ for some core ${\cal Q}_K$ of at most $\frac{ 1 }{ \eps^2 }$ items, noting that its corresponding bulky pair guarantees a profit of at least $\psi^* \geq \newobj_{ \myheavy } ( \pi^* )$. In order to establish a parallel claim with respect to the discretized program $\tilde{F}$, we prove in Lemma~\ref{lem:descretize_psi} a lower bound of $(1 - \eps) \cdot \psi^*$ on the analogous quantity $\tilde{\psi}$ that satisfies $\tilde{F}(K, \tilde{\psi}, {\cal Q}_K) < \infty$; the proof is provided in Appendix~\ref{app: proof_lem_descretize_psi}.  It follows that our dynamic program computes a bulky pair $(S, \pi)$ in which the permutation $\pi$ has a profit of $\newobj( \pi ) \geq \tilde{\psi} \geq (1 - \eps) \cdot \psi^* \geq (1 - \eps) \cdot \newobj_{ \myheavy } ( \pi^* )$.

\begin{lemma} \label{lem:descretize_psi}
There exists a value $\tilde{\psi} \in {\cal D}_{\psi}$ such that $\tilde{\psi} \geq (1 - \eps) \cdot \psi^*$ and such that $\tilde{F}(K, \tilde{\psi}, {\cal Q}_K) < \infty$ for some core ${\cal Q}_K$.
\end{lemma}
 
\paragraph{Running time.} We first observe that the function $\tilde{F}(k, \psi_k, {\cal Q}_k)$ needs to be evaluated over $O(n^{O(1/ \eps^2)} \cdot |{\cal I}|)$ possible states. Indeed, there are $O(K)$ choices for the interval index $k$, where by definition, $K =  \lceil \log_{1 + \eps} (\sum_{i \in [n]} w_i) \rceil = O( \frac{ |{\cal I}| }{ \eps } )$. As for the profit parameter $\psi_k$, following its restriction to the set ${\cal D}_{\psi}$, we ensure that $\psi_k$ takes only $| {\cal D}_{\psi} | = O( \frac{n^2}{\eps} )$ values. Finally, since the core ${\cal Q}_k \subseteq [n]$ consists of at most $\frac{ 1 }{ \eps^2 }$ items, there are only $O(n^{O(1/ \eps^2)})$ subsets to consider.

Now, evaluating each state requires minimizing the restricted function $\tilde{F}(k-1, \psi_{k-1}, {\cal Q}_{k-1}) + w(Q)$ over all choices of $Q$, $\pi_Q$, $\psi_{k-1}$, and ${\cal Q}_{k-1}$ that simultaneously satisfy conditions~1-3 of the recursive equations (see Section~\ref{subsec:continuous_dp}). In this context, the number of joint configurations for these parameters is $O(n^{O(1/ \eps^2)} )$. Specifically, the profit parameter $\psi_{k-1}$ and the core ${\cal Q}_{k-1}$ respectively take $O( \frac{n^2}{\eps} )$ and $O(n^{O(1/ \eps^2)})$ values as before. In addition, the number of choices for the augmenting set $Q$ is $O( n^{ O(1/\eps) } )$, due to being comprised of at most $\frac{ 1 }{ \eps }$ items, and there are 
only $O( ( \frac{ 1 }{ \eps } )^{ O( 1 / \eps ) } )$ permutations $\pi_Q$ of these items. To summarize, we incur an overall running time of $O(n^{O(1/ \eps^2)} \cdot |{\cal I}|)$.

\subsection{Algorithm for light contributions} \label{sec:light}

In this section, we construct a suitably-defined instance of the maximum generalized assignment problem, intended to compete against $\newobj_{ \mylight } ( \pi^* )$. We show that, when applied to this highly-structured instance, the LP-based algorithm of \cite{ShmoysT93} can be leveraged to identify in $O( ( \frac{ |{\cal I}| }{ \eps } )^{O(1)} )$ time a permutation $\pi_{ \mylight }$ with a profit of $\newobj( \pi_{ \mylight } ) \geq (1 - 13\eps) \cdot \newobj_{ \mylight } ( \pi^* )$. 

\subsubsection{Instance construction}\label{subsec:light_construction}

\paragraph{Intuition.} The general intuition behind our construction resides in viewing the intervals ${\cal I}_1, \ldots, {\cal I}_{K-1}$ as distinct buckets, to which items should be assigned subject to capacity constraints. Clearly, this perspective lacks the extra flexibility of the sequencing formulation, where items may be crossing between multiple successive intervals. In addition, any item-to-bucket assignment has to be associated with a specific profit a-priori, whereas the sequencing-related profits depend on the exact completion time of each item. As explained in the sequel, our approach bypasses the first obstacle by focusing on light items, for which greedy repacking of rounded solutions will be argued to be near-optimal. In regard to the second obstacle, we will allow seemingly unattainable profits, showing that appropriately scaled fractional solutions can be rounded to attain these profits up to negligible loss in optimality.

\paragraph{The construction.} Guided by this intuition, we define an instance of the maximum generalized assignment problem as follows:
\begin{itemize}
\item {\em Buckets}: For every $k \in [K-1]$, we set up a bucket ${\cal B}_k$. The capacity of this bucket is $\mycap( {\cal B}_k ) = (1 + \eps)^k - (1 + \eps)^{ k-1 }$, i.e., precisely the length of the interval ${\cal I}_k$. It is worth mentioning that there are no buckets corresponding to the intervals ${\cal I}_0$ and ${\cal I}_K$.

\item {\em Items}: The set of items is still $[n]$, where each item has a weight of $w_i$.

\item {\em Allowed assignments and profits}: An item $i$ can be assigned to bucket ${\cal B}_k$ only when $i$ is $(k+1)$-light. For such an assignment, our profit is $q_{ik} = \max \{ p_{i,t} : t \in [T+1] \text{ and } W_t \geq (1 + \eps)^k \}$.
\end{itemize}
The goal is to compute a capacity-feasible assignment whose total profit is maximized.

\paragraph{IP formulation.} Moving forward, it is instructive to represent this instance through its standard integer programming formulation:
\begin{equation} \label{eqn:IP-formulation} \tag{IP}
\begin{array}{lll}
    \max & {\displaystyle \sum_{i \in [n]} \sum_{k \in [K-1]: i \in L_{k+1}} q_{ik} x_{ik}} \\
    \text{s.t.} & {\displaystyle \sum_{k \in [K-1] : i \in L_{k+1}} x_{ik} \leq 1} & \forall \, i \in [n] \\
    & {\displaystyle \sum_{i \in L_{k+1}} w_i x_{ik} \leq \mycap( {\cal B}_k )} \qquad \qquad & \forall \, k \in [K-1] \\
    & x_{ik} \in \{ 0, 1 \} & \forall \, k \in [K-1], \, i \in L_{k+1}
    \end{array}
\end{equation}
In this formulation, each decision variable $x_{ik}$ indicates whether item $i$ is assigned to bucket ${\cal B}_k$. The first constraint guarantees that every item is assigned to at most one bucket, and the second constraint ensures that the total weight of the items assigned to each bucket fits within its capacity. The next lemma shows that any feasible assignment can be efficiently mapped to a permutation for our sequencing formulation that collects at least as much profit; the proof is provided in Appendix~\ref{app:proof_lem_ip-to-perm}. 

\begin{lemma}\label{lem:ip-to-perm}
Any feasible solution $x$ to~\eqref{eqn:IP-formulation} can be translated in $O(nK)$ time to a permutation $\pi_x : [n] \to [n]$ satisfying $\newobj( \pi_x ) \geq \sum_{i \in [n]} \sum_{k \in [K-1]: i \in L_{k+1}} q_{ik} x_{ik}$.
\end{lemma}

\paragraph{LP-relaxation and lower bound.} The linear relaxation of this integer program, (LP), is obtained by replacing the integrality constraints $x_{ik} \in \{ 0, 1 \}$ with non-negativity constraints, $x_{ik} \geq 0$. The next lemma relates the resulting fractional optimum to $\newobj_{ \mylight } ( \pi^* )$.

\begin{lemma} \label{lem:lb_frac_opt}
$\opt( \mathrm{LP} ) \geq (1 - 5\eps) \cdot \newobj_{ \mylight } ( \pi^* )$.
\end{lemma}
\begin{proof}
In order to derive the desired bound, we prove that (LP) has a feasible fractional solution $x$ with an objective value of at least $(1 - 5\eps) \cdot \newobj_{ \mylight } ( \pi^* )$. To this end, let $C_k^* = \{ i \in L_k : C_{ \pi^* }( i ) \in {\cal I}_k \}$ be the subset of $k$-light items whose completion time with respect to the optimal permutation $\pi^*$ falls in ${\cal I}_k$. With this notation, recall that
\begin{equation} \label{eqn:lb_frac_opt_light}
\newobj_{ \mylight } ( \pi^* ) = \sum_{k \in [K]_0} \sum_{ i \in C_k^* }      \varphi_{ \pi^* }( i ) = \sum_{k = 2}^K \sum_{ i \in C_k^* } \varphi_{ \pi^* }( i ) \ ,
\end{equation}
where the second equality follows by observing that completion times cannot fall in either of the intervals ${\cal I}_0$ and ${\cal I}_1$, since their union is $[0,2 + \eps]$ whereas $w_{\min} = 3$, by our initial assumption in Section~\ref{subsec:decomp_outline}.

We define a fractional solution $x$ to (LP) by setting $x_{i,k-1} = 1 - 5\eps$ for every $2 \leq k \leq K$ and $i \in C_k^*$; all other variables take zero values. To verify the feasibility of this solution, note that we clearly have $\sum_{k \in [K-1]: i \in L_{k+1}} x_{ik} \leq 1$ for every item $i \in [n]$. As for the capacity constraints, for every $2 \leq k \leq K$,
\begin{eqnarray*}
\sum_{i \in [n]} w_i x_{i,k-1} & = & (1 - 5 \eps) \cdot \sum_{i \in C_k^*} w_i \\
& \leq & (1 - 5 \eps) \cdot \left( (1 + \eps)^k - (1 + \eps)^{ k-1 } + \eps^2 \cdot (1 + \eps)^k \right) \\
& \leq & (1 - 5 \eps) \cdot (1 + 5 \eps) \cdot \left( (1 + \eps)^{k-1} - (1 + \eps)^{ k-2 } \right) \\
& \leq & \mycap \left( {\cal B}_{k-1} \right) \ .
\end{eqnarray*}
Here, the first inequality holds since all items in $C_k^*$ have completion times in ${\cal I}_k$, implying that their total weight is upper bounded by the length $(1 + \eps)^k - (1 + \eps)^{ k-1 }$ of this interval plus the maximum weight of any item in $C_k^*$, which is at most $\eps^2 \cdot (1 + \eps)^k$ due to being $k$-light. The second inequality can easily be verified to hold for every $\eps \in (0,1)$. 

Consequently, the fractional optimum can be lower-bounded by the objective function of $x$, to obtain 
\begin{eqnarray*}
\opt( \mathrm{LP} ) & \geq & \sum_{i \in [n]} \sum_{k \in [K-1]: i \in L_{k+1}} q_{ik} x_{ik} \\
& = & (1 - 5 \eps) \cdot \sum_{k = 2}^K \sum_{ i \in C_k^* } q_{i,k-1} \\
& \geq & (1 - 5 \eps) \cdot \sum_{k = 2}^K \sum_{ i \in C_k^* } \varphi_{ \pi^* }( i ) \\
& = & (1 - 5 \eps) \cdot \newobj_{ \mylight } ( \pi^* ) \ ,
\end{eqnarray*}
where the last equality is precisely~\eqref{eqn:lb_frac_opt_light}. To understand the second inequality, note that for every item $i \in C_k^*$,
\begin{eqnarray*}
\varphi_{ \pi^* }( i ) & = & \max \left\{ p_{i,t} : t \in [T+1] \text{ and } W_t \geq C_{ \pi^* }( i ) \right\} \\
& \leq & \max \left\{ p_{i,t} : t \in [T+1] \text{ and } W_t \geq (1 + \eps)^{k-1} \right\} \\
& = & q_{i,k-1} \ ,
\end{eqnarray*}
where the above inequality holds since $C_{ \pi^* }( i ) \geq (1 + \eps)^{k-1}$.
\end{proof}

\subsubsection{Employing the Shmoys-Tardos algorithm}

\paragraph{The rounding algorithm.} We proceed by utilizing the LP-rounding approach of \citet[Sec.~2]{ShmoysT93}, which was originally proposed for the minimum generalized assignment problem. Specifically, given an optimal fractional solution to the linear program (LP), their algorithm computes an integral vector $\hat{x}$ that satisfies the following properties:
\begin{enumerate}
\item {\em Objective value}: $\hat{x}$ has a super-optimal objective value, i.e.,
    \begin{equation} \label{eqn:prop_ST_super_opt}
    \sum_{i \in [n]} \sum_{k \in [K-1]: i \in L_{k+1}} q_{ik} \hat{x}_{ik} \geq \opt( \mathrm{LP} ) \ .
    \end{equation}

\item {\em Item assignment}: $\hat{x}$ assigns each item to at most one bucket, namely, $ \sum_{k \in [K-1]: i \in L_{k+1}} \hat{x}_{ik} \leq 1$ for every $i \in [n]$.

\item {\em Fixable capacity}: For every bucket ${\cal B}_k$, if its capacity is violated (i.e., $\sum_{i \in L_{k+1}} w_i \hat{x}_{ik} > \mycap( {\cal B}_k )$), there exists a single infeasibility item $i_{\inf(k)}$ with $\hat{x}_{i_{\inf(k)},k} = 1$ whose removal restores the feasibility of that bucket, i.e.,
\begin{equation} \label{eqn:prop_ST_fix}
\sum_{i \in L_{k+1}} w_i \hat{x}_{ik} -w_{ i_{\inf(k)} } \leq  \mycap( {\cal B}_k ) \ .
\end{equation}
\end{enumerate}

\paragraph{Restoring feasibility with negligible profit loss.} Given the above-mentioned properties, a feasible integral solution can obviously be obtained by eliminating the infeasibility item of each bucket with violated capacity. However, this straightforward approach may decrease the objective value by a non-$\eps$-bounded factor. Instead, the final step of our algorithm greedily defines an integral solution $\hat{x}^- \leq \hat{x}$ which is feasible for~\eqref{eqn:IP-formulation} and has an objective value of at least $(1 - 8\eps) \cdot \opt( \mathrm{LP} )$. To this end, for every bucket ${\cal B}_k$ whose capacity is not violated by $\hat{x}$, we simply have $\hat{x}^-_{ik} = \hat{x}_{ik}$ for all $i \in L_{k+1}$. In contrast, for every bucket ${\cal B}_k$ whose capacity is violated, we proceed as follows:
\begin{itemize}
\item Let $i_1, \ldots, i_M$ be an indexing of the set $\{ i \in L_{k+1} : \hat{x}_{ik} = 1 \}$ such that $\frac{ q_{i_1,k} }{ w_{i_1} } \geq \cdots \geq \frac{ q_{i_M,k} }{ w_{i_M} }$.

\item Let $\mu$ be the maximal index for which $\sum_{m \in [\mu]} w_{i_m} \leq \mycap( {\cal B}_k )$.

\item Then, our solution sets $\hat{x}^-_{i_1,k} = \cdots = \hat{x}^-_{i_{\mu},k} = 1$ and $\hat{x}^-_{ik} = 0$ for any other item. Clearly, $\hat{x}^-_{ik} \leq \hat{x}_{ik}$ for all $i \in L_{k+1}$.
\end{itemize}
The next claim shows that the profit collected by $\hat{x}^-$ nearly matches the fractional optimum.

\begin{lemma} \label{lem:x_hatminus_vs_opt_lp}
$\sum_{i \in [n]} \sum_{k \in [K-1]: i \in L_{k+1}} q_{ik} \hat{x}_{ik}^- \geq (1 - 8\eps) \cdot \opt( \mathrm{LP} )$.
\end{lemma}
\begin{proof}
Recall that the super-optimality property of $\hat{x}$, as stated in~\eqref{eqn:prop_ST_super_opt}, corresponds to having $\sum_{i \in [n]} \sum_{k \in [K-1]: i \in L_{k+1}} q_{ik} \hat{x}_{ik} \geq \opt( \mathrm{LP} )$. Therefore, by changing the order of summation, we can establish the desired claim by proving that $\sum_{i \in L_{k+1}} q_{ik} \hat{x}^-_{ik} \geq (1 - 8\eps) \cdot \sum_{i \in L_{k+1}} q_{ik} \hat{x}_{ik}$ for every $k \in [K-1]$. Moreover, since one has $\hat{x}^-_{\cdot k} = \hat{x}_{\cdot k}$ with respect to buckets whose capacity is not violated by $\hat{x}$, it remains to focus on violated buckets.

For such buckets, we first observe that, by the maximality of $\mu$, 
\begin{equation} \label{eqn:bucket_profit_lb}
\sum_{m \in [\mu]} w_{i_m} > \mycap \left( {\cal B}_k \right) - w_{i_{\mu+1}} \geq (1 - 4\eps) \cdot \mycap( {\cal B}_k ) \ ,
\end{equation}
where the second inequality holds since $i_{\mu+1} \in L_{k+1}$, and therefore $w_{i_{\mu+1}} \leq \eps^2 \cdot (1 + \eps)^{k+1} \leq 4\eps \cdot ( (1 + \eps)^k - (1 + \eps)^{ k-1 } ) = 4\eps \cdot \mycap( {\cal B}_k )$ for $\eps \in (0,1)$. On the other hand,
\begin{eqnarray}
\sum_{m \in [M]} w_{i_m} & = & \sum_{i \in L_{k+1}} w_i \hat{x}_{ik} \nonumber \\
& \leq & \mycap( {\cal B}_k ) + w_{ i_{\inf(k)} } \nonumber \\
& \leq & (1 + 4\eps) \cdot \mycap( {\cal B}_k ) \ , \label{eqn:bucket_profit_ub}
\end{eqnarray} 
where the equality follows from how the indices $i_1, \dots, i_M$ were defined, the first inequality is precisely the fixable capacity property of $\hat{x}$ (see~\eqref{eqn:prop_ST_fix}), and the second inequality holds since $w_{ i_{\inf(k)} } \leq 4\eps \cdot \mycap( {\cal B}_k )$, as explained earlier for $w_{i_{\mu+1}}$. Consequently,
\begin{eqnarray*}
\sum_{i \in L_{k+1}} q_{ik} \hat{x}^-_{ik} & = & \sum_{m \in [\mu]} q_{i_m,k} \\
& \geq & \frac{ \sum_{m \in [\mu]} w_{i_m} }{ \sum_{m \in [M]} w_{i_m} } \cdot \sum_{m \in [M]} q_{i_m,k} \\
& \geq & \frac{ 1 - 4\eps }{ 1 + 4\eps } \cdot \sum_{m \in [M]} q_{i_m,k} \\
& \geq & (1 - 8\eps) \cdot \sum_{i \in L_{k+1}} q_{ik} \hat{x}_{ik} \ ,
\end{eqnarray*}
where the first inequality holds since $\frac{ q_{i_1,k} }{ w_{i_1} } \geq \cdots \geq \frac{ q_{i_M,k} }{ w_{i_M} }$, and the second inequality is obtained by plugging in~\eqref{eqn:bucket_profit_lb} and~\eqref{eqn:bucket_profit_ub}.
\end{proof}

\paragraph{Performance guarantee.} We conclude by noting that, since $\hat{x}^-$ is a feasible solution to~\eqref{eqn:IP-formulation}, Lemma~\ref{lem:ip-to-perm} allows us to construct a permutation $\pi_{\mylight}$ with an overall profit of
\begin{eqnarray*}
\newobj( \pi_{ \mylight } ) & \geq & 
\sum_{i \in [n]} \sum_{k \in [K-1]: i \in L_{k+1}} q_{ik} \hat{x}^-_{ik} \\
& \geq & (1 - 8\eps) \cdot \opt( \mathrm{LP} ) \\
& \geq & (1 - 13\eps) \cdot \newobj_{ \mylight } ( \pi^* ) \ ,
\end{eqnarray*}
where the second and third inequalities follow from Lemmas~\ref{lem:x_hatminus_vs_opt_lp} and~\ref{lem:lb_frac_opt}, respectively.

From a running time perspective, the computational bottleneck of our approach is the Shmoys-Tardos algorithm \citeyearpar{ShmoysT93}. As the latter is applied to a maximum generalized assignment instance consisting of $n$ items and $O(K) = O( \frac{ |{\cal I}| }{ \eps } )$ buckets, it requires $O( ( \frac{ |{\cal I}| }{ \eps } )^{O(1)} )$ time in total. Beyond that, restoring the feasibility of $\hat{x}$ and translating the resulting solution $\hat{x}^-$ back to a permutation can both be implemented in $O((nK)^{O(1)})$ time.

\section{QPTAS for Bounded Weight Ratio}\label{sec:qptas-one}

In this section, we develop an approximation scheme for the generalized incremental knapsack problem by embedding our LP-based approach for competing against light contributions within a self-improving algorithm. As formally stated in Theorem~\ref{thm:qptas-bounded} below,
the running time of this algorithm will be exponentially-dependent on $\log (n \cdot \frac{w_{\max}}{w_{\min}})$, meaning that it provides a quasi-polynomial time approximation scheme (QPTAS) when the ratio between the extremal item weights is polynomial in the input size. In Section~\ref{sec:qptas-two}, these ideas will be exploited within an approximate dynamic programming framework to derive a true QPTAS, without making any assumptions on the ratio $\frac{w_{\max}}{w_{\min}}$.

\begin{theorem} \label{thm:qptas-bounded}
For any accuracy level $\eps \in (0,1)$, the generalized incremental knapsack problem can be approximated within a factor of $1 - \eps$ in time $O((nT)^{O(\frac{1}{\eps^{5}  } \cdot \log(n \cdot \frac{w_{\max}}{w_{\min}}))} \cdot |{\cal I}|^{O(1)})$. 
\end{theorem}

\paragraph{Outline.} As an instructive step, we dedicate Section~\ref{subsec:qptas_prelim} to explaining how, given any feasible chain, one can define a residual instance on the remaining (non-inserted) items. In this context, we establish a number of structural properties that relate between the solution spaces of the original and residual instances, which will be useful moving forward. As explained in Section~\ref{subsec:self-improve}, the basic idea behind our ``self-improving'' algorithm resides in arguing that, given a black-box $\alpha$-approximation for the generalized incremental knapsack problem, efficient guessing methods can be utilized to construct a solution that optimally competes against heavy contributions, and simultaneously, $\alpha$-competes against light contributions. In Section~\ref{subsec:recurse_self_improve}, we combine this result with our near-optimal algorithm for light contributions and attain a performance guarantee of $\frac{1}{2-\alpha}$, up to lower-order terms. Repeated applications of these $\alpha \mapsto \frac{1}{2-\alpha}$ improvements will be shown to obtain a $(1-\eps)$-fraction of the optimal profit within $O( \frac{ 1 }{ \eps })$ rounds. It is important to mention that each such application by itself incurs an exponential dependency on $\log(n \cdot \frac{w_{\max}}{w_{\min}})$, meaning that the results of this section are incomparable to those stated in Theorem~\ref{thm:half-approx}, where the running time involved is truly polynomial for any fixed $\eps > 0$. 

\subsection{Residual instances and their properties} \label{subsec:qptas_prelim}

\paragraph{Instance representation.} Due to working with modified instances in subsequent sections, we will designate the underlying set of items in a given instance by ${\cal N}$. As before, each item $i \in {\cal N}$ is associated with a weight of $w_i$, each time period $t \in [T]$ has a capacity of $W_t$, and we gain a profit of $p_{it}$ for introducing item $i$ in period $t$. That said, what differentiates between one instance and the other are two ingredients: The item set ${\cal N}$ and the time period capacities $W = (W_1, \ldots, W_T)$ with respect to which these instances are defined. It is important to point out that, regardless of the instance being considered, the item weights $w_i$, the number of time periods $T$, and the item-to-period profits $p_{it}$ will be kept unchanged. For these reasons, we denote a generalized incremental knapsack instance simply by ${\cal I}=({\cal N},W)$.

\paragraph{The $\boldsymbol{|_{G}}$-operator.} In the following, we introduce additional definitions, notation, and structural properties related to modified instances and their solution space. For a pair of chains, ${\cal S}=(S_1,\dots,S_T)$ and ${\cal G}=(G_1,\dots,G_T)$, we define the \emph{union of ${\cal S}$ and ${\cal G}$} as ${\cal S} \cup {\cal G} = (S_1\cup G_1,\dots,S_T\cup G_T)$, which is clearly a chain itself. For a chain ${\cal S}$ and a subset of items $G \subseteq {\cal N}$, we denote by ${\cal S} |_{G}$ the restriction of ${\cal S}$ to $G$, namely, ${\cal S} |_{G} = (S_1 \cap G, \ldots, S_T \cap G)$; one can easily verify that ${\cal S} |_{G}$ is a chain as well. The next claim, whose straightforward proof is omitted, establishes the feasibility of ${\cal S} |_{G}$ whenever ${\cal S}$ is feasible.

\begin{observation}\label{obs:subset}
Let ${\cal S}$ be a feasible chain for ${\cal I}$. Then, for any set of items $G \subseteq {\cal N}$, the chain ${\cal S} |_{G}$ is feasible as well.  
\end{observation}

\paragraph{The residual instance.} Given a feasible chain ${\cal G}=(G_1,\dots,G_T)$ for an instance ${\cal I}=({\cal N},W)$, we define the \emph{residual} generalized incremental knapsack instance ${\cal I}^{-{\cal G}}=({\cal N}^{-{\cal G}}, W^{-{\cal G}})$ as follows: 
\begin{itemize}
    \item The new set of items is ${\cal N}^{-{\cal G}}={\cal N}\setminus G_T$. Namely, we eliminate all items that were introduced at any point in time by ${\cal G}$.
    
    \item The residual capacity of every time $t \in [T]$ is set to $W^{-{\cal G}}_t = \min_{t \leq \tau \leq T} (W_{\tau} -w({G}_{\tau}))$.
    
    \item As previously mentioned, all item weights and profits remain unchanged. 
\end{itemize}
To verify that the residual instance ${\cal I}^{-{\cal G}}$ is well defined, it suffices to show that the residual capacities $W^{-{\cal G}}$ are non-negative and non-decreasing over time. The former property holds since $w(G_t) \leq W_t$ for every $t \in [T]$, by feasibility of ${\cal G}$. The latter property follows by observing that 
\[ W^{-{\cal G}}_t = \min_{t \leq \tau \leq T} (W_{\tau} -w({G}_{\tau}) ) \leq \min_{t+1 \leq \tau \leq T} (W_{\tau} -w({G}_{\tau}) ) = W^{-{\cal G}}_{t+1} \ . \]

The next two claims, whose respective proofs appear in  Appendices~\ref{app:proof_lem_union-chain} and~\ref{app:proof_lem_residual_chain}, explain the relationship between the solution spaces of the original instance ${\cal I}$ and its residual instance ${\cal I}^{-{\cal G}}$. For our purposes, the main implication of this relationship will be that, whenever we are able to ``guess'' a chain ${\cal G} = {\cal S}^* |_{G}$, where ${\cal S}^*$ is an optimal chain for ${\cal I}$, it suffices to focus on solving the residual instance ${\cal I}^{-{\cal G}}$. With an appropriate guess for the set of items $G$, this property will be a key idea within the approximation scheme we devise in the remainder of this section.

\begin{lemma}\label{lem:union-chain}
Let ${\cal G}$ be a feasible chain for ${\cal I}$ and let ${\cal R}$ be a feasible chain for ${\cal I}^{-{\cal G}}$. Then, ${\cal G} \cup {\cal R}$ is a feasible chain for ${\cal I}$ with profit $\objfunc({\cal G} \cup {\cal R}) = \objfunc({\cal G})+\objfunc({\cal R})$. 
\end{lemma}

\begin{lemma}\label{lem:residual-chain}
Let ${\cal S}$ be a feasible chain for ${\cal I}$ and let ${\cal G} = {\cal S} |_G$, for some set of items $G \subseteq {\cal N}$. Then, ${\cal S} |_{{\cal N } \setminus G}$ is a feasible chain for ${\cal I}^{-{\cal G}}$ with profit $\objfunc({\cal S} |_{{\cal N } \setminus G}) = \objfunc({\cal S}) - \objfunc({\cal G})$. Moreover, if ${\cal S}$ is optimal for ${\cal I}$, then ${\cal S} |_{{\cal N } \setminus G}$ is optimal for ${\cal I}^{-{\cal G}}$.
\end{lemma}

\subsection{The boosting algorithm} \label{subsec:self-improve}

Given a generalized incremental knapsack instance ${\cal I}=({\cal N},W)$, let us focus our attention on a fixed optimal chain ${\cal S}^*$. As argued in Lemma~\ref{lem:reformulation}, this chain can be mapped to a permutation $\pi_{{\cal S}^*} : {\cal N} \to [|{\cal N}|]$ whose objective value with respect to the corresponding sequencing formulation is $\newobj(\pi_{{\cal S}^*}) \geq \objfunc({\cal S}^*)$. By decomposing the overall profit $\newobj(\pi_{{\cal S}^*})$ into heavy and light contributions, as prescribed by Equation~\eqref{eq:profit_decomp}, we have:
\begin{equation}\label{eq:opt-profit-decomp} \newobj( \pi_{{\cal S}^*} ) = \underbrace{ \sum_{k \in [K]_0} \sum_{ \MyAbove{ i \in H_k : }{ C_{ \pi_{{\cal S}^*} }( i ) \in {\cal I}_k } } \varphi_{ \pi_{{\cal S}^*} }( i ) }_{ \newobj_{ \myheavy } ( \pi_{{\cal S}^*} ) } + \underbrace{ \sum_{k \in [K]_0} \sum_{ \MyAbove{ i \in L_k : }{ C_{ \pi_{{\cal S}^*} }( i ) \in {\cal I}_k } } \varphi_{ \pi_{{\cal S}^*} }( i ) }_{ \newobj_{ \mylight } ( \pi_{{\cal S}^*} ) } \ . 
\end{equation}
As a side note, similarly to Section~\ref{sec:2-approx}, we assume without loss of generality that $w_{\min} \geq 3$. Given these quantities, for $\alpha_H, \alpha_L \in [0,1]$, we say that an algorithm ${\cal A}$ guarantees an $(\alpha_H, \alpha_L)$-approximation with respect to ${\cal S}^*$ when it computes a feasible chain ${\cal S}$ with  $\objfunc({\cal S}) \geq \alpha_H \cdot \newobj_{ \myheavy } ( \pi_{{\cal S}^*} ) + \alpha_L \cdot \newobj_{ \mylight } ( \pi_{{\cal S}^*} )$. We mention in passing that this definition depends on the specific permutation $\pi_{{\cal S}^*}$, and is generally different from the standard notion of an $\alpha$-approximation, where the chain ${\cal S}$ is required to satisfy $\objfunc({\cal S}) \geq \alpha \cdot \objfunc({\cal S}^*)$.

\paragraph{From $\boldsymbol{\alpha}$-approximation to $\boldsymbol{(1,\alpha)}$-approximation.} In what follows, we show how to boost the profit performance of any approximation algorithm for the generalized incremental knapsack problem. For every $\alpha \in [0,1]$, we explain how to combine a black-box $\alpha$-approximation with further guesses for the positioning of heavy items with respect to the permutation $\pi_{{\cal S}^*}$ in order to derive a $(1, \alpha)$-approximation, incurring an extra multiplicative factor of $O((nT)^{ O(\frac{ 1 }{ \eps^2 } \log(n \rho))})$ in running time, where $\rho = \frac{w_{\max}}{w_{\min}}$. This result can be formally stated as follows.

\begin{lemma}\label{lem:from-alpha-to-one,alpha}
Suppose that the algorithm ${\cal A}$ constitutes an $\alpha$-approximation for generalized incremental knapsack, for some $\alpha \in [0,1]$. Then, there exists a $(1, \alpha)$-approximation whose running time is $O((nT)^{ O(\frac{ 1 }{ \eps^2 } \log(n \rho))} \cdot \mytime_{\cal A}( n,T ))$. Here, $\mytime_{\cal A}( n,T )$ designates the worst-case running time of ${\cal A}$ for instances with $n$ items and $T$ time periods. 
\end{lemma}

\paragraph{Preliminaries.} We remind the reader that Section~\ref{subsec:decomp_outline} has previously defined the intervals ${\cal I}_0 = [0,1 )$ and ${\cal I}_k = ((1+\eps)^{k-1}, (1+ \eps)^k]$ for $k \in [K]$, where $K= \lceil \log_{1+\eps} (\sum_{i \in [n]} w_i) \rceil$. In this regard, an item $i$ is $k$-heavy when $w_i \geq \eps^2 \cdot (1+\eps)^k$, with the convention that $H_k$ stands for the collection of $k$-heavy items. Let $G^{*\myheavy}$ be the set of items that are heavy for the interval that contains their completion time with respect to the permutation $\pi_{{\cal S}^*}$, i.e., $G^{*\myheavy} = \bigcup_{k \in [K]_0} \{ i \in H_k : C_{\pi_{{\cal S}^*}}(i) \in {\cal I}_k \}$. The following lemma, whose proof appears in Appendix~\ref{app:proof_lem_heavy_bound}, provides an upper bound on the cardinality of this set. 

\begin{lemma} \label{lem:heavy_bound}
$|G^{*\myheavy}| \leq \frac{ 3 \log(n \rho) }{\eps^2}$.
\end{lemma}

We proceed by considering the restriction of the optimal chain ${\cal S}^*$ to the set of items $G^{*\myheavy}$, which will be denoted by ${\cal H}^* = {\cal S}^* |_{ G^{*\myheavy} }$. By Observation~\ref{obs:subset}, ${\cal H}^*$ is a feasible chain for ${\cal I}$. The next lemma, whose proof can be found in Appendix~\ref{app:proof_lem_heavy_profit}, relates between the profit of this chain and heavy contributions with respect to the permutation $\pi_{{\cal S}^*}$.

\begin{lemma}\label{lem:heavy_profit}
$\objfunc({\cal H}^*) = \newobj_{ \myheavy } ( \pi_{{\cal S}^*} )$.
\end{lemma}

\paragraph{The algorithm.} At a high level, our algorithm relies on ``knowing'' the restricted chain ${\cal H}^*$ in advance, which will be justified by guessing all items in $G^{*\myheavy}$ and their insertion times with respect to the optimal chain ${\cal S}^*$. This procedure will be implemented by enumerating over all possible configurations of these parameters. For each such guess, we construct the residual generalized incremental knapsack instance, to which the $\alpha$-approximation algorithm ${\cal A}$ is applied. Formally, given an instance ${\cal I}=({\cal N},W)$ and an error parameter $\eps > 0$, we proceed as follows:
\begin{enumerate}
    \item For every feasible chain ${\cal G} = (G_1, \ldots, G_T)$ with $|G_T| \leq \frac{3 \log (n \rho) }{\eps^2}$:
    \begin{enumerate}
        \item Construct the residual instance ${\cal I}^{-{\cal G}}$. 
        
        \item Apply the algorithm ${\cal A}$ to obtain an $\alpha$-approximate feasible chain ${\cal S}^{- {\cal G}} = (S^{- {\cal G}}_1, \dots, S^{- {\cal G}}_T)$ for ${\cal I}^{-{\cal G}}$.
    \end{enumerate}
    
    \item Return the chain ${\cal G}^* \cup {\cal S}^{-{\cal G}^*}$ of maximum profit among those considered above. 
\end{enumerate}

\paragraph{Analysis: Feasibility and running time.} We first observe that, for any feasible chain ${\cal G}$ constructed in step~1, since ${\cal S}^{-{\cal G}}$ is a feasible chain for ${\cal I}^{-{\cal G}}$, the  feasibility of ${\cal G}\cup  {\cal S}^{-{\cal G}}$ for ${\cal I}$ follows by Lemma~\ref{lem:union-chain}. In terms of running time, we are considering only chains that introduce at most $\frac{ 3 \log(n \rho) }{\eps^2}$ items over all time periods. Thus, the number of chains being enumerated is $O((nT)^{ O(\frac{ 1 }{ \eps^2 } \log(n \rho))})$. For each residual instance, consisting of $T$ time periods and at most $n$ items, we apply the algorithm ${\cal A}$ once, implying that the overall running time is indeed $O((nT)^{ O(\frac{ 1 }{ \eps^2 } \log(n \rho))} \cdot \mytime_{\cal A}( n,T ))$. 

\paragraph{Analysis: $\boldsymbol{(1, \alpha)}$-approximation guarantee.} We conclude the proof of Lemma~\ref{lem:from-alpha-to-one,alpha} by arguing that ${\cal G}^* \cup {\cal S}^{-{\cal G}^*}$ is a $(1,\alpha)$-approximate chain with respect to ${\cal S}^*$ for the original instance ${\cal I}$.

\begin{lemma}
$\objfunc({\cal G}^* \cup  {\cal S}^{-{\cal G}^*}) \geq \newobj_{ \myheavy } ( \pi_{{\cal S}^*} ) + \alpha \cdot \newobj_{ \mylight } ( \pi_{{\cal S}^*} )$.
\end{lemma}
\begin{proof}
We begin by observing that the feasible chain ${\cal H}^* = {\cal S}^* |_{ G^{*\myheavy} }$ is one of those considered in step~1. To verify this claim, note that $|G^{*\myheavy}| \leq \frac{ 3 \log(n \rho) }{\eps^2}$ by Lemma~\ref{lem:heavy_bound}, meaning that ${\cal H}^*$ introduces at most that many items across all time periods. As a result, since the chain ${\cal G}^* \cup {\cal S}^{-{\cal G}^*}$ attains a maximum profit among those considered, we have $\objfunc({\cal G}^* \cup  {\cal S}^{-{\cal G}^*}) \geq \objfunc({\cal H}^* \cup  {\cal S}^{-{\cal H}^*})$, and it remains to prove that $\objfunc({\cal H}^* \cup  {\cal S}^{-{\cal H}^*}) \geq \newobj_{ \myheavy } ( \pi_{{\cal S}^*} ) + \alpha \cdot \newobj_{ \mylight } ( \pi_{{\cal S}^*} )$.

For this purpose, let ${\cal L}^* = {\cal S}^* |_{{\cal N} \setminus G^{*\myheavy}}$ be the restriction of ${\cal S}^*$ to the set ${\cal N} \setminus G^{*\myheavy}$, which is a feasible chain for ${\cal I}$ by Observation~\ref{obs:subset}. We next show that $\objfunc({\cal L}^*) = \newobj_{ \mylight } ( \pi_{{\cal S}^*} )$. In order to derive this claim, note that since $L^*_T$ and $H^*_T$ are disjoint and ${\cal S}^* = {\cal H}^* \cup {\cal L}^*$, it follows that
\begin{eqnarray*}
\objfunc({\cal L}^*) & = & \objfunc({\cal S}^*) - \objfunc({\cal H}^*) \\
& = & \objfunc({\cal S}^*) - \newobj_{ \myheavy } ( \pi_{{\cal S}^*} ) \\
& = & \newobj(\pi_{{\cal S}^*}) - \newobj_{ \myheavy } ( \pi_{{\cal S}^*} ) \\
& = & \newobj_{ \mylight } ( \pi_{{\cal S}^*} ) \ ,
\end{eqnarray*}
where the second equality holds due to Lemma~\ref{lem:heavy_profit}, the third equality is obtained by recalling that $\newobj(\pi_{{\cal S}^*}) = \objfunc({\cal S}^*)$, as shown along the proof of Lemma~\ref{lem:heavy_profit}, and the last equality follows from the profit decomposition~\eqref{eq:opt-profit-decomp}. 

However, the crucial observation is that ${\cal L}^*$ is a feasible chain for the residual instance ${\cal I}^{-{\cal H}^*}$, by Lemma~\ref{lem:residual-chain}. Consequently, since the algorithm ${\cal A}$ computes an $\alpha$-approximate feasible chain ${\cal S}^{-{\cal H}^*}$ for the latter instance, $\objfunc({\cal S}^ {-{\cal H}^*}) \geq \alpha \cdot \objfunc({\cal L}^*) = \alpha \cdot \newobj_{ \mylight } ( \pi_{{\cal S}^*} )$, implying that ${\cal H}^* \cup  {\cal S}^{-{\cal H}^*}$ indeed has a profit of $\objfunc({\cal H}^* \cup {\cal S}^{-{\cal H}^*}) =  \objfunc({\cal H}^*) + \objfunc({\cal S}^{- {\cal H}^*}) \geq \newobj_{ \myheavy } ( \pi_{{\cal S}^*} ) + \alpha \cdot \newobj_{ \mylight } ( \pi_{{\cal S}^*} )$.
\end{proof}

\subsection{The ratio improvement and final algorithm} \label{subsec:recurse_self_improve}

We proceed by revealing the self-improving feature of our approach, by showing that a $(1, \alpha)$-approximation for generalized incremental knapsack leads in turn to a $\frac{1-\delta}{2-\alpha}$-approximation, when combined with our algorithm for light items, presented in Section~\ref{sec:light}. We will then show how to recursively apply this self-improving idea to eventually derive an approximation scheme.

\begin{lemma}\label{lem:one-step}
Suppose that, for some $\alpha \in [0,1]$, the algorithm ${\cal A}$ constitutes an $\alpha$-approximation. Then, for any accuracy level $\delta > 0$, the generalized incremental knapsack problem can be approximated within factor $\frac{1-\delta}{2-\alpha}$ in time $O( (nT)^{ O(\frac{ 1 }{ \delta^2 } \log(n \rho))} \cdot \mytime_{\cal A}( n,T ) + ( \frac{ |{\cal I}| }{ \eps } )^{O(1)} )$. 
\end{lemma}
\begin{proof}
As explained in Section~\ref{subsec:self-improve}, the optimal chain ${\cal S}^*$ can be mapped to a permutation $\pi_{{\cal S}^*}$ whose overall profit $\newobj(\pi_{{\cal S}^*})$ decomposes into heavy and light contributions, $\newobj(\pi_{{\cal S}^*} ) = \newobj_{ \myheavy } ( \pi_{{\cal S}^*} )+ \newobj_{ \mylight } ( \pi_{{\cal S}^*} )$. Now, on the one hand, Lemma~\ref{lem:from-alpha-to-one,alpha} provides us with a $(1, \alpha)$-approximation in $O((nT)^{ O(\frac{ 1 }{ \delta^2 } \log(n \rho))} \cdot \mytime_{\cal A}( n,T ))$ time. That is, we obtain a feasible chain ${\cal S}_{ (1, \alpha) }$ with $\objfunc({\cal S}_{ (1, \alpha) }) \geq\newobj_{ \myheavy } ( \pi_{{\cal S}^*} ) + \alpha \cdot \newobj_{ \mylight } ( \pi_{{\cal S}^*} ) $. On the other hand, the main result of Section~\ref{sec:light} allows us to compute in $O( ( \frac{ |{\cal I}| }{ \eps } )^{O(1)} )$ time a permutation $\pi_{ \mylight }$ with a profit of $\newobj( \pi_{ \mylight } ) \geq (1- \delta) \cdot \newobj_{ \mylight } ( \pi_{{\cal S}^*} )$. By converting this permutation to a feasible chain ${\cal S}_{ (0, 1 - \delta) }$ along the lines of Lemma~\ref{lem:reformulation}, we clearly obtain a $(0, 1 - \delta)$-approximation, meaning that $\objfunc({\cal S}_{ (0, 1 - \delta) }) \geq (1 - \delta) \cdot \newobj_{ \mylight } ( \pi_{{\cal S}^*} )$. Our combined approach independently employs both algorithms and returns the more profitable of the two feasible chains computed, ${\cal S}_{ (1, \alpha) }$ and ${\cal S}_{ (0, 1 - \delta) }$, to obtain a profit of 
\begin{eqnarray*}
\max \left\{\objfunc({\cal S}_{ (1, \alpha) }), \objfunc({\cal S}_{ (0, 1 - \delta) }) \right\} & \geq & \max \left\{ \newobj_{ \myheavy } ( \pi_{{\cal S}^*} ) + \alpha \cdot \newobj_{ \mylight } ( \pi_{{\cal S}^*} ), (1- \delta) \cdot \newobj_{ \mylight } ( \pi_{{\cal S}^*} ) \right\} \\ 
& \geq &  \frac{ 1 }{ 2 - \alpha } \cdot \left( \newobj_{ \myheavy } ( \pi_{{\cal S}^*} ) + \alpha \cdot \newobj_{ \mylight } ( \pi_{{\cal S}^*} ) \right) \\
& & \mbox{} + \left( 1 - \frac{ 1 }{ 2 - \alpha } \right) \cdot (1-\delta) \cdot \newobj_{ \mylight } ( \pi_{{\cal S}^*} ) \\
& \geq &   \frac{ 1- \delta }{ 2 - \alpha } \cdot \left( \newobj_{ \myheavy } ( \pi_{{\cal S}^*} ) + \newobj_{ \mylight } ( \pi_{{\cal S}^*} ) \right) \\
& = &   \frac{ 1- \delta }{ 2 - \alpha } \cdot \newobj( \pi_{{\cal S}^*} ) \\
& \geq &   \frac{ 1- \delta }{ 2 - \alpha } \cdot \objfunc( {\cal S}^* ),
\end{eqnarray*}
where the last inequality follows from Lemma~\ref{lem:reformulation}.
\end{proof}

\paragraph{The final approximation scheme.} We conclude by explaining how our $\alpha \mapsto \frac{1-\delta}{2-\alpha}$ improvement, outlined in Lemma~\ref{lem:one-step}, can be iteratively applied to derive an approximation scheme for the generalized incremental knapsack problem, thereby sealing the proof of Theorem~\ref{thm:qptas-bounded}.

For the purpose of ensuring a $(1 - \eps)$-fraction of the optimal profit, we will set the error tolerance $\delta$ in Lemma~\ref{lem:one-step} as a function of $\eps$, where the exact dependency will be determined later on. 
Given this self-improving result, we define a sequence of algorithms ${\cal A}_0, {\cal A}_1, \dots$, with the convention that the approximation ratio of each such algorithm ${\cal A}_r$ is denoted by  $\alpha_r$. Specifically, this sequence begins with the trivial algorithm ${\cal A}_0$ that returns an empty solution $(\emptyset, \ldots, \emptyset)$, meaning that $\alpha_0 = 0$. Then, by applying Lemma~\ref{lem:one-step} with respect to ${\cal A}_0$, we obtain the algorithm ${\cal A}_1$, for which $\alpha_1 = \frac{1- \delta}{2}$. Subsequently, by a similar application with respect to ${\cal A}_1$, we obtain ${\cal A}_2$, with $\alpha_2 = \frac{1-\delta}{2 - \alpha_1}$. In general, for every integer $r \geq 1$, the resulting algorithm ${\cal A}_r$ guarantees an approximation ratio of $\alpha_r = \frac{1 - \delta}{2 - \alpha_{r-1}}$. The next lemma, whose proof is presented in Appendix~\ref{app:proof_lem_alpha_sequence}, provides a closed-form lower bound on $\alpha_r$.

\begin{lemma} \label{lem:alpha_sequence}
$\alpha_r \geq \frac{r}{r+1} - r \delta$, for every $r \geq 0$.
\end{lemma}

By choosing an error tolerance of $\delta = \frac{\eps^2}{2}$, the above lemma implies that $\lceil \frac{2}{\eps} \rceil$ self-improving rounds produce an algorithm ${\cal A}_{ \lceil \frac{2}{\eps} \rceil }$ for computing a feasible chain ${\cal S}$ with a profit of $\objfunc( {\cal S} ) \geq ( \frac{ \lceil 2/\eps \rceil }{ \lceil 2/\eps \rceil + 1 } - \lceil \frac{2}{\eps} \rceil \cdot \frac{\eps^2}{2} ) \cdot \objfunc({\cal S}^*) \geq (1 - \eps) \cdot \objfunc({\cal S}^*)$, thereby deriving the approximation guarantee of Theorem~\ref{thm:qptas-bounded}. Furthermore, it is not difficult to verify that algorithm ${\cal A}_{ \lceil \frac{2}{\eps} \rceil }$ runs in $O((nT)^{O(\frac{1}{\eps^5} \cdot \log(n \rho))} \cdot |{\cal I}|^{O(1)})$ time, by induction on $r$.

\section{QPTAS for General Instances} \label{sec:qptas-two}

Thus far, we have developed an approximation scheme whose running time includes an exponential dependency on $\log(n\cdot \frac{w_{\max}}{w_{\min}})$, leading to a quasi-PTAS for problem instances where the ratio $\frac{w_{\max}}{w_{\min}}$ is polynomial in the input size. In what follows, we show how to obtain a true quasi-PTAS, without mitigating assumptions on $\frac{w_{\max}}{w_{\min}}$. 

\begin{theorem} \label{thm:qptas2}
For any accuracy level $\eps \in (0,1)$, the generalized incremental knapsack problem can be approximated within a factor of $1 - \eps$ in time $O( |{\cal I}|^{O((\frac{1}{\epsilon}\log |{\cal I}|)^{O(1)})} )$. 
\end{theorem}

\subsection{Technical overview}\label{subsec:qptas2-overview}

\paragraph{Step 1: Creating a well-spaced instance.} We begin by slightly altering a given instance ${\cal I}=({\cal N}, W)$, with the objective of creating nearly-ideal circumstances for the approximation scheme of Section~\ref{sec:qptas-one} to operate, losing negligible profits along the way. For this purpose, given an error parameter $\eps > 0$, we say that the instance ${\cal I}$ is well-spaced when its set of items ${\cal N}$ can be partitioned into clusters ${\cal C}_1, \ldots, {\cal C}_M$ satisfying the following properties:
\begin{enumerate}
    \item {\em Weight ratio within clusters}: For every $m \in [M]$, the weights of any two items in cluster ${\cal C}_m$ differ by a multiplicative factor of at most $n^{ 1/\eps }$. 
    
    \item {\em Weight gap between clusters}: For every $m_1, m_2 \in [M]$ with $m_1 < m_2$, the weight of any item in cluster ${\cal C}_{m_2}$ is greater than the weight of any item in cluster ${\cal C}_{m_1}$ by a multiplicative factor of at least $n^{ 1 + (m_2 - m_1 - 1 ) / \eps }$.
\end{enumerate}
In Section~\ref{subsec:well-spaced}, we show that one can efficiently identify a subset of items over which the induced instance is well-spaced, while still admitting a near-optimal solution. We derive this result, as formally stated below, through an application of the shifting method (see, for instance, \citep{hochbaum1985approximation, Baker94}).

\begin{lemma}\label{lem:skip-buckets}
There exists an item set ${\cal N}_{\myspace} \subseteq {\cal N}$ for which ${\cal I}_{\myspace} = ({\cal N}_{\myspace}, W)$ is a well-spaced instance, whose optimal chain ${\cal S}_{\myspace}$ guarantees a profit of $\objfunc({\cal S}_{\myspace}) \geq (1 - \eps) \cdot \objfunc({\cal S}^*)$. Such a set can be determined in $O((n/\eps )^{ O(1) })$ time.
\end{lemma}

\paragraph{Step 2: Proving the sparse-crossing property.} For simplicity of notation, we assume from this point on that the instance ${\cal I}=({\cal N}, W)$ is well-spaced, with clusters ${\cal C}_1, \ldots, {\cal C}_M$. Now suppose that the optimal permutation $\pi^*$ for the sequencing-based formulation of this instance was known to be ``crossing-free'', namely, items belonging to cluster ${\cal C}_1$ appear first in $\pi^*$, followed by  those belonging to cluster ${\cal C}_2$, so on and so forth. In other words, a left-to-right scan of the permutation $\pi^*$ reveals that it is weakly-increasing by cluster. In this ideal situation, the approximation scheme we propose in Section~\ref{sec:qptas-one} can be sequentially employed to the clusters ${\cal C}_1, \ldots, {\cal C}_M$ in increasing order. This way, we would have obtained a $(1-\eps)$-approximation in truly quasi-polynomial time, since the extremal weight ratio within each cluster is $n^{ 1/\eps }$-bounded, by property~1.

Unfortunately, elementary examples show that an optimal permutation $\pi^*$ may not be crossing-free, in the sense that items in any given cluster can be preceded by items belonging to higher-index clusters. That said, a suitable relaxation of these ideas can still be exploited. Formally, let us denote by $\mycross_m( \pi )$ the number of items in clusters ${\cal C}_{m+1}, \ldots, {\cal C}_M$ that appear in the permutation $\pi$ before the last item belonging to cluster ${\cal C}_m$; note that crossing-free is equivalent to having $\mycross_1( \pi ) = \cdots = \mycross_M( \pi ) = 0$. Our next structural result, formally established in Section~\ref{subsec:proof_structure-cross}, proves the existence of a near-optimal permutation with very few items crossing each cluster. 

\begin{lemma}\label{lem:structure-cross}
There exist an item set ${\cal N}_{\mysparse} \subseteq {\cal N}$ and a permutation $\pi_{\mysparse}: {\cal N}_{\mysparse} \to [| {\cal N}_{\mysparse} |]$  satisfying:
\begin{enumerate}
    \item Sparse crossing: $\max_{m \in [M]} \mycross_m( \pi_{\mysparse} ) \leq \frac{\lceil \log_2 M \rceil}{\eps}$.
    
    \item Near-optimal profit: $\newobj( \pi_{\mysparse} ) \geq (1 -  \eps ) \cdot \newobj( \pi^* )$.
\end{enumerate}
\end{lemma}

Technically speaking, our proof is based on applying a sequence of recursive transformations with respect to the unknown optimal permutation $\pi^*$. To convey the high-level idea, let $i_{\mymiddle}$ be the last-appearing item in $\pi^*$ out of clusters ${\cal C}_1, \ldots, {\cal C}_{M/2}$. When fewer than $1/ \eps$ items in clusters ${\cal C}_{(M/2)+1}, \ldots, {\cal C}_M$ appear before $i_{\mymiddle}$, each of the clusters ${\cal C}_1, \ldots, {\cal C}_{M/2}$ has at most $1/ \eps$ crossings due to items in ${\cal C}_{(M/2)+1}, \ldots, {\cal C}_M$. We can therefore recursively proceed into the left part of $\pi^*$, stretching up to the item $i_{\mymiddle}$, and into its right part, consisting of the remaining items. In the opposite case, where at least $1/ \eps$ items in clusters ${\cal C}_{(M/2)+1}, \ldots, {\cal C}_M$ appear before $i_{\mymiddle}$, the important observation is that we can eliminate the cheapest out of the first $1/\eps$ such items while losing only an $O( \eps)$-fraction of their combined profit. However, since this item is heavier than any item in lower-index clusters by a factor of at least $n$ (see property~2), the gap we have just created is sufficiently large to pull back each and every item in clusters ${\cal C}_1, \ldots, {\cal C}_{M/2}$, only increasing their profit contributions. We can now recursively proceed into the left and right parts. 

\paragraph{Step 3: The external dynamic program.}  Given the sparse-crossing property, we dedicate Section~\ref{subsec:DP_general} to proposing a dynamic programming approach for computing a near-optimal  permutation. For this purpose, by recycling some of the notation introduced in Section~\ref{sec:heavy}, our state description $(m, \psi_m, {\cal Q}_{>m})$ will consists of the following parameters: 
\begin{itemize}
\item The index of the current cluster, $m$.

\item The profit requirement, $\psi_m$.

\item The set of items ${\cal Q}_{>m}$ belonging to clusters ${\cal C}_{m+1}, \ldots, {\cal C}_M$ that will be crossing into lower-index clusters, noting that Lemma~\ref{lem:structure-cross} allows us to consider only small sets, of size $O( \frac{ \log M }{  \eps } )$.
\end{itemize}
At a high level, the value function $F(m, \psi_m, {\cal Q}_{>m})$ will represent the minimum makespan $w(S)$ that can be attained, over all subset of items $S$  within the union of ${\cal Q}_{>m}$ and the clusters ${\cal C}_1, \ldots, {\cal C}_m$ (namely, $S \subseteq {\cal Q}_{>m} \uplus (\biguplus_{\mu \in [m]} {\cal C}_{\mu})$)  and over all permutations $\pi : S \to [|S|]$ that generate a total profit of at least $\psi_m$. Clearly, the best-possible profit of a sparse-crossing permutation corresponds to the maximal value $\psi_M$ that satisfies $F(M, \psi_M, \emptyset) < \infty$, which is at least $(1 -  \eps ) \cdot \newobj( \pi^* )$, by Lemma~\ref{lem:structure-cross}. 

As formally explained in Section~\ref{subsec:DP_general}, within the recursive equations for computing $F(m, \psi_m, {\cal Q}_{>m})$, evaluating the marginal makespan increase of each possible action involves solving a single-cluster subproblem. Specifically for the latter, the approximation scheme we have devised in Section~\ref{sec:qptas-one} will be shown to incur a quasi-polynomial running time. In parallel, the dominant factor in determining the underlying number of states emerges from the set of items ${\cal Q}_{>m}$, taking $O( n^{ O( \frac{ 1 }{ \eps } \log M ) })$ possible values, respectively, thus forming the second source of quasi-polynomiality in our approach and concluding the proof of Theorem~\ref{thm:qptas2}. 

\subsection{Proof of Lemma~\ref{lem:skip-buckets}: Creating a well-spaced instance} \label{subsec:well-spaced}

\paragraph{Bucketing.} For the purpose of identifying the desired subset ${\cal N}_{\myspace}$, we initially partition the overall collection of items ${\cal N}$ into buckets ${\cal B}_1, \ldots, {\cal B}_L$ according to their weights. This partition will be geometric, by powers of $n$, meaning that $L = \lceil  \log_n(\frac{w_{\max}}{w_{\min}}) \rceil + 1$. Specifically, the first bucket ${\cal B}_1$ consists of items whose weight resides in $[w_{\min}, n \cdot w_{\min})$, the second bucket ${\cal B}_2$ consists of those with weight in $[n \cdot w_{\min}, n^2 \cdot w_{\min})$, so on and so forth, where in general, bucket ${\cal B}_{\ell}$ corresponds to the interval $[n^{\ell-1} \cdot w_{\min}, n^{\ell} \cdot w_{\min})$. It is easy to verify that ${\cal B}_1, \ldots, {\cal B}_L$ is indeed a partition of ${\cal N}$.

\paragraph{Creating clusters.} Now let $r \in \{ 0, \ldots, \frac{1}{\eps} - 1 \}$ be an integer parameter whose value will be determined later. Accordingly, we create a subset of items ${\cal N}_r \subseteq {\cal N}$, that will be clustered into ${\cal C}_1^r, \ldots, {\cal C}_M^r$ with $M = \lfloor \eps L+2\rfloor$, as follows. Intuitively, we introduce ``gaps'' within the sequence of buckets ${\cal B}_1, \ldots, {\cal B}_L$, spaced apart by $\frac{1}{\eps}$ indices, through eliminating every bucket ${\cal B}_{\ell}$ with $\ell \mod \frac{1}{\eps} = r$; then, between every pair of successive gaps, buckets will be unified to form a single cluster. That is, the first cluster is defined as ${\cal C}_1^r = \biguplus_{\ell = 1}^{ r-1 } {\cal B}_{\ell}$, the second cluster is ${\cal C}_2^r = \biguplus_{\ell = r+1}^{ r - 1 + 1/\eps } {\cal B}_{\ell}$, the third is ${\cal C}_3^r = \biguplus_{\ell = r + 1 + 1/\eps}^{ r - 1 + 2/\eps } {\cal B}_{\ell}$, and so on. Finally, we define the subset of items ${\cal N}_r$ as the union of all clusters, i.e., ${\cal N}_r = \biguplus_{m \in [M]} {\cal C}_m^r$, with a corresponding generalized incremental knapsack instance ${\cal I}_r = ({\cal N}_r, W)$.

\paragraph{Analysis.} In what follows, we argue that for every $r \in \{ 0, \ldots, \frac{1}{\eps} - 1 \}$, the instance ${\cal I}_r$ we have just constructed is in fact well-spaced, and the partition of ${\cal N}_r$ into clusters is given by ${\cal C}_1^r, \ldots, {\cal C}_M^r$. For this purpose, we separately prove each of the required well-spaced properties.
\begin{enumerate}
    \item {\em Weight ratio within clusters}: Consider two items $i_1$ and $i_2$ belonging to the same cluster ${\cal C}_m^r$. Letting ${\cal B}_{\ell_1}$ and ${\cal B}_{\ell_2}$ be the buckets containing these items, respectively, their weight ratio can be upper bounded by observing that 
    \begin{eqnarray*}
    \frac{ w_{i_2} }{ w_{i_1} } & \leq & \frac{ \max_{ i \in {\cal B}_{\ell_2} } w_i }{ \min_{ i \in {\cal B}_{\ell_1} } w_i } \\
    & \leq & n^{\ell_2 - (\ell_1 - 1)} \\
    & \leq & n^{ (1 / \eps) - 1 } \ ,
    \end{eqnarray*}
    where the second inequality holds since each bucket ${\cal B}_{\ell}$ contains items whose weight falls within $[n^{\ell-1} \cdot w_{\min}, n^{\ell} \cdot w_{\min})$, and the third inequality follows by noting that each cluster represents the union of at most $\frac{ 1 }{ \eps } - 1$ successive buckets, implying that $\ell_2 - \ell_1 \leq \frac{ 1 }{ \eps } - 2$.
    
    \item {\em Weight gap between clusters}: Similarly, let $i_1$ and $i_2$ be a pair of items that belong to clusters ${\cal C}^r_{m_1}$ and ${\cal C}^r_{m_2}$, respectively, with $m_1 < m_2$. In this case, when we denote the corresponding buckets by ${\cal B}_{\ell_1}$ and ${\cal B}_{\ell_2}$, their weight ratio can be lower bounded by  
    \begin{eqnarray*}
    \frac{ w_{i_2} }{ w_{i_1} } & \geq & \frac{ \min_{ i \in {\cal B}_{\ell_2} } w_i }{ \max_{ i \in {\cal B}_{\ell_1} } w_i } \\
    & \geq & n^{(\ell_2 - 1) - \ell_1} \\
    & \geq & n^{ 1 + (m_2 - m_1 - 1 ) / \eps } \ ,
    \end{eqnarray*}
    where the last inequality holds since $\ell_1 \in \{r + 1 + \frac{ m_1 - 2 }{ \eps }, \ldots, r - 1 + \frac{ m_1 - 1 }{ \eps } \}$ and $\ell_2 \in \{r + 1 + \frac{ m_2 - 2 }{ \eps }, \ldots, r - 1 + \frac{ m_2 - 1 }{ \eps } \}$, by definition of ${\cal C}^r_{m_1}$ and ${\cal C}^r_{m_2}$. 
\end{enumerate}

We conclude the proof by showing that at least one of the well-spaced instances ${\cal I}_0, \ldots {\cal I}_{ \frac{ 1 }{ \eps } - 1 }$ is associated with an optimal profit of at least $(1 - \eps) \cdot \objfunc({\cal S}^*)$. To this end, with respect to the optimal chain ${\cal S}^*$ for the original instance ${\cal I}$, note that the restriction of this chain ${\cal S}^*|_{{\cal N}_r}$ to the item set ${\cal N}_r$ is clearly feasible for ${\cal I}_r$, by Observation~\ref{obs:subset}. Letting ${\cal S}^{r*}$ be an optimal chain for ${\cal I}_r$, we consequently have
\begin{eqnarray*}
\max_{0 \leq r \leq (1/\eps)-1} \objfunc({\cal S}^{r*}) & \geq & \max_{0 \leq r \leq (1/\eps)-1} \objfunc ( {\cal S}^*|_{{\cal N}_r} ) \\
& \geq & \eps \cdot \sum_{r = 0}^{(1/\eps)-1} \objfunc({\cal S}^*|_{{\cal N}_r}) \\
& = & \eps \cdot \sum_{r = 0}^{(1/\eps)-1} \sum_{t \in [T]} \sum_{i \in (S_t^* \setminus S_{t-1}^*) \cap {\cal N}_r} p_{it} \\
& = & \eps \cdot \sum_{t \in [T]} \sum_{i \in S_t^* \setminus S_{t-1}^*} \left| \left\{ r \in \left\{ 0, \ldots, \frac{ 1 }{ \eps } - 1 \right\} : i \in \left( S_t^* \setminus S_{t-1}^* \right) \cap {\cal N}_r \right\} \right| \cdot p_{it} \\
& = & (1 - \eps ) \cdot \sum_{t \in [T]} \sum_{i \in S_t^* \setminus S_{t-1}^*} p_{it} \\
& = & (1 - \eps) \cdot \objfunc({\cal S}^*) \ ,
\end{eqnarray*}
where the next-to-last equality holds since every item introduced in the optimal chain ${\cal S}^*$ appears in all but one of the sets ${\cal N}_0, \ldots, {\cal N}_{(1/\eps) - 1 }$.

\subsection{Proof of Lemma~\ref{lem:structure-cross}: The sparse-crossing property} \label{subsec:proof_structure-cross}

\paragraph{Preliminaries.} We begin by introducing some additional definitions and notation that will be utilized throughout this proof. For a set of cluster indices ${\cal M} \subseteq [M]$, we use ${\cal C}_{\cal M}$ to designate the union of ${\cal M}$-indexed clusters, i.e., ${\cal C}_{\cal M} = \biguplus_{m \in {\cal M}} {\cal C}_m$. Expanding upon the definition of $\mycross_m(\pi)$, given disjoint sets, ${\cal M}_1 \subseteq [M]$ and ${\cal M}_2 \subseteq [M]$, let $\mycross_{{\cal M}_1, {\cal M}_2}(\pi)$ denote the number of items in ${\cal C}_{{\cal M}_2}$ that appear in the permutation $\pi$ before the last item in ${\cal C}_{{\cal M}_1}$, namely,
\[ \mycross_{{\cal M}_1, {\cal M}_2}(\pi) = \left| \left\{i \in {\cal C}_{{\cal M}_2} : \pi(i) < \max_{ j \in {\cal C}_{{\cal M}_1}} \pi(j) \right\} \right|.  \]
When $\mycross_{{\cal M}_1, {\cal M}_2}(\pi) \geq \frac{1}{\eps}$, we use ${\cal X}_{{\cal M}_1, {\cal M}_2}(\pi)$ to designate the set comprised of the first $\frac{1}{\eps}$ items in ${\cal M}_2$-indexed clusters in the permutation $\pi$. When $\mycross_{{\cal M}_1, {\cal M}_2}(\pi) < \frac{1}{\eps}$, we simply set ${\cal X}_{{\cal M}_1, {\cal M}_2}(\pi) = \emptyset$.

\paragraph{Fixing permutations.} In order to formalize the notion of ``pulling back'' items within a given permutation, as briefly sketched in Section~\ref{subsec:qptas2-overview}, we define a fixing procedure, $\myfix( \pi, {\cal M}^-, {\cal M}^+ )$. Here, we receive as input a permutation $\pi: {\cal Q} \to [|{\cal Q}|]$ over an item set ${\cal Q} \subseteq {\cal N}$, along with two disjoint sets of cluster indices, ${\cal M}^-$ and ${\cal M}^+$, which are assumed to satisfy $\max {\cal M}^- < \min {\cal M}^+$, i.e., any index in ${\cal M}^-$ is strictly smaller than any index in ${\cal M}^+$. As explained below, this procedure constructs in polynomial time a modified permutation $\bar{\pi} : \bar{\cal Q} \to [|\bar{\cal Q}|]$, over a subset $\bar{\cal Q} \subseteq {\cal Q}$, that satisfies the following properties:
\begin{enumerate}
    \item {\em Sparse $({\cal M}^-, {\cal M}^+)$-crossing:} $\mycross_{{\cal M}^-, {\cal M}^+}(\bar{\pi}) \leq \frac{ 1 }{ \eps }$.
    
    \item {\em Completion times:} $C_{\bar{\pi}}(i) \leq C_{\pi}(i)$, for every $i \in \bar{\cal Q}$.
    
    \item {\em Difference:} ${\cal Q} \setminus \bar{\cal Q}$ consists of at most one item, which is a member of ${\cal X}_{{\cal M}^-, {\cal M}^+}(\pi)$. 
\end{enumerate}

For this purpose, when $\mycross_{{\cal M}^-, {\cal M}^+}(\pi) < \frac{ 1 }{ \eps }$, the procedure $\myfix(\pi, {\cal M}^-, {\cal M}^+)$ returns exactly the same permutation (i.e., $\bar{\pi} = \pi$), without any alterations. In the opposite case, when $\mycross_{{\cal M}^-, {\cal M}^+}(\pi) \geq \frac{ 1 }{ \eps }$, let $i_{{\cal M}^-, {\cal M}^+}$ be the least profitable item in ${\cal X}_{{\cal M}^-, {\cal M}^+}(\pi)$ with respect to the optimal permutation $\pi^*$, namely, $i_{{\cal M}^-, {\cal M}^+} = \argmin \{\varphi_{\pi^*}(i): i \in {\cal X}_{{\cal M}^-, {\cal M}^+}(\pi) \}$. Our construction consists of eliminating $i_{{\cal M}^-, {\cal M}^+}$ and placing instead all items in ${\cal C}_{{\cal M}^-}$ appearing in $\pi$  after $i_{{\cal M}^-, {\cal M}^+}$; this alteration results in a permutation $\bar{\pi}$ over ${\cal Q} \setminus \{ i_{{\cal M}^-, {\cal M}^+} \}$. Formally, let ${\cal A}^-$ and $\bar{\cal A}^-$ be the items appearing after $i_{{\cal M}^-, {\cal M}^+}$ out of ${\cal C}_{{\cal M}^-}$ and ${\cal N} \setminus {\cal C}_{{\cal M}^-}$, respectively, i.e., 
\[ {\cal A}^-= \left\{i \in {\cal C}_{{\cal M}^-}: \pi(i) > \pi(i_{{\cal M}^-, {\cal M}^+}) \right\} \text{  and  }  \bar{\cal A}^- = \left\{i \in {\cal N} \setminus {\cal C}_{{\cal M}^-}: \pi(i)>\pi(i_{{\cal M}^-, {\cal M}^+}) \right\} \ . \]
For simplicity, we index the items in ${\cal A}^-$ according to their order within the permutation $\pi$, which results in having ${\cal A}^-=\{i_1,\dots,i_{|{\cal A}^-|}\}$ with $\pi(i_1) < \dots < \pi(i_{|{\cal A}^-|})$. Now, the modified permutation $\bar{ \pi}$ is constructed as follows:
\begin{itemize}
    \item {\em Before $i_{{\cal M}^-, {\cal M}^+}$:} Items in positions $1, \dots, \pi(i_{{\cal M}^-, {\cal M}^+})-1$ of the permutation $\pi$ remain within their original positions, meaning that $\bar{\pi}(i) = \pi(i)$ for every item $i$ with $\pi(i)\leq \pi(i_{{\cal M}^-, {\cal M}^+})-1$.
    
    \item {\em Instead of $i_{{\cal M}^-, {\cal M}^+}$:} Items in ${\cal A}^-$ will appear in place of $i_{{\cal M}^-, {\cal M}^+}$ following their relative order in $\pi$. That is, $\bar{\pi}(i_k) = \pi(i_{{\cal M}^-, {\cal M}^+}) - 1 + k$ for every $k \in [|{\cal A}^-|]$.
    
    \item {\em After $i_{{\cal M}^-, {\cal M}^+}$:} Items in $\bar{{\cal A}}^-$ will appear after those in ${\cal A}^-$, again following their relative order in $\pi$. In other words, $\bar{\pi}(i) = \pi(i) - 1 + | \{ k \in [|{\cal A}^-|] : \pi(i_k)>\pi(i)\} |$ for every item $i \in \bar{{\cal A}}^-$. 
\end{itemize}
In Appendix~\ref{app:proof_lem_mod_perm_properties}, we show that the resulting permutation satisfies its desired properties, as formally stated below. 

\begin{lemma} \label{lem:mod_perm_properties}
The permutation $\bar{\pi}$ satisfies properties~1-3. 
\end{lemma}

\paragraph{The recursive construction.} We are now ready to explain how recursive applications of the fixing procedure allow us to conclude the proof of Lemma~\ref{lem:structure-cross}. At a high level, we bisect the cluster indices $[M]$, such that in each step the indices being considered are split into their lower half ${\cal M}^-$ and upper half ${\cal M}^+$, with respect to which the fixing procedure $\myfix(\cdot, {\cal M}^-, {\cal M}^+ )$ will be applied. The resulting permutation will then be divided into left and right parts, which are recursively bisected along the same lines.

To present the specifics of this bisection as simply as possible, we assume without loss of generality that the number of clusters $M$ is a power of $2$; otherwise, empty clusters can be appended to the sequence ${\cal C}_1, \ldots, {\cal C}_M$. At the upper level of the recursion, we bisect the entire collection of cluster indices $[M]$ into ${\cal M}_{ [1,\frac{M}{2}] }= \{1, \dots ,\frac{M}{2} \}$ and ${\cal M}_{ [\frac{M}{2}+1,M] } = \{\frac{M}{2}+1, \dots, M \}$. Designating the optimal permutation by $\pi_{[1,M]} = \pi^*$, we employ our fixing procedure with $\myfix(\pi_{[1,M]}, {\cal M}_{ [1,\frac{M}{2}] }, {\cal M}_{ [\frac{M}{2}+1,M] })$, to obtain the permutation $\bar{\pi}_{[1,M]}$. Now, we break the latter into its left and right part, $\pi_{[1,\frac{M}{2}]}$ and $\pi_{[\frac{M}{2}+1,M]}$, such that the left permutation $\pi_{[1,\frac{M}{2}]}$ is the prefix of $\bar{\pi}_{[1,M]}$ ending at the last item in ${\cal C}_{{\cal M}_{[1, \frac{M}{2}]}} \cup {\cal X}_{{\cal M}_{[1, \frac{M}{2}]},{\cal M}_{[\frac{M}{2}+1, M]}}(\pi_{[1,M]})$, whereas the right permutation $\pi_{[\frac{M}{2}+1,M]}$ is comprised of the remaining suffix.  

In the second level of the recursion, for the left permutation $\pi_{[1,\frac{M}{2}]}$, we bisect ${\cal M}_{ [1,\frac{M}{2}] }$ into ${\cal M}_{ [1,\frac{M}{4}] } = \{ 1, \dots, \frac{M}{4} \}$ and ${\cal M}_{ [\frac{M}{4}+1, \frac{M}{2}] } = \{ \frac{M}{4}+1, \dots, \frac{M}{2} \}$, followed by applying $\myfix( \pi_{[1,\frac{M}{2}]}, {\cal M}_{ [1,\frac{M}{4}] }, {\cal M}_{ [\frac{M}{4}+1, \frac{M}{2}] } )$. Similarly, for the right permutation $\pi_{[\frac{M}{2}+1,M]}$, its corresponding set of cluster indices ${\cal M}_{ [\frac{M}{2}+1,M] }$ is bisected into ${\cal M}_{ [\frac{M}{2}+1,\frac{3M}{4}] } = \{ \frac{M}{2}+1, \dots, \frac{3M}{4} \}$ and ${\cal M}_{ [\frac{3M}{4}+1,M] } =  \{ \frac{3M}{4}+1, \dots, M \}$, in which case we apply $\myfix( \pi_{[\frac{M}{2}+1,M]}, {\cal M}_{ [\frac{M}{2}+1,\frac{3M}{4}] }, {\cal M}_{ [\frac{3M}{4}+1,M] } )$. This recursive procedure continues up until the resulting sets of cluster indices are singletons. At that point in time, our final permutation $\pi_{\mysparse}$ is obtained by concatenating $\pi_{[1,1]}, \pi_{[2,2]}, \dots, \pi_{[M,M]}$.

\paragraph{Analysis.} For ease of presentation, we make use of $\Omega$ to denote the set of pairs of cluster index sets with respect to which $\myfix(\cdot,\cdot,\cdot)$ is employed throughout our recursive construction, meaning that 
\[ \begin{array}{llc}
\Omega = \biggl{\{} & \left( {\cal M}_{ [1,\frac{M}{2}] }, {\cal M}_{ [\frac{M}{2}+1,M] } \right) , & [\text{level 1}] \\
& \left( {\cal M}_{ [1,\frac{M}{4}] }, {\cal M}_{ [\frac{M}{4}+1, \frac{M}{2}] } \right), \left( {\cal M}_{ [\frac{M}{2}+1,\frac{3M}{4}] }, {\cal M}_{ [\frac{3M}{4}+1,M] } \right), & [\text{level 2}] \\
& \cdots \\
& \left( {\cal M}_{[1,1]}, {\cal M}_{[2,2]} \right), \left( {\cal M}_{[3,3]}, {\cal M}_{[4,4]} \right), \ldots, \left( {\cal M}_{[M-1,M-1]}, {\cal M}_{[M,M]} \right) \, \, \biggl{\}} . \qquad & [\text{level $\log_2 M$}]
\end{array} \]
With this notation, we show in the next two claims that the permutation $\pi_{\mysparse}$ indeed satisfies the sparse crossing and near-optimal profit properties of Lemma~\ref{lem:structure-cross}. 

\begin{lemma} \label{lem:cross-bound}
$\mycross_m( \pi_{\mysparse} ) \leq \frac{\log_2M}{\eps}$, for every $m \in [M]$.
\end{lemma}
\begin{proof}
By construction of $\pi_{\mysparse}$, every item belonging to one of the clusters ${\cal C}_{m+1}, \ldots, {\cal C}_M$ that appears in this permutation before the last item in cluster ${\cal C}_m$ necessarily resides in ${\cal X}_{{\cal M}^-, {\cal M}^+}(\pi_{[\min {\cal M}^-, \max {\cal M}^+]})$, for some pair $({\cal M}^-, {\cal M}^+) \in \Omega$ with $m \in {\cal M}^-$. To verify this claim, consider such a crossing item $i$, say belonging to cluster ${\cal C}_{m^+}$. By the way our recursive construction of $\Omega$ is defined, there exists a unique pair of cluster index sets $({\cal M}^-, {\cal M}^+) \in \Omega$ for which $m \in {\cal M}^-$ and $m^+ \in {\cal M}^+$; we argue that $i \in {\cal X}_{{\cal M}^-, {\cal M}^+}(\pi_{[\min {\cal M}^-, \max {\cal M}^+]})$. Indeed, in the next recursion level, the left permutation $\pi_{[\min {\cal M}^-, \max {\cal M}^-]}$ is the prefix of $\bar{\pi}_{[\min {\cal M}^-, \max {\cal M}^+]}$ ending with the last item in ${\cal C}_{{\cal M}^-} \cup {\cal X}_{{\cal M}^-, {\cal M}^+}(\pi_{[\min {\cal M}^-, \max {\cal M}^+]})$. Furthermore, by construction, all items in the right permutation $\pi_{[\min {\cal M}^+, \max {\cal M}^+]}$ will appear in $\pi_{\mysparse}$ after all items in the left permutation $\pi_{[\min {\cal M}^-, \max {\cal M}^-]}$. Therefore, since $i \in {\cal C}_{m^+}$ with $m^+ \in {\cal M}^+$ and since this item appears in $\pi_{\mysparse}$ before the last item in cluster ${\cal C}_m$, we know that $i$ appears as part of the left permutation $\pi_{[\min {\cal M}^-, \max {\cal M}^-]}$, implying that $i \in {\cal X}_{{\cal M}^-, {\cal M}^+}(\pi_{[\min {\cal M}^-, \max {\cal M}^+]})$.

As any such item $i \in {\cal X}_{{\cal M}^-, {\cal M}^+}(\pi_{[\min {\cal M}^-, \max {\cal M}^+]})$ contributes at most once toward $\mycross_{{\cal M}^-, {\cal M}^+}(\bar{\pi}_{[\min {\cal M}^-, \max {\cal M}^+]})$, we have
\begin{eqnarray*}
\mycross_m(\pi_{\mysparse}) & \leq & \sum_{\MyAbove{({\cal M}^-,{\cal M}^+) \in \Omega :}{ m \in {\cal M}^- }} \mycross_{{\cal M}^-, {\cal M}^+} \left( \bar{\pi}_{[\min {\cal M}^-, \max {\cal M}^+]} \right)\\
& \leq & \frac{ 1 }{ \eps } \cdot \left| \left\{({\cal M^-},{\cal M}^+) \in \Omega: m \in {\cal M}^-  \right\} \right| \\
& \leq & \frac{\log_2M }{\eps} \ .
\end{eqnarray*}
Here, the second inequality holds since $\mycross_{{\cal M}^-,{\cal M}^+}(\bar{\pi}_{[\min {\cal M}^-, \max {\cal M}^+]}) \leq \frac{1}{\eps}$ by property~1 of the fixing procedure. The third inequality is obtained by observing that, as the definition of $\Omega$ shows, all sets appearing in a single level of the recursion form a partition of $[M]$, implying that $m \in {\cal M}^-$ for at most one pair $({\cal M}^-, {\cal M}^+)$ in that level.  As there are $\log_2 M$ levels overall, it follows that $| \{({\cal M^-},{\cal M}^+) \in \Omega: m \in {\cal M}^-  \} | \leq \log_2 M$. 
\end{proof}

\begin{lemma}
$\newobj( \pi_{\mysparse} ) \geq (1 -  \eps ) \cdot \newobj( \pi^* )$.
\end{lemma}
\begin{proof}
To prove the desired claim, we first establish two auxiliary claims, that will enable us to relate between the profits $\newobj( \pi_{\mysparse} )$ and $\newobj( \pi^* )$. For ease of presentation, the corresponding proofs can be found in Appendices~\ref{app:proof_cl_pi_sparse_profit} and~\ref{app:proof_cl_disjoint_cross}, respectively. 

\begin{claim} \label{cl:pi-sparse-profit}
$\newobj(\pi_{\mysparse}) \geq \newobj(\pi^*) - \eps \cdot \sum_{({\cal M}^-, {\cal M}^+) \in \Omega} \varphi_{\pi^*}( {\cal X}_{{\cal M}^-, {\cal M}^+}({\pi}_{[\min {\cal M}^-, \max {\cal M}^+]}))$. 
\end{claim}

\begin{claim} \label{cl:disjoint-cross}
For any two distinct pairs $({\cal M}^-_1, {\cal M}^+_1)$ and $({\cal M}^-_2, {\cal M}^+_2)$ in $\Omega$, the item sets ${\cal X}_{{\cal M}^-_1, {\cal M}^+_1}({\pi}_{[\min {\cal M}^-_1, \max {\cal M}^+_1]})$ and ${\cal X}_{{\cal M}^-_2, {\cal M}^+_2}({\pi}_{[\min {\cal M}^-_2, \max {\cal M}^+_2]})$ are disjoint. 
\end{claim}

Consequently, the profit attained by the permutation $\pi_{\mysparse}$ can be bounded by noting that 
\begin{eqnarray*}
\newobj(\pi_{\mysparse}) & \geq & \newobj(\pi^*) - \eps \cdot \sum_{({\cal M}^-, {\cal M}^+) \in \Omega} \varphi_{\pi^*} \left( {\cal X}_{{\cal M}^-, {\cal M}^+} \left({\pi}_{[\min {\cal M}^-, \max {\cal M}^+]} \right) \right) \\
& \geq & \newobj(\pi^*) - \eps \cdot \sum_{i \in {\cal N}} \varphi_{\pi^*}(i) \\
& = & (1 - \eps) \cdot \newobj(\pi^*), 
\end{eqnarray*}
where the first inequality is precisely  Claim~\ref{cl:pi-sparse-profit}, and the second inequality follows from Claim~\ref{cl:disjoint-cross}.
\end{proof}

\subsection{The external dynamic program} \label{subsec:DP_general}

Given the sparse-crossing property of the near-optimal permutation $\pi_{\mysparse}$, whose existence has been established in Lemma~\ref{lem:structure-cross}, we turn our attention to formally presenting a dynamic programming approach for computing a permutation with a profit of at least $(1-2\eps) \cdot \newobj(\pi_{\mysparse})$. 

\paragraph{States.} Building on the intuition provided in Section~\ref{subsec:qptas2-overview}, we remind the reader that each state $(m, \psi_m, {\cal Q}_{>m})$ of our dynamic program consists of the following parameters:
\begin{itemize}
    \item The index of the current cluster $m$, taking values in $[M]_0$.
    
    \item The total profit $\psi_m$ collected thus far. Initially, $\psi_m$ will be treated as a continuous parameter, taking values in $[0, n p_{\max}]$, where $p_{\max}$ is the maximum profit attainable by any single item, i.e., $p_{\max} = \max \{ p_{it}: i \in [n], t \in [T], \text{ and } w_i \leq W_t \}$.
    
    \item The set of items ${\cal Q}_{>m}$ belonging to clusters ${\cal C}_{m+1}, \ldots, {\cal C}_M$ that will be crossing into lower-index clusters.
    Motivated by the sparse-crossing property established in Lemma~\ref{lem:structure-cross}, we only consider sets ${\cal Q}_{>m}$ of cardinality at most $\frac{\lceil \log_2M \rceil}{\eps}$.
\end{itemize}

\paragraph{Value function.} For a subset of items $S \subseteq {\cal N}$ and a permutation $\pi: S \to [|S|]$, we say that the pair $(S, \pi)$ is thin when $\mycross_m(\pi) \leq \frac{ \lceil \log_2 M \rceil }{ \eps }$ for all $m \in [M]$. Given this definition, the value function $F(m, \psi_m, {\cal Q}_{>m})$ represents the minimum makespan $w(S)$ that can be attained over all thin pairs $(S, \pi)$ that satisfy the following conditions:
\begin{enumerate}
    \item {\em Allowed items}: The set $S$ consists of items that belong to one of the clusters ${\cal C}_1 \dots, {\cal C}_m$ or to ${\cal Q}_{>m}$. In other words, $S \subseteq {\cal C}_{[1,m]}  \uplus {\cal Q}_{>m}$, where ${\cal C}_{[1,m]} = \biguplus_{\mu \in [1,m]} {\cal C}_{\mu}$ by convention. 
    
    \item {\em Required crossing items}: The set $S$ contains all items in ${\cal Q}_{>m}$, meaning that ${\cal Q}_{>m} \subseteq S$. 
    
    \item {\em Total profit}: $\newobj( \pi ) \geq \psi_m$.
\end{enumerate}
Recycling some of the notation introduced in Section~\ref{subsec:continuous_dp}, we use $\mythin(m, \psi_m, {\cal Q}_{>m})$ to denote the collection of thin pairs that meet  conditions~1-3 above. When the latter set is empty, we define $F(m, \psi_m, {\cal Q}_{>m}) = \infty$. With these definitions, Lemma~\ref{lem:structure-cross} proves the existence of a thin pair $(S,\pi) \in \mythin(M, \newobj(\pi_{\mysparse}), \emptyset)$ with $F(M, \newobj(\pi_{\mysparse}), \emptyset) \leq W_T$. It is worth pointing out that, for the item set ${\cal N}_{\mysparse}$ over which the permutation $\pi_{\mysparse}$ is defined, we can indeed assume that $w({\cal N}_{\mysparse}) \leq W_T$, as all items whose completion time with respect to $\pi_{\mysparse}$ exceeds $W_T$ can be eliminated, leaving us with a permutation that still satisfies Lemma~\ref{lem:structure-cross}. Therefore, had we been able to compute the maximal value $\psi^*$ for which $F(M, \psi^*, \emptyset) \leq W_T$, its corresponding permutation would have guaranteed a profit of at least $\psi^* \geq \newobj(\pi_{\mysparse}) \geq (1 -  \eps ) \cdot \newobj( \pi^* )$. Once again, since ${\psi}_m$ is a continuous parameter, we will eventually explain how to discretize $\psi_m$ to take polynomially-many values, incurring only an $\eps$-loss in profit.

\paragraph{Optimal substructure.} In what follows, we identify the optimal substructure that allows us to compute the value function $F$ by means of dynamic programming. To this end, suppose that $(S, \pi)$ is a thin pair that minimizes $w(S)$ over $\mythin(m, \psi_m, {\cal Q}_{>m})$. We will argue that by eliminating from $(S, \pi)$ a carefully-selected suffix of the permutation $\pi$ consisting of items in clusters ${\cal C}_m, \ldots, {\cal C}_M$, one obtains a thin pair that attains $F(m-1, \psi_{m-1}, {\cal Q}_{>m-1})$ for an appropriately defined state $(m-1, \psi_{m-1}, {\cal Q}_{>m-1})$. We proceed by first defining the latter state, for which a suitable alteration of $(S, \pi)$ will be shown to be optimal: 
\begin{itemize}
    \item {\em Crossing set}: ${\cal Q}_{>m-1}$ is defined as the set of items in ${\cal C}_m \uplus {\cal Q}_{>m}$ that appear before the last item in ${\cal C}_1, \dots, {\cal C}_{m-1}$ with respect to the permutation $\pi$. Namely, 
    \begin{equation} \label{eqn:def_qmmin1}
    {\cal Q}_{>m-1} = \left\{ i \in S \cap ({\cal C}_m \uplus {\cal Q}_{>m}): \pi(i) < \max_{ j \in S \cap {\cal C}_{[1,m-1]} } \pi(j) \right\} \ .
    \end{equation}
  
    \item {\em Profit requirement}: $\psi_{m-1} = [\psi_m - \sum_{i \in S \setminus ({\cal C}_{[1,m-1]} \uplus {\cal Q}_{>m-1})} \varphi_{\pi}(i)]^+$.
\end{itemize}
It is worth pointing out that, for this state to be well-defined, we should ensure that ${\cal Q}_{>m-1}$ indeed consists of at most $\frac{\lceil \log_2M \rceil}{\eps}$ items. To understand why this property is satisfied, note that since every item in ${\cal Q}_{>m-1}$ appears in the permutation $\pi$ before the last item in $S \cap {\cal C}_{[1,m-1]}$, we have $|{\cal Q}_{>m-1}| \leq \max_{\mu \in [m-1]} \mycross_{\mu}(\pi) \leq \frac{\lceil \log_2 M \rceil}{\eps}$, where the last inequality holds since $(S, \pi)$ is thin.

Now, let us define the pair $(\hat{S}, \hat{\pi})$, in which $\hat{S} = S \cap ({\cal C}_{[1,m-1]} \uplus {\cal Q}_{>m-1})$, meaning that $\hat{S}$ is the restriction of $S$ to items belonging to either one of the clusters ${\cal C}_1, \dots, {\cal C}_{m-1}$ or to ${\cal Q}_{>m-1}$. It is not difficult to verify that any item in $\hat{S}$ appears in $\pi$ before any item in $S \setminus \hat{S}$, as any item in $S \cap {\cal C}_{[m,M]}$ that appears before an item in ${\cal C}_{[1,m-1]}$ is necessarily a member of ${\cal Q}_{>m-1}$. Therefore, the items in $\hat{S}$ form a prefix of $\pi$, whereas those in $S \setminus \hat{S}$ form the remaining suffix. Given this observation, we define the permutation $\hat \pi: \hat{S} \to [|\hat{S}|]$ as the former prefix, or equivalently, as the restriction of $\pi$ to the items in $\hat{S}$.

In Lemma~\ref{lem:dp-feasibility-subproblem} below, we show that the pair $(\hat{S}, \hat{\pi})$ indeed resides within $\mythin (m-1, \psi_{m-1}, {\cal Q}_{>m-1})$. Subsequently, we prove in Lemma~\ref{lem:qptas-dp-substructure} that this pair is in fact makespan-optimal over the latter set. To avoid deviating from the overall flow of this section, the corresponding proofs are presented in Appendices~\ref{app:proof_lem_dp_feasibility_subproblem} and~\ref{app:proof_lem_qptas_dp_substructure}, respectively.

\begin{lemma} \label{lem:dp-feasibility-subproblem}
$(\hat{S}, \hat{\pi}) \in \mythin(m-1, \psi_{m-1}, {\cal Q}_{>m-1})$.
\end{lemma}

\begin{lemma} \label{lem:qptas-dp-substructure}
$w(\hat{S}) = F(m-1, \psi_{m-1}, {\cal Q}_{>m-1})$. 
\end{lemma}

\paragraph{Recursive equations.} Given the optimal substructure characterization discussed above, we proceed by explaining how to express $F(m, \psi_m, {\cal Q}_{>m})$ in recursive form. In essence, had we known what the preceding state $(m-1, \psi_{m-1}, {\cal Q}_{>m-1})$ is, the remaining question would have been that of identifying the lightest set of ``extra'' items ${\cal E}$ to be appended, along with their internal permutation $\pi_{\cal E} : {\cal E} \to [|{\cal E}|]$, under a marginal profit constraint. Formally, to capture the agreement between crossing items, we say that state $(m-1, \psi_{m-1}, {\cal Q}_{>m-1})$ is conceivable for state $(m, \psi_m, {\cal Q}_{>m})$ when ${\cal Q}_{>m-1} \setminus {\cal C}_m \subseteq {\cal Q}_{>m}$. In the opposite direction, $(m, \psi_m, {\cal Q}_{>m})$ is reachable from $(m-1, \psi_{m-1}, {\cal Q}_{>m-1})$ when there exist an item set ${\cal E}$ and permutation $\pi_{\cal E} : {\cal E} \to [|{\cal E}|]$ that simultaneously satisfy the following constraints:
\begin{enumerate}
    \item {\em Extra items}: The collection of extra items can be written as ${\cal E} = {\cal E}_{m} \uplus ({\cal Q}_{>m} \setminus {\cal Q}_{>m-1})$. Here, items in ${\cal E}_{m}$ are to be picked out of cluster ${\cal C}_m$, with the exclusion of those appearing in ${\cal Q}_{>m-1}$, meaning that we have the constraint ${\cal E}_{m} \subseteq {\cal C}_m \setminus {\cal Q}_{>m-1}$. Concurrently, each and every item in ${\cal Q}_{>m} \setminus {\cal Q}_{>m-1}$ should be picked. 
    
    \item {\em Marginal profit}: $\sum_{i \in {\cal E}} \varphi_{ \pi_{\cal E} }^{ \rightsquigarrow }( i ) \geq \psi_m - \psi_{m-1}$, where the term $\varphi_{ \pi_{\cal E} }^{ \rightsquigarrow }( i )$ denotes the profit of item $i$ with respect to the permutation $\pi_{\cal E}$, when its completion time is increased by $F(m-1, \psi_{m-1}, {\cal Q}_{>m-1})$. This constraint guarantees that, by appending  $\pi_{\cal E}$ to the permutation that achieves $F(m-1, \psi_{m-1}, {\cal Q}_{>m-1})$, we obtain a total profit of at least $\psi_m$. 
\end{enumerate}
Letting $\myextra[ \MyAbove{ (m, \psi_m, {\cal Q}_{>m}) }{ (m-1, \psi_{m-1}, {\cal Q}_{>m-1}) } ]$ denote the collection of item sets and permutations that satisfy these constraints, we mention in passing that this set may be empty. Moreover, it will be utilized only for purposes of analysis, and in particular, we will not assume that $\myextra[ \MyAbove{ (m, \psi_m, {\cal Q}_{>m}) }{ (m-1, \psi_{m-1}, {\cal Q}_{>m-1}) } ]$ can be efficiently constructed. Nevertheless, the function value $F(m, \psi_m, {\cal Q}_{>m})$ can still be expressed by minimizing $F(m-1, \psi_{m-1}, {\cal Q}_{>m-1}) + w( {\cal E} )$ over all conceivable states $(m-1, \psi_{m-1}, {\cal Q}_{>m-1})$ and over all item sets and permutations $({\cal E}, \pi_{\cal E}) \in \myextra[ \MyAbove{ (m, \psi_m, {\cal Q}_{>m}) }{ (m-1, \psi_{m-1}, {\cal Q}_{>m-1}) } ]$. For convenience, when $F(m, \psi_m, {\cal Q}_{>m}) \leq W_T$, we use $\mybest(m, \psi_m, {\cal Q}_{>m})$ to denote an arbitrary state $(m-1, \psi_{m-1}, {\cal Q}_{>m-1})$ chosen out of those for which the minimum value $F(m, \psi_m, {\cal Q}_{>m})$ is attained. As mentioned earlier, we wish to compute the maximal value $\psi^*$ that satisfies $F(M, \psi^*, \emptyset) \leq W_T$, as its corresponding permutation guarantees a profit of at least $(1 -  \eps ) \cdot \newobj( \pi^* )$.

\paragraph{Approximate recursion.} That said, due to having a lower bound on the marginal profit, even when $\mybest(m, \psi_m, {\cal Q}_{>m})$ is known, the recursive formulation above is expected to identify an item set and permutation $({\cal E}, \pi_{\cal E}) \in \myextra[ \MyAbove{ (m, \psi_m, {\cal Q}_{>m}) }{ \mybest(m, \psi_m, {\cal Q}_{>m}) } ]$ for which $w( {\cal E} )$ is minimized. This setting can be viewed as an ``inverse''  generalized incremental knapsack problem, where the objective is to minimize makespan rather than to maximize profit. To deal with this obstacle, we employ our QPTAS for bounded weight ratio instances (see Section~\ref{sec:qptas-one}) in order to approximately solve these recursive equations.

Specifically, for $\Delta \geq 0$, we say that constraint~2 is $(\eps,\Delta)$-satisfied when $\sum_{i \in {\cal E}} \varphi_{ \pi_{\cal E} }^{ +\Delta }( i ) \geq (1 - \eps) \cdot (\psi_m - \psi_{m-1})$, where $\varphi_{ \pi_{\cal E} }^{ +\Delta }( i )$ is the profit of item $i$ with respect to the permutation $\pi_{\cal E}$, when its completion time is increased by $\Delta$. As such, the standard sense of satisfying this constraint can be recovered by picking $\eps = 0$ and $\Delta = F(m-1, \psi_{m-1}, {\cal Q}_{>m-1})$. With this definition, we say that state $(m, \psi_m, {\cal Q}_{>m})$ is $(\eps,\Delta)$-reachable from state $(m-1, \psi_{m-1}, {\cal Q}_{>m-1})$ when there exist an item set ${\cal E}$ and permutation $\pi_{\cal E} : {\cal E} \to [|{\cal E}|]$ that satisfy constraint~1 and $(\eps,\Delta)$-satisfy constraint~2; as before, $\myextra_{ \eps,\Delta }[ \MyAbove{ (m, \psi_m, {\cal Q}_{>m}) }{ (m-1, \psi_{m-1}, {\cal Q}_{>m-1}) } ]$ will stand for the collection of such item sets and permutations. In what follows, we devise an auxiliary procedure for approximately solving the recursive equations, as summarized in the next claim; for readability purposes, the proof is deferred to Appendix~\ref{app:proof_lem_approx_rec_eqn}.

\begin{lemma} \label{lem:approx_rec_eqn}
Suppose that $(m, \psi_m, {\cal Q}_{>m})$ and $(m-1, \psi_{m-1}, {\cal Q}_{>m-1})$ are two given states, such that $F(m, \psi_m, {\cal Q}_{>m}) \leq W_T$ and $(m-1, \psi_{m-1}, {\cal Q}_{>m-1}) = \mybest(m, \psi_m, {\cal Q}_{>m})$. Given a parameter $\Delta \leq F(m-1, \psi_{m-1}, {\cal Q}_{>m-1})$, we can identify an item set $\hat{\cal E}$ and permutation $\hat{\pi}_{\hat{\cal E}} : \hat{\cal E} \to [|\hat{\cal E}|]$ for which:
\begin{enumerate}
    \item $(\hat{\cal E}, \hat{\pi}_{\hat{\cal E}}) \in \myextra_{ \eps,\Delta }[ \MyAbove{ (m, \psi_m, {\cal Q}_{>m}) }{ (m-1, \psi_{m-1}, {\cal Q}_{>m-1}) } ]$. 

    \item $w(\hat{\cal E}) \leq F(m, \psi_m, {\cal Q}_{>m}) - F(m-1, \psi_{m-1}, {\cal Q}_{>m-1})$.
\end{enumerate}
The running time of our algorithm is $O((nT)^{O(\frac{1}{\eps^6} \cdot( \log n + \log M))} \cdot |{\cal I}|^{O(1)})$, regardless of whether the assumptions above hold or not. 
\end{lemma}

With this procedure in-hand, we define an approximate value function $\hat{F}$, whose state space is identical to that of $F$. However, rather than attempting to solve an inverse  generalized incremental knapsack problem, the recursive equations through which $\hat{F}$ is defined will tackle the latter problem in an approximate way via our auxiliary procedure. To formalize this approach, the function value $\hat{F}(m, \psi_m, {\cal Q}_{>m})$ is evaluated as follows:
\begin{itemize}
    \item {\em Terminal states ($m=0$)}:  Here, we simply define $\hat{F}(0, \psi_0, {\cal Q}_{>0}) = F(0, \psi_0, {\cal Q}_{>0})$. While $F$-values are unknown in general, $F(0, \psi_0, {\cal Q}_{>0})$ evaluates to either $w({\cal Q}_{>0})$, when there exists a permutation $\pi_{ {\cal Q}_{>0} } : {\cal Q}_{>0} \to [| {\cal Q}_{>0} |]$ with profit $\newobj( \pi_{ {\cal Q}_{>0} } ) \geq \psi_0$, or to $\infty$ otherwise. This distinction can be made by enumerating over all permutations of ${\cal Q}_{>0}$ in time $O( (\frac{ 1 }{ \eps } \log M)^{ O(\frac{ 1 }{ \eps } \log M ) } ) = O( |{\cal I}|^{O((\frac{1}{\epsilon}\log |{\cal I}|)^{O(1)})} )$, since $|{\cal Q}_{>0}| \leq \frac{\lceil \log_2 M \rceil}{\eps}$.
    
    \item {\em General states $(m \in [M])$}: For each  state $(m-1, \psi_{m-1}, {\cal Q}_{>m-1})$, we instantiate Lemma~\ref{lem:approx_rec_eqn} with $\Delta = \hat{F}(m-1, \psi_{m-1}, {\cal Q}_{>m-1})$, to obtain the item set $\hat{\cal E}$ and its permutation $\hat{\pi}_{\hat{\cal E}} : \hat{\cal E} \to [|\hat{\cal E}|]$. The value $\hat{F}(m, \psi_m, {\cal Q}_{>m})$ is determined by minimizing $\hat{F}(m-1, \psi_{m-1}, {\cal Q}_{>m-1}) + w( \hat{\cal E} )$ over all conceivable states $(m-1, \psi_{m-1}, {\cal Q}_{>m-1})$ for which $(\hat{\cal E}, \hat{\pi}_{\hat{\cal E}}) \in \myextra_{ \eps,\Delta }[ \MyAbove{ (m, \psi_m, {\cal Q}_{>m}) }{ (m-1, \psi_{m-1}, {\cal Q}_{>m-1}) } ]$, noting that the latter condition can easily be tested.
\end{itemize}
It is important to emphasize that, when employing our auxiliary procedure above, we have no way of knowing a-priori whether the assumptions made in Lemma~\ref{lem:approx_rec_eqn} hold or not. Nevertheless, as we show in the next lemma, whose proof is provided in Appendix~\ref{app:proof_lem_approximate_F}, any profit requirement which is attainable by the original dynamic program $F$ can be attained up to factor $1-\eps$ by our approximate program $\hat{F}$. The precise relationship we establish between these functions can be formally stated as follows.

\begin{lemma} \label{lem:approximate-F}
Let $(m, \psi_m, {\cal Q}_{>m})$ be a state for which $F(m, \psi_m, {\cal Q}_{>m}) \leq W_T$. Then, $\hat{F}(m, \psi_m, {\cal Q}_{>m}) \leq F(m, \psi_m, {\cal Q}_{>m})$, where the makespan $\hat{F}(m, \psi_m, {\cal Q}_{>m})$  is attained by an item set $\hat{S}_m$ and a permutation $\hat{\pi}_{\hat{S}_m} : \hat{S}_m \to [| \hat{S}_m |]$ for which:
\begin{itemize}
    \item {\em Allowed and required items}: $\hat{S}_m \subseteq {\cal C}_{[1,m]}  \uplus {\cal Q}_{>m}$ and ${\cal Q}_{>m} \subseteq \hat{S}_m$.
    
    \item {\em Profit}: $\newobj(\hat{\pi}_{\hat{S}_m}) \geq (1 - \eps) \cdot \psi_m$. 
\end{itemize}
\end{lemma}

As previously mentioned, the primary intent of this section is to compute a permutation with a profit of at least $(1-2\eps) \cdot \newobj(\pi_{\mysparse})$. To argue that we have nearly achieved this objective, recall that Lemma~\ref{lem:structure-cross} proves the existence of a thin pair $(S, \pi) \in \mythin(M, \newobj(\pi_{\mysparse}), \emptyset)$ with $F(M,\newobj(\pi_{\mysparse}), \emptyset) \leq W_T$. Therefore, as an immediate consequence of Lemma~\ref{lem:approximate-F}, we infer that $\hat{F}(M,\newobj(\pi_{\mysparse}), \emptyset) \leq W_T$, which is attained by a permutation $\pi$ with a profit of $\newobj(\pi) \geq (1 - \eps) \cdot \newobj(\pi_{\mysparse})$.

\paragraph{The discrete program $\bs{\tilde{F}}$.} That said, the above-mentioned existence proof still does not correspond to a constructive algorithm, due to the continuity of the profit requirement parameter $\psi_m$. To discretize this parameter, similarly to Section~\ref{subsec:discretize_DP}, we restrict $\psi_m$ to a finite set of values, ${\cal D}_{\psi} = \{  d \cdot  \frac{ \eps p_{\max} }{ 2n } : d \in [\frac{2n^2}{\eps}]_0 \}$. In turn, we use  $\tilde{F}(m, \psi_m, {\cal Q}_{>m})$ to denote the resulting dynamic program over the discretized set of states, whose recursive equations are identical to those of $\hat{F}$, except for instantiating Lemma~\ref{lem:approx_rec_eqn} with $\Delta = \tilde{F}(m-1, \psi_{m-1}, {\cal Q}_{>m-1})$. 

We conclude our analysis by lower-bounding the best-possible profit achievable through this dynamic program, showing that it indeed matches that of the  permutation $\pi_{\mysparse}$ up to $\eps$-related terms. To avoid redundancy, we omit the corresponding proof, as it is nearly identical to that of Lemma~\ref{lem:descretize_psi}.

\begin{lemma} \label{lem:discretize-hat-F}
There exists a value $\tilde{\psi} \in {\cal D}_{\psi}$ such that $\tilde{\psi} \geq (1 - \eps) \cdot \newobj(\pi_{\mysparse})$ and such that $\tilde{F}(M, \tilde{\psi}, \emptyset) \leq W_T$. This makespan is attained by an item set $\tilde{S}$ and a permutation $\tilde{\pi}_{\tilde{S}}$ whose profit is $\newobj(\tilde{\pi}_{\tilde{S}}) \geq (1 - \eps) \cdot \tilde{\psi} \geq (1 - 2 \eps) \cdot \newobj(\pi_{\mysparse})$.
\end{lemma}

\paragraph{Running time.} We first observe that the function $\tilde{F}(m, \psi_m, {\cal Q}_{>m})$ is being evaluated over $O(n^{O(\frac{1}{\eps}\log M)} \cdot |{\cal I}|^{O(1)})$  possible states. To verify this claim, note that there are $O(M) = O(|{\cal I}|)$ choices for the cluster index $m$, and that the discretized profit parameter $\psi_m$ takes values in ${\cal D}_{\psi}$, with $|{\cal D}_{\psi}| = O(\frac{n^2}{\eps})$. In addition, the set of crossing items ${\cal Q}_{>m}$ is of cardinality at most $\frac{\lceil \log_2 M \rceil}{\eps}$, implying that there are only $O(n^{O(\frac{1}{\eps}\log M)})$ subsets to consider for this parameter. Now, evaluating $\tilde{F}(m, \psi_m, {\cal Q}_{>m})$ for a given state depends on its type:
\begin{itemize}
    \item {\em Terminal states ($m=0$)}: As previously explained, such states are handled by enumerating over all permutations of ${\cal Q}_{>0}$ in time $O( |{\cal I}|^{O((\frac{1}{\epsilon}\log |{\cal I}|)^{O(1)})} )$.
    
    \item {\em General states $(m \in [M])$}: Here, each state $(m-1, \psi_{m-1}, {\cal Q}_{>m-1})$ would involve a single application of our auxiliary procedure, running in $O((nT)^{O(\frac{1}{\eps^6} \cdot( \log n + \log M))} \cdot |{\cal I}|^{O(1)})$ according to Lemma~\ref{lem:approx_rec_eqn}. As argued above, there are only $O(n^{O(\frac{1}{\eps}\log M)} \cdot |{\cal I}|^{O(1)})$ states of the form $(m-1, \psi_{m-1}, {\cal Q}_{>m-1})$ to be considered.
\end{itemize}
Overall, we incur a running time of $O( |{\cal I}|^{O((\frac{1}{\epsilon}\log |{\cal I}|)^{O(1)})} )$, as stated in Theorem~\ref{thm:qptas2}.

\bibliographystyle{plainnat}
\bibliography{BIB-GIK}

\appendix
\addtocontents{toc}{\setcounter{tocdepth}{1}}

\section{Additional Proofs from Section~\ref{sec:2-approx}}

\subsection{Proof of Claim~\ref{clm:feasible_subproblem}} \label{app:proof_clm_feasible_subproblem}

We first show that $(\tilde{S}^+, \tilde{\pi}^+)$ is indeed a bulky pair. For this purpose, since $(\tilde{S}, \tilde{\pi} )$ is bulky, it suffices to explain why each item $i \in Q$ is necessarily $k_i$-heavy, where $k_i$ is the unique index for which $C_{ \tilde{\pi}^+ }( i ) \in {\cal I}_{k_i}$. This claim follows by noting that, for such items, the way we construct $(\tilde{S}^+, \tilde{\pi}^+)$ leads to a completion time of 
\begin{eqnarray}
C_{ \tilde{\pi}^+ }( i ) & = & w ( \tilde{S} ) + \sum_{j \in Q : \pi(j) \leq \pi(i)} w_j \nonumber \\
& < & w ( \hat{S} ) + \sum_{j \in Q : \pi(j) \leq \pi(i)} w_j \nonumber \\
& = & C_{ \pi }( i ) \ . \label{eqn:proof_bulky_comp}
\end{eqnarray}
Recalling that $Q = \{ i \in S : C_{ \pi }( i ) \in {\cal I}_k \}$, we have just shown that $k_i \leq k$, and since item $i$ is $k$-heavy due to the bulkiness of $(S,\pi)$, it is $k_i$-heavy as well.

We proceed by showing that $(\tilde{S}^+, \tilde{\pi}^+)$ satisfies conditions 1-3:
\begin{enumerate}
\item {\em Top index}: $\mytop(\tilde{S}^+, \tilde{\pi}^+) \leq k$. To verify this property, note that when $Q = \emptyset$, we clearly have $w( \tilde{S}^+ ) = w( \tilde{S} ) < w( \hat{S} ) = w(S)$, and therefore, $\mytop(\tilde{S}^+, \tilde{\pi}^+) \leq \mytop(S, \pi ) \leq k$. In the opposite case, where $Q \neq \emptyset$, the makespans of both $\tilde{S}^+$ and $S$ are attained by the respective completion times of precisely the same item in $Q$. However, by inequality~\eqref{eqn:proof_bulky_comp}, we have $C_{ \tilde{\pi}^+ }( i ) \leq C_{ \pi }( i )$ for every $i \in Q$, and it follows that $\mytop(\tilde{S}^+, \tilde{\pi}^+) \leq \mytop(S, \pi ) \leq k$.

\item {\em Total profit}: $\newobj( \tilde{\pi}^+ ) 
 \geq \psi_k$. Along the same lines, since $C_{ \tilde{\pi}^+ }( i ) \leq C_{ \pi }( i )$ for every $i \in Q$, it follows that $\varphi_{\tilde{\pi}^+ }( i ) \geq \varphi_{\pi}( i )$ for such items. Thus,
\begin{eqnarray*}
\newobj \left( \tilde{\pi}^+ \right) &=& \sum_{i \in \tilde{S}} \varphi_{\tilde{\pi}^+}(i) + \sum_{i \in Q} \varphi_{\tilde{\pi}^+}(i) \\
& = & \sum_{i \in \tilde{S}} \varphi_{\tilde{\pi}}(i) + \sum_{i \in Q} \varphi_{\tilde{\pi}^+}(i) \\
& \geq & \psi_{k-1} + \sum_{i \in Q}\varphi_{{\pi}}(i) \\
& = & \left[ \psi_k - \sum_{i \in Q} \varphi_{ \pi }( i ) \right]^+ + \sum_{i \in Q}\varphi_{{\pi}}(i) \\
& \geq & \psi_k \ .
\end{eqnarray*}
Here, the second equality holds since the permutations $\tilde{\pi}^+$ and $\tilde{\pi}$ are identical when restricted to items in $\tilde{S}$. The first inequality follows by recalling that $(\tilde{S},\tilde{\pi}) \in \mybulky(k-1, \psi_{k-1}, {\cal Q}_{k-1})$, meaning in particular that $\sum_{i \in \tilde{S}} \varphi_{\tilde{\pi}}(i) = \newobj( \tilde{\pi} ) \geq \psi_{k-1}$.

\item {\em Core}: $\mycore(\tilde{S}^+) = {\cal Q}_k$. One can easily verify that, for any pair of disjoint sets of items, $S_1$ and $S_2$, we have $\mycore(S_1 \cup S_2) = \mycore( \mycore(S_1) \cup \mycore(S_2))$. Therefore, 
\begin{eqnarray*}
\mycore ( {\tilde S}^+ ) & = & \mycore ( \tilde{S} \cup Q ) \\
& = & \mycore ( \mycore ( \tilde{S} ) \cup \mycore(Q) ) \\
& = & \mycore ( \mycore(S \setminus Q) \cup \mycore( Q)) \\
& = & \mycore(S) \\
& = & {\cal Q}_k \ ,
\end{eqnarray*}
where the second equality follows by noting that $\tilde{S}$ and $Q$ are disjoint, and similarly, the fourth equality holds since $S \setminus Q$ and $Q$ are clearly disjoint.
\end{enumerate}

\subsection{Proof of Lemma~\ref{lem:descretize_psi}} \label{app: proof_lem_descretize_psi}

Let us consider the sequence of states traversed by the dynamic program $F$, as it arrives to the optimal state $(K, \psi_K^*, {\cal Q}^*_K)$; the latter is ``optimal'' in the sense that $\psi_K^* = \psi^*$ and $F(K, \psi_K^*, {\cal Q}^*_K) < \infty$. This sequence, along with the specific parameters and the bulky pair corresponding to each state will be designated by: 
\begin{eqnarray*}
& & \MyAbove{ (0, \psi^*_0, {\cal Q}^*_0) }{ (S^*_0, \pi_{S^*_0}) } \xrightarrow[Q_1^*, \pi_{ Q_1^* }]{} \MyAbove{ (1, \psi^*_1, {\cal Q}^*_1) }{ (S^*_1, \pi_{S^*_1}) } \xrightarrow[Q_2^*, \pi_{ Q_2^* }]{} \MyAbove{ (2, \psi^*_2, {\cal Q}^*_2) }{ (S^*_2, \pi_{S^*_2}) } \xrightarrow[\cdots\cdots]{} \cdots \xrightarrow[Q_k^*, \pi_{ Q_k^* }]{} \MyAbove{ (k, \psi^*_k, {\cal Q}^*_k) }{ (S^*_k, \pi_{S^*_k}) } \\
& & \qquad \qquad \qquad \qquad \qquad \qquad \qquad \qquad \qquad \qquad \qquad \qquad  \xrightarrow[\cdots\cdots]{} \cdots \xrightarrow[Q_K^*, \pi_{ Q_K^* }]{} \MyAbove{ (K, \psi^*_K, {\cal Q}^*_K) }{ (S^*_K, \pi_{S^*_K}) } \ .
\end{eqnarray*}
To better understand this illustration, we note that for every $k \in [K]$, the collection of items $Q^*_k$ and their internal permutation $\pi_{Q^*_k}$ are precisely those by which the dynamic program $F$ transitions from state $(k-1, \psi^*_{k-1}, {\cal Q}^*_{k-1})$ to state $(k, \psi^*_k, {\cal Q}^*_k)$. Consequently, the resulting item set is $S^*_k = S^*_{k-1} \uplus Q^*_k$, whereas the resulting permutation $\pi_{S^*_k}$ is obtained by appending $\pi_{Q^*_k}$ to $\pi_{S^*_{k-1}}$. In addition, for the starting state, we have $\psi^*_0 = 0$ and ${\cal Q}^*_0  = \emptyset$. 

To prove the desired claim, we argue that one feasible sequence of states that can be traversed by the approximate program $\tilde{F}$ is obtained when each profit parameter $\psi_k^*$ is substituted by $\tilde{\psi}_k = \lceil \psi_k^* - \min\{k,|S_k^*|\} \cdot \frac{ \eps p_{\max} }{ n }  \rceil_{ {\cal D}_{\psi} }$. Here, the operator $\lceil \cdot \rceil_{ {\cal D}_{\psi} }$ rounds its argument up to the nearest value in ${\cal D}_{\psi}$. In other words, as shown in Claim~\ref{clm:tilde_psi_feasible} below, we prove that 
\begin{eqnarray*}
& & \MyAbove{ (0, \tilde{\psi}_0, {\cal Q}^*_0) }{ (S^*_0, \pi_{S^*_0}) } \xrightarrow[Q_1^*, \pi_{ Q_1^* }]{} \MyAbove{ (1, \tilde{\psi}_1, {\cal Q}^*_1) }{ (S^*_1, \pi_{S^*_1}) } \xrightarrow[Q_2^*, \pi_{ Q_2^* }]{} \MyAbove{ (2, \tilde{\psi}_2, {\cal Q}^*_2) }{ (S^*_2, \pi_{S^*_2}) } \xrightarrow[\cdots\cdots]{} \cdots \xrightarrow[Q_k^*, \pi_{ Q_k^* }]{} \MyAbove{ (k, \tilde{\psi}_k, {\cal Q}^*_k) }{ (S^*_k, \pi_{S^*_k}) } \\
& & \qquad \qquad \qquad \qquad \qquad \qquad \qquad \qquad \qquad \qquad \qquad \qquad  \xrightarrow[\cdots\cdots]{} \cdots \xrightarrow[Q_K^*, \pi_{ Q_K^* }]{} \MyAbove{ (K, \tilde{\psi}_K, {\cal Q}^*_K) }{ (S^*_K, \pi_{S^*_K}) } 
\end{eqnarray*}
forms a feasible sequence of states, action parameters, and bulky pairs for $\tilde{F}$. That is, we have $(S^*_k, \pi_{S^*_k}) \in \widetilde{\mybulky}(k, \tilde{\psi}_k, {\cal Q}^*_k)$, for every $k \in [K]_0$. In light of this result, we conclude in particular that  
$\tilde{F}(K, \tilde \psi_K, {\cal Q}^*_K) < \infty$ with
\begin{eqnarray*}
\tilde{\psi}_K & = & \left\lceil \psi_K^* - \min\{K,|S_K^*|\} \cdot \frac{ \eps p_{\max} }{ n }  \right\rceil_{ {\cal D}_{\psi} } \\
& \geq & \psi^* - \eps p_{\max} \\
& \geq & (1 - \eps) \cdot \psi^* \ .
\end{eqnarray*}
Here, the first inequality holds since $\psi_K^* = \psi^*$ and $|S_K^*| \leq n$. To understand the second inequality, note that for every item $i \in [n]$, the pair that consists of introducing this item and nothing more is necessarily bulky. Indeed, as a result, the completion time of item $i$ would fall within the interval ${\cal I}_{k_i}$, where $k_i$ is the unique integer for which $(1 + \eps)^{k_i-1} < w_i \leq (1 + \eps)^{k_i}$. However, since $w_i > (1 + \eps)^{k_i-1} \geq \eps^2 \cdot (1+ \eps)^{k_i}$ for $\eps \leq \frac{1}{2}$, it follows that item $i$ is $k$-heavy, implying in turn that the pair in question is bulky. Now, noting that this pair guarantees a profit of $\max \{ p_{it}: t \in [T] \text{ and } w_i \leq W_t \}$, any such expression provides a lower bound on $\psi^*$, meaning that $\psi^* \geq \max \{ p_{it}: i \in [n], t \in [T], \text{ and } w_i \leq W_t \} = p_{\max}$.

\begin{claim} \label{clm:tilde_psi_feasible}
$(S^*_k, \pi_{S^*_k}) \in \widetilde{\mybulky}(k, \tilde{\psi}_k, {\cal Q}^*_k)$, for every $k \in [K]_0$. 
\end{claim}
\begin{proof}
We first note that the parameter $\tilde \psi_k$ is indeed well-defined for all $k \in [K]_0$, since $\tilde{\psi}_k \leq \lceil \psi^*_K \rceil_{ {\cal D}_{\psi} } \leq n p_{\max} = \max {\cal D}_\psi$. Given this observation, we proceed to prove the claim by induction on $k$.

In the base case of $k=0$, the claim trivially holds since $\tilde{\psi}_0 = 0$, ${\cal Q}^*_0  = \emptyset$, $S^*_0 = \emptyset$, and $\pi_{S^*_0}$ is the empty permutation. In the general case of $k \geq 1$, to argue that $(S^*_k, \pi_{S^*_k}) \in \tilde{B}(k, \tilde{\psi}_k, {\cal Q}^*_k)$, we consider two scenarios, depending on whether $Q_k^*$ is empty or not:
\begin{itemize}
\item {\em Case 1: $Q_k^* = \emptyset$.} We first observe that, since $S^*_k = S^*_{k-1} \cup Q^*_{k}$, we have $S^*_k = S^*_{k-1}$ by the case hypothesis,  implying in turn that $\psi^*_k = \psi^*_{k-1}$ and ${\cal Q}^*_{k} = {\cal Q}^*_{k-1}$. Consequently, 
\begin{eqnarray*}
\tilde \psi_k & = & \left\lceil \psi_k^* - \min\{k,|S_k^*|\} \cdot \frac{ \eps p_{\max} }{ n }  \right\rceil_{ {\cal D}_{\psi} } \\
& \leq & \left\lceil \psi_{k-1}^* - \min\{k-1,|S_{k-1}^*|\} \cdot \frac{ \eps p_{\max} }{ n } \right\rceil_{ {\cal D}_{\psi} } \\
& = & \tilde \psi_{k-1} \ .
\end{eqnarray*}
and it follows that $\widetilde{\mybulky}(k, \tilde{\psi}_k, {\cal Q}^*_k) \supseteq  \widetilde{\mybulky}(k, \tilde{\psi}_{k-1}, {\cal Q}^*_{k-1}) \supseteq \widetilde{\mybulky}(k-1, \tilde{\psi}_{k-1}, {\cal Q}^*_{k-1})$,
where the first inclusion holds since $\tilde \psi_k \leq \tilde \psi_{k-1}$ and ${\cal Q}_k^* = {\cal Q}^*_{k-1}$. Thus, $(S^*_k, \pi_{S^*_k}) = (S^*_{k-1}, \pi_{S^*_{k-1}}) \in \widetilde{\mybulky}({k-1}, \tilde{\psi}_{k-1}, {\cal Q}^*_{k-1}) \subseteq \widetilde{\mybulky}(k, \tilde{\psi}_k, {\cal Q}^*_k)$, where the middle transition is precisely our induction hypothesis.

\item {\em Case 2: $Q_k^* \neq \emptyset$.} In this case, $|S^*_{k}| = |S^*_{k-1}| + |Q_k^*| \geq |S^*_{k-1}|+1$, as $S^*_k$ is the disjoint union of $S^*_{k-1}$ and $Q_k^*$. By the inductive hypothesis, $(S^*_{k-1}, \pi_{S^*_{k-1}}) \in \widetilde{\mybulky}(k-1, \tilde \psi_{k-1}, {\cal Q}^*_{k-1}) $, meaning that for the purpose of proving $(S^*_k, \pi_{S^*_k}) \in \widetilde{\mybulky}(k, \tilde{\psi}_k, {\cal Q}^*_k)$, it suffices to show that $\tilde \psi_{k-1} + \sum_{i \in Q^*_{k}} \varphi_{\pi_{S^*_{k}}}(i) \geq \tilde \psi_{k}$. We establish the latter inequality by noting that 
\begin{eqnarray*}
\tilde \psi_{k-1} + \sum_{i \in Q^*_{k}}  \varphi_{\pi_{S^*_{k}}}(i) & = & \tilde \psi_{k-1} + \psi^*_{k} - \psi^*_{k-1} \\
& \geq & \left( \psi^*_{k-1} - \min \{k-1, |S_{k-1}^*| \} \cdot \frac{\eps p_{\max}}{n} \right) + \psi^*_{k} - \psi^*_{k-1} \\
& \geq &  \left\lceil \psi_{k}^* - \min\{k,|S_{k}^*|\} \cdot \frac{ \eps p_{\max} }{ n }  \right\rceil_{ {\cal D}_{\psi} } \\
& = & \tilde \psi_{k}\ ,
\end{eqnarray*}
where the first equality holds since $\psi^*_{k} = \psi^*_{k-1} + \sum_{i \in Q^*_{k}}  \varphi_{\pi_{S^*_{k}}}(i)$, by the optimality of $\psi^*_{k}$.
\end{itemize}
\end{proof}

\subsection{Proof of Lemma~\ref{lem:ip-to-perm}} \label{app:proof_lem_ip-to-perm}

In order to construct the required permutation, for every $k \in [K-1]$, let $\pi_k$ be an arbitrary permutation of the items that were assigned by $x$ to bucket ${\cal B}_k$, i.e., $\{ i \in [n] : x_{ik} = 1 \}$. In addition, let $\pi_{-}$ be an arbitrary permutation of the remaining items, i.e., those that were not to assigned to any bucket. The permutation $\pi_x$ is now defined by concatenating these permutations in order of increasing index, with $\pi_{-}$ appended at the end, namely, $\pi_x = \langle \pi_1, \ldots, \pi_{K-1}, \pi_{-} \rangle$. It is easy to verify that this construction can be implemented in $O(nK)$ time.

To obtain a lower bound of $\sum_{i \in [n]} \sum_{k \in [K-1]: i \in L_{k+1}} q_{ik} x_{ik}$ on the profit of this permutation, $\newobj( \pi_x ) = \sum_{i \in [n]} \varphi_{ \pi_x }( i )$, note that since each item is assigned to at most one bucket, it suffices to show that for every $i \in [n]$ and $k \in [K-1]$ with $x_{ik} = 1$, we necessarily have $\varphi_{ \pi_x }( i ) \geq q_{ik}$. For this purpose, we observe that
\begin{eqnarray*}
\varphi_{ \pi_x }( i ) & = & \max \left\{ p_{i,t} : t \in [T+1] \text{ and } W_t \geq C_{ \pi_x }( i ) \right\} \\
& \geq & \max \left\{ p_{i,t} : t \in [T+1] \text{ and } W_t \geq (1 + \eps)^k \right\} \\
& = & q_{ik} \ ,
\end{eqnarray*}
where the inequality above holds since $C_{ \pi_x }( i ) \leq (1 + \eps)^k$. Indeed, this bound on the completion time of item $i$ can be derived by observing that every item $j$ that appears before $i$ in the permutation $\pi_x$ (i.e., $\pi_x(j) < \pi_x(i)$) was assigned by the solution $x$ to one of the buckets ${\cal B}_1, \ldots, {\cal B}_k$, and therefore,
\begin{eqnarray*}
C_{ \pi_x }( i ) & = & \sum_{j \in [n] : \pi_x(j) \leq \pi_x(i)} w_j \\
& \leq & \sum_{\kappa \in [k]} \sum_{j \in L_{\kappa+1}} w_j x_{j\kappa} \\
&\leq & \sum_{\kappa \in [k]} \mycap( {\cal B}_{\kappa} ) \\
& = & \sum_{\kappa \in [k]} \left( (1 + \eps)^{\kappa} - (1 + \eps)^{ \kappa-1 } \right) \\
& \leq & (1 + \eps)^k \ ,
\end{eqnarray*}
where the second inequality follows from the second constraint of~\eqref{eqn:IP-formulation}.

\section{Additional Proofs from Section~\ref{sec:qptas-one}}

\subsection{Proof of Lemma~\ref{lem:union-chain}} \label{app:proof_lem_union-chain}

Clearly, ${\cal R}\cup{\cal G}$ is a chain for ${\cal I}$, as each of ${\cal R}$ and ${\cal G}$ is such a chain by itself. To verify the feasibility of ${\cal R}\cup{\cal G}$, note that for any time period $t \in [T]$, since ${\cal R}$ is feasible for ${\cal I}^{-{\cal G}}$ we have
\begin{eqnarray*}
w(R_t) & \leq & W^{-{\cal G}}_t \\
& = & \min_{t \leq \tau\leq T} \left( W_{\tau} -w({G}_{\tau}) \right) \\
& \leq & W_t - w(G_t) \ .
\end{eqnarray*}
By recalling that $G_1 \subseteq \cdots \subseteq G_T$ and $R_1 \subseteq \cdots \subseteq R_T \subseteq {\cal N}^{-{\cal G}}={\cal N}\setminus G_T$, it follows in particular that $G_t$ and $R_t$ are disjoint, implying in turn that $w( R_t \cup G_t ) = w( R_t ) + w( G_t ) \leq W_t$ as required.

Now, to account for the profit of ${\cal R}\cup{\cal G}$, we conclude that
\begin{eqnarray*}
\objfunc({\cal R}\cup{\cal G}) & = &  \sum_{t \in [T]} \sum_{i \in (R_t \cup G_t) \setminus (R_{t-1} \cup G_{t-1})}  p_{it} \\ 
& =  & \sum_{t \in [T]} \left( \sum_{i \in R_t \setminus R_{t-1}} p_{it} + \sum_{i \in G_t \setminus G_{t-1}} p_{it} \right) \\
& = & \objfunc({\cal R}) + \objfunc({\cal G}) \ .
\end{eqnarray*}
Here, the second equality holds again due to the observation above, since having both $G_1 \subseteq \cdots \subseteq G_T$ and $R_1 \subseteq \cdots \subseteq R_T \subseteq {\cal N}\setminus G_T$ means that $(R_t \cup G_t) \setminus (R_{t-1} \cup G_{t-1})$ can be written as the disjoint union of $R_t \setminus R_{t-1}$ and $G_t \setminus G_{t-1}$.

\subsection{Proof of Lemma~\ref{lem:residual-chain}} \label{app:proof_lem_residual_chain}

For convenience, let us denote the chain in question by ${\cal R}={\cal S}|_{{\cal N} \setminus G}$. By observing that $R_T = (S_T \cap ({\cal N} \setminus G) )=(S_T \setminus G_T) \subseteq {\cal N}\setminus G_T$, it follows that ${\cal R}$ is also a chain for ${\cal I}^{-{\cal G}}$. We proceed by arguing that ${\cal R}$ is in fact feasible for the latter instance. To this end, note that for every $t \leq \tau$,
\begin{eqnarray*}
w(R_t) & \leq & w(R_{\tau}) \\
& =& w(S_{\tau})-w(G_{\tau}) \\
&\leq & W_{\tau} - w(G_{\tau}) \ ,
\end{eqnarray*}
where the middle equality follows by recalling that $S_t$ is the disjoint union of $G_t$ and $R_t$, and the last inequality is implied by the feasibility of ${\cal S}$ for ${\cal I}$. As a result, $w(R_t) \leq  \min_{t \leq \tau\leq T} (W_{\tau} -w({G}_{\tau}) ) = W^{-{\cal G}}_t$, which proves that ${\cal R}$ is a feasible chain for ${\cal I}^{-{\cal G}}$. 

We now turn our attention to showing that $\objfunc({\cal R}) = \objfunc({\cal S}) - \objfunc({\cal G})$. Again, based on the observation that $S_t$ is the disjoint union of $G_t$ and $R_t$ for every $t \in [T]$, we conclude that
\begin{eqnarray*}
\objfunc({\cal R}) + \objfunc({\cal G}) & = & \sum_{t \in [T]} \left( \sum_{i \in R_t \setminus R_{t-1}} p_{it} + \sum_{i \in G_t \setminus G_{t-1}} p_{it} \right) \\ 
& = &  \sum_{t \in [T]} \sum_{i \in (R_t \cup G_t) \setminus (R_{t-1} \cup G_{t-1})}  p_{it} \\
& = & \sum_{t \in [T]} \sum_{i \in S_t \setminus S_{t-1}} p_{it} \\
& = & \objfunc({\cal S}) \ .
\end{eqnarray*}
Finally, suppose that ${\cal S}$ is optimal for ${\cal I}$, but on the other hand, ${\cal R}$ is not optimal for ${\cal I}^{-{\cal G}}$, meaning that there exists a feasible chain ${\cal R}'$ for ${\cal I}^{-{\cal G}}$ with profit $\objfunc({\cal R}')>\objfunc({\cal R})$. Then, by Lemma~\ref{lem:union-chain}, we infer that ${\cal R}'\cup{\cal G}$ is a feasible chain for ${\cal I}$, with profit $\objfunc( {\cal R}'\cup{\cal G} ) = \objfunc({\cal G})+\objfunc({\cal R}')>\objfunc({\cal G})+\objfunc({\cal R})=\objfunc({\cal S})$, contradicting the optimality of ${\cal S}$.

\subsection{Proof of Lemma~\ref{lem:heavy_bound}} \label{app:proof_lem_heavy_bound}

We say that an interval ${\cal I}_k$ is non-empty with respect to the permutation $\pi_{{\cal S}^*}$ if it contains the completion time of at least one item. Note that, since the latter completion time is within $[w_{\min}, n w_{\max}]$ and we assume that $w_{\min} = 3$ (see Section~\ref{subsec:decomp_outline}), the interval ${\cal I}_0 = [0,1]$ is clearly empty. Furthermore, any non-empty interval ${\cal I}_k = ((1+\eps)^{k-1}, (1+ \eps)^k]$ necessarily has $\lfloor \log_{1 + \eps} (w_{\min}) \rfloor \leq k \leq \lceil \log_{1+ \eps}(n w_{\max}) \rceil$. Therefore, the number of non-empty intervals with respect to $\pi_{{\cal S}^*}$ is at most $\lceil \log_{1 + \eps}(n w_{\max}) \rceil - \lfloor \log_{1 + \eps}(w_{\min}) \rfloor + 1\leq 2 \cdot \lceil \log_{1 + \eps} (n \rho) \rceil$. Now, any such interval ${\cal I}_k$ is of length $(1+\eps)^{k} - (1+ \eps)^{k-1}$, meaning that the number of $k$-heavy items with a completion time in this interval is at most $\frac{ (1 + \eps)^k - (1 + \eps)^{k-1} }{ \eps^2 \cdot (1 + \eps)^k } \leq \frac{ 1 }{ \eps }$, as every $k$-heavy item has a weight of at least $\eps^2 \cdot (1 + \eps)^k$. All in all, we have just shown that $|G^{* \myheavy}| \leq \frac{2 \cdot \lceil \log_{1 + \eps} (n \rho) \rceil}{\eps} \leq \frac{ 3 \log(n \rho) }{\eps^2}$.

\subsection{Proof of Lemma~\ref{lem:heavy_profit}} \label{app:proof_lem_heavy_profit}

For every item $i \in {\cal N}$, let $t_i$ be its insertion time with respect to the optimal chain ${\cal S}^*$. By convention, for non-inserted items (i.e., those in ${\cal N}\setminus S^*_T$), we say that their ``insertion time'' is $T+1$, with a profit of $p_{i, T+1} = 0$. As explained during the proof of Lemma~\ref{lem:reformulation}, our construction of the permutation $\pi_{{\cal S}^*}$ guarantees that $\varphi_{\pi_{{\cal S}^*}} (i) \geq p_{i,t_i}$ for every item $i \in {\cal N}$. While this inequality was established for any chain-to-permutation mapping, one can easily notice that, due to the optimality of ${\cal S}^*$, we actually have $\varphi_{\pi_{{\cal S}^*}} (i) = p_{i,t_i}$ for every $i \in {\cal N}$. Otherwise, there would have been at least one item with $\varphi_{\pi_{{\cal S}^*}} (i) > p_{i,t_i}$, implying that $\newobj(\pi_{{\cal S}^*}) > \objfunc({\cal S}^*)$. By Lemma~\ref{lem:reformulation}, the permutation $\pi_{{\cal S}^*}$ can then be mapped to a feasible chain ${\cal S}$ with $\objfunc({\cal S}) = \newobj(\pi_{{\cal S}^*}) > \objfunc({\cal S}^*)$, contradicting the optimality of ${\cal S}^*$. Thus, $\objfunc({\cal H}^*) = \sum_{i \in G^{* \myheavy}} p_{i,t_i} = \sum_{i \in G^{* \myheavy}} \varphi_{\pi_{{\cal S}^*}} (i) = \newobj_{\myheavy}(\pi_{{\cal S}^*})$.

\subsection{Proof of Lemma~\ref{lem:alpha_sequence}} \label{app:proof_lem_alpha_sequence}

We prove the lower bound $\alpha_r \geq \frac{r}{r+1} - r \delta$ by induction on $r$. For $r=0$, we have $\alpha_0 = 0$ and the claim clearly holds. Now, for $r \geq 1$,  
\begin{eqnarray*}
\alpha_r & = & \frac{1 - \delta}{2 - \alpha_{r-1}} \\
& \geq & \frac{1-\delta}{2 - (\frac{r-1}{r} - (r-1) \delta )} \\ 
& = & \frac{r(1- \delta)}{r + 1 + r(r-1) \delta}  \\
& \geq &\frac{r(1- \delta)}{(r+1)(1 + (r-1)\delta)} \\
& = &\frac{r}{r+1} \cdot \left( 1 - \frac{r \delta}{1 + (r-1) \delta}\right)\\
& \geq & \frac{r}{r+1} - r \delta.
\end{eqnarray*}

\section{Additional Proofs from Section~\ref{sec:qptas-two}}

\subsection{Proof of Lemma~\ref{lem:mod_perm_properties}} \label{app:proof_lem_mod_perm_properties}

\paragraph{Sparse $\bs{({\cal M}^-, {\cal M}^+)}$-crossing.} On the one hand, our construction guarantees that the last item in ${\cal C}_{{\cal M}^-}$ appears in position $\pi(i_{{\cal M}^-, {\cal M}^+}) - 1 + |{\cal A}^-|$ of the permutation $\bar{\pi}$. On the other hand, every item in ${\cal C}_{{\cal M}^+}$ that appears before this position necessarily belongs to ${\cal X}_{{\cal M}^-, {\cal M}^+}(\pi)$. It follows that there are at most $| {\cal X}_{{\cal M}^-, {\cal M}^+}(\pi) | = \frac{ 1 }{ \eps }$ such items, and therefore, $\mycross_{{\cal M}^-, {\cal M}^+}(\bar{\pi}) \leq \frac{ 1 }{ \eps }$.

\paragraph{Completion times.} We establish this property by considering three cases, depending on whether the item in question appears before $i_{{\cal M}^-, {\cal M}^+}$, belongs to ${\cal A}^-$, or belongs to $\bar{\cal A}^-$.
\begin{itemize}
    \item \emph{Before $i_{{\cal M}^-, {\cal M}^+}$:} For every item $i \in {\cal N}$ with $\pi(i) \leq \pi(i_{{\cal M}^-, {\cal M}^+})-1$ we clearly have $C_{\bar{\pi}}(i) = C_{\pi}(i)$, since the permutations $\bar{\pi}$ and $\pi$ are identical up to position $\pi(i_{{\cal M}^-, {\cal M}^+})-1$.
    
    \item \emph{Items in ${\cal A}^-$:} For every item $i \in {\cal A}^-$, we have $C_{\bar{\pi}}(i) \leq C_{\pi}(i)$, since the collection of items appearing before $i$ in $\bar{\pi}$ is a subset of those appearing before $i$ in $\pi$. 
    
    \item \emph{Items in $\bar{\cal A}^-$:} For every item $i \in \bar{\cal A}^-$, the important observation is that the collection of items appearing before $i$ in $\bar{\pi}$ consists of: (1)~The same items appearing before $i$ in $\pi$, except for the eliminated item $i_{{\cal M}^-, {\cal M}^+}$; as well as (2)~All items in ${\cal A}^-$ appearing after $i$ in $\pi$. Therefore, 
    \[ C_{\bar{\pi}}(i) \leq C_{\pi}(i) - w_{i_{{\cal M}^-, {\cal M}^+}} + w( {\cal A}^- ) \leq C_{\pi}(i) \ . \]
    To understand the last inequality, recall that $i_{{\cal M}^-, {\cal M}^+} \in {\cal X}_{{\cal M}^-, {\cal M}^+}(\pi)$, meaning in particular that this item resides within ${\cal C}_{{\cal M}^+}$. Since ${\cal I}=({\cal N}, W)$ is well-spaced, property~2 of such instances implies that $w_{i_{{\cal M}^-, {\cal M}^+}}$ is greater than the weight of any item in ${\cal C}_{{\cal M}^-}$ by a multiplicative factor of at least $n^{ 1 + (\min {\cal M}^+ - \max {\cal M}^-  - 1 ) / \eps } \geq n$, as $\max {\cal M}^- < \min {\cal M}^+$. Consequently, since all items in ${\cal A}^-$ reside within ${\cal C}_{{\cal M}^-}$, we indeed have $w_{i_{{\cal M}^-, {\cal M}^+}} \geq n \cdot \max_{ j \in {\cal C}_{{\cal M}^-} } w_j \geq w( {\cal A}^- )$.
\end{itemize}

\paragraph{Difference.} This property is straightforward, by construction of $\bar{\pi}$.

\subsection{Proof of Claim~\ref{cl:pi-sparse-profit}} \label{app:proof_cl_pi_sparse_profit}

For simplicity of notation, let ${\cal D} = \{ i_{{\cal M}^-, {\cal M}^+}: ({\cal M}^-, {\cal M}^+) \in \Omega, {\cal X}_{{\cal M}^-, {\cal M}^+} \neq \emptyset \}$ be the collection of items that were removed throughout all recursive calls to our fixing procedure. Then, the profit of the resulting permutation $\pi_{\mysparse}$ can be lower-bounded by observing that
\begin{eqnarray*}
\newobj(\pi_{\mysparse}) & = & \sum_{i \in {\cal N} \setminus {\cal D}} \varphi_{\pi_{\mysparse}}(i) \\
& \geq & \sum_{i \in {\cal N} \setminus {\cal D}} \varphi_{\pi^*}(i) \\
& = & \newobj(\pi^*) - \sum_{i \in {\cal D}} \varphi_{\pi^*}(i) \\
& \geq & \newobj(\pi^*) - \eps \cdot \sum_{({\cal M}^-,{\cal M}^+) \in \Omega}\varphi_{\pi^*}\left({\cal X}_{{\cal M}^-,{\cal M}^+} \left(\pi_{[\min {\cal M}^-, \max {\cal M}^+]} \right) \right).
\end{eqnarray*}
Here, the first inequality holds since, for any remaining item $i \in {\cal N} \setminus {\cal D}$, it is not difficult to verify (by induction on the recursion level) that property~2 of the fixing procedure implies $C_{\pi_{\mysparse}}(i) \leq C_{\pi^*}(i)$, and we therefore have $\varphi_{\pi_{\mysparse}}(i) \geq \varphi_{\pi^*}(i)$. The second inequality is obtained by recalling that any item $i_{{\cal M}^-, {\cal M}^+} \in {\cal D}$ was chosen as the least profitable item in ${\cal X}_{{\cal M}^-,{\cal M}^+}(\pi_{[\min {\cal M}^-, \max {\cal M}^+]}) $ with respect to $\pi^*$, thus
\begin{eqnarray*}
\varphi_{\pi^*}(i_{{\cal M}^-, {\cal M}^+}) & \leq & \frac{\varphi_{\pi^*}({\cal X}_{{\cal M}^-,{\cal M}^+}(\pi_{[\min {\cal M}^-, \max {\cal M}^+]}))}{|{\cal X}_{{\cal M}^-,{\cal M}^+}(\pi_{[\min {\cal M}^-, \max {\cal M}^+]})|} \\
& = & \eps \cdot \varphi_{\pi^*} \left( {\cal X}_{{\cal M}^-,{\cal M}^+} \left(\pi_{[\min {\cal M}^-, \max {\cal M}^+]}\right) \right) \ .
\end{eqnarray*}

\subsection{Proof of Claim~\ref{cl:disjoint-cross}} \label{app:proof_cl_disjoint_cross}

By definition, ${\cal X}_{{\cal M}^-_1, {\cal M}^+_1}({\pi}_{[\min {\cal M}^-_1, \max {\cal M}^+_1]})$ and ${\cal X}_{{\cal M}^-_2, {\cal M}^+_2}({\pi}_{[\min {\cal M}^-_2, \max {\cal M}^+_2]})$ contain only items in ${\cal M}^+_1$-indexed clusters and ${\cal M}^+_2$-indexed clusters, respectively. Thus, when ${\cal M}^+_1$ and ${\cal M}^+_2$ are disjoint, ${\cal X}_{{\cal M}^-_1, {\cal M}^+_1}({\pi}_{[\min {\cal M}^-_1, \max {\cal M}^+_1]})$ and ${\cal X}_{{\cal M}^-_2, {\cal M}^+_2}({\pi}_{[\min {\cal M}^-_2, \max {\cal M}^+_2]})$ must be disjoint as well. Hence, it remains to consider the scenario where ${\cal M}^+_1$ and ${\cal M}^+_2$ are not disjoint. In this case, the permutations $\pi_{[\min {\cal M}^-_1, \max {\cal M}^+_1]}$ and ${\pi}_{[\min {\cal M}^-_2, \max {\cal M}^+_2]}$ must have been created at different levels of the recursive construction; we assume without loss of generality that ${\pi}_{[\min {\cal M}^-_1, \max {\cal M}^+_1]}$ was created at a lower-index level. Therefore, ${\cal M}^+_2 \subseteq {\cal M}^+_1$, and ${\cal X}_{{\cal M}^-_2, {\cal M}^+_2}({\pi}_{[\min {\cal M}^-_2, \max {\cal M}^+_2]})$ consists of only items in the right permutation, $\pi_{[\min {\cal M}^+_1, \max {\cal M}^+_1]}$. On the other hand, by construction, any item in ${\cal X}_{{\cal M}^-_1, {\cal M}^+_1}({\pi}_{[\min {\cal M}^-_1, \max {\cal M}^+_1]})$ ends up in the left permutation, ${\pi}_{[\min {\cal M}^-_1, \max {\cal M}^-_1]}$, implying the disjointness of ${\cal X}_{{\cal M}^-_1, {\cal M}^+_1}({\pi}_{[\min {\cal M}^-_1, \max {\cal M}^+_1]})$ and ${\cal X}_{{\cal M}^-_2, {\cal M}^+_2}(\pi_{[\min {\cal M}^-_2, \max {\cal M}^+_2]})$. 

\subsection{Proof of Lemma~\ref{lem:dp-feasibility-subproblem}} \label{app:proof_lem_dp_feasibility_subproblem}

We first observe that the pair $(\hat{S}, \hat{\pi})$ is indeed thin. To this end, note that since the permutation $\hat{\pi}$ is a prefix of $\pi$, for every $m \in [M]$ we clearly have $\mycross_m(\hat{\pi}) \leq \mycross_m(\pi) \leq \frac{\lceil \log_2 M \rceil}{\eps}$, where the last inequality holds since $(S,\pi)$ is thin. Next, we show that $(\hat{S}, \hat{\pi})$ satisfies conditions~1-3:
\begin{enumerate}
    \item {\em Allowed items}: By construction, $\hat{S} = S \cap ({\cal C}_{[1,m-1]} \uplus {\cal Q}_{>m-1})$, implying that $\hat{S}$ forms a subset of ${\cal C}_{[1,m-1]} \uplus {\cal Q}_{>m-1}$. 
    
    \item {\em Required crossing items}: An additional implication of our definition of $\hat{S}$ is that ${\cal Q}_{>m-1} \subseteq \hat{S}$, since ${\cal Q}_{>m-1} \subseteq S$ by~\eqref{eqn:def_qmmin1}.

    \item {\em Total profit}: To obtain a lower bound on the profit of $\hat{\pi}$, we observe that  
    \begin{eqnarray*}
    \newobj(\hat{\pi}) & = & \sum_{i \in \hat{S}} \varphi_{\hat{\pi}}(i) \\
    & = & \sum_{i \in \hat{S}} \varphi_{\pi}(i) \\
    & = & \newobj(\pi) - \sum_{i \in S \setminus ({\cal C}_{[1,m-1]} \uplus {\cal Q}_{>m-1})} \varphi_{\pi}(i) \\
    & \geq & \left[\psi_m - \sum_{i \in S \setminus ({\cal C}_{[1,m-1]} \uplus {\cal Q}_{>m-1})} \varphi_{\pi}(i) \right]^+ \\
    & = & \psi_{m-1} \ .
    \end{eqnarray*}
    Here, the second equality holds since $\hat{\pi}$ is a prefix of $\pi$, as previously mentioned. The third equality follows by noting that $S \setminus \hat{S} = S \setminus ({\cal C}_{[1,m-1]} \uplus {\cal Q}_{>m-1})$. The inequality above is obtained by observing that its left-hand-side is non-negative, and by recalling that $(S, \pi) \in \mythin(m, \psi_m, {\cal Q}_{>m})$, implying that  $\newobj(\pi) \geq \psi_m$. The last equality is precisely the definition of $\psi_{m-1}$.
\end{enumerate}

\subsection{Proof of Lemma~\ref{lem:qptas-dp-substructure}} \label{app:proof_lem_qptas_dp_substructure}

By way of contradiction, suppose there exists a pair $(\tilde{S}, \tilde{\pi}) \in \mythin(m-1, \psi_{m-1}, {\cal Q}_{>m-1})$ whose makespan is smaller than that of $\hat{S}$, namely, $w(\tilde{S}) < w(\hat{S})$. We begin by noticing that the item sets
$S \setminus \hat{S}$ and $\tilde{S}$ are disjoint, since $S \setminus \hat{S} \subseteq ({\cal C}_m \uplus {\cal Q}_{>m} ) \setminus {\cal Q}_{>m-1} \subseteq {\cal C}_{[m,M]} \setminus {\cal Q}_{>m-1}$ whereas $\tilde{S} \subseteq {\cal C}_{[1,m-1]} \uplus {\cal Q}_{>m-1}$, as $(\tilde{S}, \tilde{\pi}) \in \mythin(m-1, \psi_{m-1}, {\cal Q}_{>m-1})$. Taking advantage of this observation, we define a new pair $(\tilde{S}^+, \tilde{\pi}^+)$ as follows:
\begin{itemize}
    \item The underlying set of items is given by $\tilde{S}^+ = \tilde{S} \uplus (S \setminus \hat{S})$.
    
    \item The permutation $\tilde{\pi}^+: \tilde{S}^+ \to [|\tilde{S}^+|]$ is constructed by appending the items in $S \setminus \hat{S}$ to $\tilde{\pi}$, following their internal order in $\pi$. 
\end{itemize}
The next claim shows that the resulting pair is a feasible solution to exactly the same subproblem for which $(S, \pi)$ is optimal. 

\begin{claim} \label{cl:feasible-S}
$(\tilde{S}^+, \tilde{\pi}^+) \in \mythin(m, \psi_{m}, {\cal Q}_{>m})$.
\end{claim}
\begin{proof}
First, we show that $(\tilde{S}^+, \tilde{\pi}^+)$ is a thin pair. To this end, for every $\mu \in [M]$ with ${\cal C}_{\mu} \cap \tilde{S}^+ \neq \emptyset$, let $i_{\mu} \in {\cal C}_{\mu}$ be the item that appears last in $\tilde{\pi}^+$ out of this cluster, i.e., $i_{\mu} = \argmax_{i \in \tilde{S}^+ \cap {\cal C}_{\mu}} \tilde{\pi}^+(i)$. We proceed by considering two cases:
\begin{itemize}
    \item {\em Item $i_{\mu}$ appears in $\tilde{\pi}$}: By construction, $\tilde{\pi}$ is a prefix of $\tilde{\pi}^+$, and therefore $\mycross_{\mu}(\tilde{\pi}^+) = \mycross_{\mu}(\tilde{\pi}) \leq \frac{\lceil \log_2 M \rceil}{\eps}$, where the last inequality holds since $(\tilde{S}, \tilde{\pi})$ is a thin pair.
    
    \item {\em Item $i_{\mu}$ does not appear in $\tilde{\pi}$}: In this case, $i_{\mu} \in S \setminus \hat{S} \subseteq {\cal C}_{[m,M]} \setminus {\cal Q}_{>m-1}$, implying that $\mu \geq m$. Thus, all items in clusters ${\cal C}_{\mu+1}, \ldots, {\cal C}_M$ that appear before $i_{\mu}$ in the permutation $\tilde{\pi}^+$ necessarily belong to ${\cal Q}_{>m}$, and we conclude that $\mycross_{\mu}(\tilde{\pi}^+) \leq |{\cal Q}_{>m}| \leq \frac{\lceil \log_2 M \rceil}{\eps}$. 
\end{itemize}
Next, we show that $(\tilde{S}^+, \tilde{\pi}^+)$ satisfies conditions~1-3: 
\begin{enumerate}
    \item {\em Allowed items}: First note that $\tilde{S} \subseteq {\cal C}_{[1,m-1]} \uplus {\cal Q}_{>m-1} \subseteq {\cal C}_{[1,m]} \uplus {\cal Q}_{>m}$, where the first inclusion holds since $(\tilde{S}, \tilde{\pi}) \in \mythin(m-1, \psi_{m-1}, {\cal Q}_{>m-1})$ and the second follows by definition of ${\cal Q}_{>m-1}$ in~\eqref{eqn:def_qmmin1}. In addition, $S \subseteq {\cal C}_{[1,m]} \uplus {\cal Q}_{>m}$, since $(S, \pi) \in \mythin(m, \psi_m, {\cal Q}_{>m})$. Combining these two observations, we have $\tilde{S}^+ = \tilde{S} \uplus (S \setminus \hat{S}) \subseteq {\cal C}_{[1,m]} \uplus {\cal Q}_{>m}$ as required.
    
    \item {\em Required crossing items}: To prove ${\cal Q}_{>m} \subseteq \tilde{S}^+$, we observe that  
    \begin{eqnarray*}
    {\cal Q}_{>m} & \subseteq & {\cal Q}_{>m-1} \uplus ({\cal Q}_{>m} \setminus {\cal Q}_{>m-1}) \\
    & \subseteq & \tilde{S} \cup (S \setminus \hat{S}) \\
    & = & \tilde{S}^+ \ .
    \end{eqnarray*}
    To better understand the second inclusion, note that ${\cal Q}_{>m-1} \subseteq \tilde{S}$, since $(\tilde{S}, \tilde{\pi}) \in \mythin(m-1, \psi_{m-1}, {\cal Q}_{>m-1})$. In addition, ${\cal Q}_{>m} \setminus {\cal Q}_{>m-1} \subseteq S \setminus \hat{S}$, since ${\cal Q}_{>m} \subseteq S$ due to having $(S,\pi) \in \mythin(m, \psi_{m}, {\cal Q}_{>m})$, and since $({\cal Q}_{>m} \setminus {\cal Q}_{>m-1}) \cap \hat{S} = \emptyset$, due to having ${\cal Q}_{>m} \subseteq {\cal C}_{[m+1,M]}$ and $\hat{S} \subseteq {\cal C}_{[1,m-1]} \uplus {\cal Q}_{>m-1}$, where the latter inclusion holds since $\hat{S} \in \mythin(m-1, \psi_{m-1}, {\cal Q}_{>m-1})$. 
    
    \item {\em Total profit}: By construction, any item $i \in S \setminus \hat{S}$ appears in the permutation $\tilde{\pi}^+$ after all items in $\tilde{S}$, and moreover, the internal order between the items in $S \setminus \hat{S}$ is determined according to $\pi$. Hence, we can bound the completion time of any item $i \in S \setminus \hat{S}$ by noting that
    \begin{eqnarray*}
     C_{\tilde{\pi}^+}(i) & = & w ( \tilde{S} ) + \sum_{\MyAbove{j \in S \setminus \hat{S}:}{ \pi(j)< \pi(i)}} w_j  \\
     & < & w ( \hat{S} ) + \sum_{\MyAbove{j \in S \setminus \hat{S}:}{ \pi(j)< \pi(i)}}w_j  \\
     & = & C_{\pi}(i) \ , 
    \end{eqnarray*}
    where the inequality above follows from our initial assumption that $w(\tilde{S}) < w(\hat{S})$. Consequently, $\varphi_{\tilde{\pi}^+}(i) \geq \varphi_{\pi}(i)$ for such items, and we have
    \begin{eqnarray}
    \newobj \left( \tilde{\pi}^+ \right) &=& \newobj \left( \tilde{\pi} \right) + \sum_{i \in S \setminus \hat{S}} \varphi_{\tilde{\pi}^+}(i) \label{eqn:pair_profit_1} \\
    & \geq & \psi_{m-1} + \sum_{i \in S \setminus \hat{S}}\varphi_{{\pi}}(i) \label{eqn:pair_profit_2} \\
    & = & \left[ \psi_m - \sum_{i \in S \setminus ( {\cal C}_{[1,m-1]} \uplus {\cal Q}_{>m-1})} \varphi_{\pi}(i) \right]^+ + \sum_{i \in S \setminus \hat{S}}\varphi_{{\pi}}(i) \label{eqn:pair_profit_3} \\
    & \geq & \psi_m - \left( \sum_{i \in S \setminus ({\cal C}_{[1,m-1]} \uplus {\cal Q}_{>m-1})} \varphi_{\pi}(i) - \sum_{i \in S \setminus \hat{S}}\varphi_{{\pi}}(i) \right) \nonumber\\
    & = & \psi_m \ . \label{eqn:pair_profit_4}
    \end{eqnarray}
    Here, equality~\eqref{eqn:pair_profit_1} holds since $\tilde{\pi}$ is a prefix of $\tilde{\pi}^+$, with the items in $S \setminus \hat{S}$ forming the remaining suffix. Inequality~\eqref{eqn:pair_profit_2} holds since $(\tilde{S}, \tilde{\pi}) \in \mythin(m-1, \psi_{m-1}, {\cal Q}_{m-1})$, meaning that $\newobj( \tilde{\pi} ) \geq \psi_{m-1}$, and since $\varphi_{\tilde{\pi}^+}(i) \geq \varphi_{\pi}(i)$ for all $i \in S \setminus \hat{S}$, as shown above. Equality~\eqref{eqn:pair_profit_3} follows from the definition of $\psi_{m-1}$. Equality~\eqref{eqn:pair_profit_4} is obtained by noting that $S \setminus ( {\cal C}_{[1,m-1]} \uplus {\cal Q}_{>m-1}) = S \setminus \hat{S}$.
\end{enumerate}
\end{proof}

Consequently, by combining our initial assumption that $w(\tilde{S}) < w(\hat{S})$ along with Claim~\ref{cl:feasible-S}, we have just identified a pair $(\tilde{S}^+, \tilde{\pi}^+) \in \mythin(m, \psi_{m}, {\cal Q}_{>m})$ with a makespan of 
\begin{eqnarray*}
w(\tilde{S}^+) & = & w(\tilde{S}) + w(S \setminus \hat{S}) \\
& < & w(\hat{S}) + w(S \setminus \hat{S}) \\
& = & w(S) \ ,
\end{eqnarray*}
contradicting the fact that $(S, \pi)$ minimizes $w(S)$ over the set $\mythin(m, \psi_m, {\cal Q}_{>m})$.

\subsection{Proof of Lemma~\ref{lem:approx_rec_eqn}} \label{app:proof_lem_approx_rec_eqn}

\paragraph{Overview.} Prior to delving into the nuts-and-bolts of our approach, we provide a high-level overview of its main ideas. For this purpose, to make sure condition~2 of Lemma~\ref{lem:approx_rec_eqn} is satisfied, meaning that the item set $\hat{\cal E}$ we compute has a total weight of at most $F(m, \psi_m, {\cal Q}_{>m}) - F(m-1, \psi_{m-1}, {\cal Q}_{>m-1})$, our algorithm relies on ``knowing'' the latter difference, which will be justified through binary search. With this limitation, restricting ourselves to the item set $({\cal C}_m \uplus {\cal Q}_{>m}) \setminus {\cal Q}_{>m-1}$, we aim to identify a feasible chain whose associated permutation $(\eps, \Delta)$-satisfies constraint~2. To this end, our algorithm ``guesses'' the insertion time of every item in ${\cal Q}_{>m} \setminus {\cal Q}_{>m-1}$ by enumerating over all feasible chains ${\cal G} = (G_1, \dots, G_T)$ whose set of introduced items is $G_T = {\cal Q}_{>m} \setminus {\cal Q}_{>m-1}$. Since there are at most $\frac{\lceil \log_2 M \rceil }{\eps}$ such items, the number of required guesses is only $O(T^{O(\frac{\log M}{\eps})})$. For each guess, we construct the residual generalized incremental knapsack instance, as explained in Section~\ref{subsec:qptas_prelim}, which will be solved to near-optimality via the approximation scheme proposed in Theorem~\ref{thm:qptas-bounded}.

\paragraph{Algorithm.} For ease of presentation, on top of all input ingredients mentioned in Lemma~\ref{lem:approx_rec_eqn}, we feed into the upcoming algorithm an additional parameter $\omega \geq 0$, whose role will be explained later on. With this parameter, our algorithm operates as follows: 
\begin{enumerate}
    \item We define the generalized incremental knapsack instance $\hat{\cal I}^{\omega} = (\hat{\cal N}, \hat{W}^{\omega})$, where:
    \begin{itemize}
        \item The set of items $\hat{\cal N}$ is comprised of those allowed by constraint~1, namely, $\hat{\cal N} = ({\cal C}_m \uplus {\cal Q}_{>m}) \setminus {\cal Q}_{>m-1}$.
        
        \item Additionally, we reduce the capacity $W_t$ of each period $t \in [T]$ by $\Delta$, while ensuring that the maximum resulting capacity does not exceed $\omega$, meaning that $\hat{W}_t^{\omega} = \min \{ [W_t - \Delta]^+, \omega \}$.
    \end{itemize}
    
    \item For every feasible chain ${\cal G} = (G_1, \dots, G_T)$ for the instance $\hat{\cal I}^{\omega}$ with $G_T = {\cal Q}_{>m} \setminus {\cal Q}_{>m-1}$, we construct the residual instance $\hat{\cal I}^{\omega,-{\cal G}} = (\hat{\cal N}^{-{\cal G}}, \hat{W}^{\omega,-{\cal G}})$. The approximation scheme we proposed in Section~\ref{sec:qptas-one} is now applied to this instance, thereby obtaining a feasible chain ${\cal R}^{\cal G}$ whose profit is within factor $1-\eps$ of the residual optimum (see Theorem~\ref{thm:qptas-bounded}). When there are no feasible chains with $G_T = {\cal Q}_{>m} \setminus {\cal Q}_{>m-1}$, we abort and report this finding. 
    
    \item Out of all chains ${\cal G}$ considered in step~2, let ${\cal G}^{\omega}$ be the one for which the sum of profits $\objfunc({\cal G}^{\omega})+\objfunc( {\cal R}^{{\cal G}^{\omega}} )$ is maximized. The item set we return is ${\cal E}_{\omega} = R^{{\cal G}^{\omega}}_T \uplus ({\cal Q}_{>m} \setminus {\cal Q}_{>m-1})$, i.e., all items inserted by the chain ${\cal R}^{{\cal G}^{\omega}}$ along with those in ${\cal Q}_{>m} \setminus {\cal Q}_{>m-1}$. We define the corresponding permutation ${\pi}_{{\cal E}_{\omega}} : {\cal E}_{\omega} \to [|{\cal E}_{\omega}|]$ as the one constructed by Lemma~\ref{lem:reformulation} for the chain ${\cal G}^{\omega} \cup {\cal R}^{{\cal G}^{\omega}}$.
\end{enumerate}

\paragraph{The binary search.} We assume without loss of generality that all item weights take integer values. This property can easily be enforced by uniform scaling, which produces an equivalent instance whose input length is polynomial in that of the original instance. Now, knowing in advance that the total weight of any item set is an integer within $[0, w({\cal N})]$, we employ our $\omega$-parameterized algorithm to conduct a binary search over this interval, with the objective of identifying the smallest integer $\omega_{\min}$ such that:
\begin{itemize}
    \item For $\omega_{\min}$, the algorithm returns a permutation ${\pi}_{{\cal E}_{\omega_{\min}}}$ that satisfies $\sum_{i \in {{\cal E}_{\omega_{\min}}}} \varphi_{ {\pi}_{{\cal E}_{\omega_{\min}}} }^{ +\Delta }( i ) \geq (1 - \eps) \cdot ( \psi_{m} - \psi_{m-1} )$.
    
    \item In contrast, for $\omega_{\min}-1/2$, the algorithm either aborts at step~2, or returns a permutation ${\pi}_{{\cal E}_{\omega_{\min}-1/2}}$ satisfying $\sum_{i \in {{\cal E}_{\omega_{\min}-1/2}}} \varphi_{ {\pi}_{{\cal E}_{\omega_{\min}-1/2}} }^{ +\Delta }( i ) < (1 - \eps) \cdot (\psi_{m} - \psi_{m-1})$.
\end{itemize}
To verify that this search procedure is well-defined, let us examine the endpoints of $[0, w({\cal N})]$. For $\omega = 0$, if we obtain a permutation ${\pi}_{{\cal E}_0}$ that satisfies $\sum_{i \in {{\cal E}_{ 0 }}} \varphi_{ {\pi}_{{\cal E}_{ 0 }} }^{ +\Delta }( i ) \geq (1 - \eps) \cdot ( \psi_{m} - \psi_{m-1} )$, our immediate conclusion is that $\omega_{\min} = 0$. For $\omega = w({\cal N})$, as shown in Lemma~\ref{lem:profit-omega-eps-satisfied} below, we are guaranteed to obtain a permutation ${\pi}_{{\cal E}_{w({\cal N})}}$ that satisfies $\sum_{i \in {{\cal E}_{w({\cal N})}}} \varphi_{ {\pi}_{{\cal E}_{w({\cal N})}} }^{ +\Delta }( i ) \geq (1 - \eps) \cdot ( \psi_{m} - \psi_{m-1} )$. 

\paragraph{Running time.} Clearly, the number of binary search iterations we incur is linear in the input size. Now, within each iteration, since there are $O(T)$ guesses for the insertion time of every item $i \in {\cal Q}_{>m} \setminus {\cal Q}_{>m-1}$ and since $|{\cal Q}_{>m} \setminus {\cal Q}_{>m-1}| \leq \frac{ \lceil \log_2 M \rceil}{\eps}$, there are only $O(T^{O(\frac{\log M}{\eps})})$ chains ${\cal G}$ to consider in step~2. The crucial observation is that, for each such chain, the residual instance $\hat{\cal I}^{\omega,-{\cal G}}$ is defined over the set of items
\begin{eqnarray}
\hat{\cal N}^{-{\cal G}} & = & \hat{\cal N} \setminus G_T \nonumber \\
& = & (({\cal C}_m \uplus {\cal Q}_{>m}) \setminus {\cal Q}_{>m-1}) \setminus ({\cal Q}_{>m} \setminus {\cal Q}_{>m-1}) \nonumber \\
& \subseteq & {\cal C}_m \setminus {\cal Q}_{>m-1} \nonumber \\
& \subseteq & {\cal C}_m \label{eq:restricted-items} \ .
\end{eqnarray}
Thus, $\hat{\cal I}^{\omega,-{\cal G}}$ is in fact a single-cluster instance, where the weights of any two items differ by a multiplicative factor of at most $n^{ 1/\eps }$, by property~1 of well-spaced instances (see Section~\ref{subsec:qptas2-overview}). By Theorem~\ref{thm:qptas-bounded}, the running time of our approximation scheme for such instances is truly quasi-polynomial, being $O((nT)^{O(\frac{1}{\eps^{6}} \cdot \log n ) } \cdot |{\cal I}|^{O(1)})$ . All in all, we incur a running time of $O((nT)^{O(\frac{1}{\eps^6} \cdot (\log n + \log M) ) } \cdot |{\cal I}|^{O(1)})$, with room to spare.

\paragraph{Final solution and analysis.} In the remainder of this section, we argue that the item set ${\cal E}_{\omega_{\min}}$ and its permutation ${\pi}_{{\cal E}_{\omega_{\min}}} : {\cal E}_{\omega_{\min}} \to [|{\cal E}_{\omega_{\min}}|]$ satisfy the properties required by Lemma~\ref{lem:approx_rec_eqn}. For this purpose, recalling that the latter lemma assumes $F(m, \psi_m, {\cal Q}_{>m}) \leq W_T$ and $(m-1, \psi_{m-1}, {\cal Q}_{>m-1}) = \mybest(m, \psi_m, {\cal Q}_{>m})$, let ${\cal E}^*$ and $\pi^*_{{\cal E}^*} : {\cal E}^* \to [|{\cal E}^*|]$ be the item set and permutation attaining the minimum makespan $w( {\cal E}^* )$ over $\myextra[ \MyAbove{ (m, \psi_m, {\cal Q}_{>m}) }{ (m-1, \psi_{m-1}, {\cal Q}_{>m-1}) } ]$, noting that by definition,  
\begin{equation} \label{eqn:weight_best}
F \left( m, \psi_m, {\cal Q}_{>m} \right) = F \left( m-1, \psi_{m-1}, {\cal Q}_{>m-1} \right) + w \left( {\cal E}^* \right) \ .
\end{equation}
At the heart of our analysis lies the following claim, showing that whenever the $\omega$-parameterized algorithm is employed with $\omega \geq w({\cal E}^*)$, we obtain a permutation whose $\Delta$-shifted profit is at least $(1 - \eps) \cdot ( \psi_{m} - \psi_{m-1} )$. We provide the proof in Appendix~\ref{app:proof_lem_profit-omega-eps-satisfied}.

\begin{lemma} \label{lem:profit-omega-eps-satisfied}
For any $\omega \geq w({\cal E}^*)$, the $\omega$-parameterized algorithm computes an item set ${\cal E}_{\omega}$ and a permutation ${\pi}_{{\cal E}_{\omega}} : {\cal E}_{\omega} \to [|{\cal E}_{\omega}|]$ that satisfy $\sum_{i \in {{\cal E}_{\omega}}} \varphi_{ {\pi}_{{\cal E}_{\omega}} }^{ +\Delta }( i ) \geq (1 - \eps) \cdot ( \psi_{m} - \psi_{m-1} )$.
\end{lemma}

With this result in place, the properties required by Lemma~\ref{lem:approx_rec_eqn} can easily be established, as we show next. 

\begin{lemma}
The item set ${\cal E}_{\omega_{\min}}$ and  permutation ${\pi}_{{\cal E}_{\omega_{\min}}}$ satisfy properties~1 and~2.
\end{lemma}
\begin{proof}
We begin by explaining why $({\cal E}_{\omega_{\min}}, {\pi}_{{\cal E}_{\omega_{\min}}}) \in \myextra_{ \eps,\Delta }[ \MyAbove{ (m, \psi_m, {\cal Q}_{>m}) }{ (m-1, \psi_{m-1}, {\cal Q}_{>m-1}) } ]$, as stated in property~1:
\begin{itemize}
    \item {\em Constraint~1 is satisfied}: We first show that ${\cal E}_{\omega_{\min}} \subseteq ({\cal C}_m \uplus {\cal Q}_{>m}) \setminus {\cal Q}_{>m-1}$ and ${\cal Q}_{>m} \setminus {\cal Q}_{>m-1} \subseteq {\cal E}_{\omega_{\min}}$. Since the item set in question is defined in step~3 as ${\cal E}_{\omega_{\min}} = R^{{\cal G}^{\omega_{\min}}}_T \uplus ({\cal Q}_{>m} \setminus {\cal Q}_{>m-1})$, it suffices to explain why $R^{{\cal G}^{\omega_{\min}}}_T \subseteq {\cal C}_m \setminus {\cal Q}_{>m-1}$. The latter inclusion follows by noting that ${\cal R}^{{\cal G}^{\omega_{\min}}}$ is a feasible chain for the instance $\hat{\cal I}^{\omega_{\min},-{\cal G}^{\omega_{\min}}}$, where the set of items is $\hat{\cal N}^{-{\cal G}^{\omega_{\min}}} \subseteq {\cal C}_m \setminus {\cal Q}_{>m-1}$, as shown in the first inclusion of~\eqref{eq:restricted-items}.
    
    \item {\em Constraint~2 is $(\eps, \Delta)$-satisfied}: To argue that $\sum_{i \in {{\cal E}_{\omega_{\min}}}} \varphi_{ {\pi}_{{\cal E}_{\omega_{\min}}} }^{ +\Delta }( i ) \geq (1 - \eps) \cdot ( \psi_{m} - \psi_{m-1} )$, following Lemma~\ref{lem:profit-omega-eps-satisfied}, there exists a value $\omega \leq w({\cal N})$ for which $\sum_{i \in {{\cal E}_{\omega}}} \varphi_{ {\pi}_{{\cal E}_{\omega}} }^{ +\Delta }( i ) \geq (1 - \eps) \cdot ( \psi_{m} - \psi_{m-1} )$, and the desired claim is implied by the termination condition of our binary search.
\end{itemize}
We now turn our attention to proving that  $w({\cal E}_{\omega_{\min}}) \leq F(m, \psi_m, {\cal Q}_{>m}) - F(m-1, \psi_{m-1}, {\cal Q}_{>m-1})$, as stated in property~2. To this end, since $F ( m, \psi_m, {\cal Q}_{>m} ) = F ( m-1, \psi_{m-1}, {\cal Q}_{>m-1} ) + w ( {\cal E}^*)$ by equation~\eqref{eqn:weight_best}, it remains to argue that $w( {\cal E}_{\omega_{\min}} ) \leq w({\cal E}^*)$. To verify this relation, note that
\begin{eqnarray*}
w \left({\cal E}_{\omega_{\min}} \right) & = & w \left(R^{{\cal G}^{\omega_{\min}}}_T \uplus ({\cal Q}_{>m} \setminus {\cal Q}_{>m-1}) \right) \\
& = & w \left( R^{{\cal G}^{\omega_{\min}}}_T \right) + w \left({G}^{\omega_{\min}}_T \right) \\
& \leq & \hat{W}_T \\
& = & \min \left\{ [W_T - \Delta]^+, \omega_{\min} \right\} \\
& \leq & \omega_{\min} \\
& \leq & w\left({\cal E}^* \right) \ .
\end{eqnarray*}
Here, the second equality holds since ${G}^{\omega_{\min}}_T = {\cal Q}_{>m} \setminus {\cal Q}_{>m-1}$, as stated in step~2. The first inequality follows by observing that the chain ${\cal R}^{{\cal G}^{\omega_{\min}}} \cup {\cal G}^{\omega_{\min}}$ is feasible for $\hat{\cal I}^{\omega_{\min}}$, due to Lemma~\ref{lem:union-chain}, meaning in particular that for period $T$ we have $w( R^{{\cal G}^{\omega_{\min}}}_T ) + w ({G}^{\omega_{\min}}_T ) \leq \hat{W}_T$. The final inequality is derived by combining Lemma~\ref{lem:profit-omega-eps-satisfied} and the termination condition of our binary search. 
\end{proof}

\subsection{Proof of Lemma~\ref{lem:profit-omega-eps-satisfied}} \label{app:proof_lem_profit-omega-eps-satisfied}

\paragraph{Constructing a feasible chain for $\bs{\hat{\cal I}^{\omega}}$.} With respect to the item set ${\cal E}^*$ and permutation $\pi^*_{{\cal E}^*}$, let us define a chain ${\cal S}_*$ for the instance $\hat{\cal I}^{\omega}$ as follows: 
\begin{itemize}
    \item The collection of inserted items is $S_{*T} = {\cal E}^*$.
    
    \item The insertion time $t_i$ of each item $i \in S_{*T}$ is the one maximizing $p_{i t_i}$ over $\{t \in [T] : W_t \geq F(m-1, \psi_{m-1}, {\cal Q}_{>m-1}) + C_{\pi^*_{{\cal E}^*}}(i) \}$. Note that the latter set is indeed non-empty, since
    \begin{eqnarray*}
    C_{\pi^*_{{\cal E}^*}}(i) & \leq & w \left( {\cal E}^* \right) \\
    & = & F \left( m, \psi_m, {\cal Q}_{>m} \right) - F \left( m-1, \psi_{m-1}, {\cal Q}_{>m-1} \right) \\
    & \leq & W_T - F \left( m-1, \psi_{m-1}, {\cal Q}_{>m-1} \right) \ ,
    \end{eqnarray*}
    where the equality above is exactly~\eqref{eqn:weight_best}, and the last inequality holds since $F( m, \psi_m, {\cal Q}_{>m}) \leq W_T$, as assumed in Lemma~\ref{lem:approx_rec_eqn}. 
\end{itemize}
The next claim establishes the feasibility and profit guarantee of ${\cal S}_*$ with respect to $\hat{\cal I}^{\omega}$. Below, $\objfunc_{\omega}(\cdot)$ stands for the profit function with respect to this instance. 

\begin{claim} \label{cl:opt-is-feasible}
The chain ${\cal S}_*$ is feasible for $\hat{\cal I}^{\omega}$, with a profit of $\objfunc_{\omega}({\cal S}_*) = \sum_{i \in {\cal E}^*} \varphi_{ \pi^*_{{\cal E}^*} }^{ \rightsquigarrow }( i )$. 
\end{claim}
\begin{proof}
To prove the feasibility of ${\cal S}_*$, we first observe that, for every time period $t \in [T]$, 
\begin{equation} \label{eq:guessed-weight}
w\left( {S}_{*t} \right) \leq w({\cal E}^*) \leq \omega \ ,
\end{equation}
where the first inequality holds since $S_{*t} \subseteq S_{*T} = {\cal E}^*$, and the second inequality is precisely what Lemma~\ref{lem:profit-omega-eps-satisfied} assumes. In addition, by definition of ${\cal S}_*$, every item $i \in S_{*t}$ is associated with a completion time of $C_{\pi^*_{{\cal E}^*}}(i) \leq W_t - F(m-1, \psi_{m-1}, {\cal Q}_{>m-1})$. Thus, when the latter difference is negative, we have $S_{*t} = \emptyset$ and therefore $w(S_{*t}) = 0 \leq [W_t - \Delta]^+$. In the opposite case,
\begin{eqnarray}
 w\left( {S}_{*t} \right) & \leq & W_t - F(m-1, \psi_{m-1}, {\cal Q}_{>m-1}) \nonumber\\
& \leq & W_t - \Delta \nonumber \\
& \leq & [W_t - \Delta]^+ \ , \label{eq:overall-feasibility}
\end{eqnarray}
where the second inequality holds since $\Delta \leq F(m-1, \psi_{m-1}, {\cal Q}_{>m-1})$, as assumed in Lemma~\ref{lem:approx_rec_eqn}. Putting together inequalities~\eqref{eq:guessed-weight} and~\eqref{eq:overall-feasibility}, we have $w({S}_{*t} ) \leq \min \{[W_t- \Delta]^+, \omega \} =\hat{W}_t^{\omega}$, meaning that the chain ${\cal S}_*$ is indeed feasible for $\hat{\cal I}^{\omega}$.

Now, to derive the profit guarantee $\objfunc_{\omega}({\cal S}_*) = \sum_{i \in {\cal E}^*} \varphi_{ \pi^*_{{\cal E}^*} }^{ \rightsquigarrow }( i )$, we observe that since $\objfunc_{\omega}({\cal S}_*) = \sum_{ i \in {\cal E}^* } p_{i t_i}$, it suffices to show that $p_{it_i} = \varphi_{ \pi^*_{{\cal E}^*} }^{ \rightsquigarrow }( i )$ for each item $i \in {\cal E}^*$. To this end, note that our choice for the insertion time $t_i$ of each item $i \in {\cal E}^*$ exactly follows the definition of $\varphi_{ \pi^*_{{\cal E}^*} }^{ \rightsquigarrow }( i )$, implying that $p_{it_i} = \varphi_{ \pi^*_{{\cal E}^*} }^{ \rightsquigarrow }( i )$. 
\end{proof}

\paragraph{Concluding the proof.} Having established this claim, we are now ready to show that the item set ${\cal E}_{\omega}$ and permutation ${\pi}_{{\cal E}_{\omega}}$ satisfy $\sum_{i \in {{\cal E}_{\omega}}} \varphi_{ {\pi}_{{\cal E}_{\omega}} }^{ +\Delta }( i ) \geq (1 - \eps) \cdot ( \psi_{m} - \psi_{m-1} )$. For this purpose, similarly to $\objfunc_{\omega}(\cdot)$, let $\newobj_{\omega}(\cdot)$ be the profit function of a given permutation with respect to the instance $\hat{\cal I}^{\omega}$ in its sequencing formulation. With this notation, we obtain the required lower bound by arguing that
\begin{eqnarray*}
\sum_{i \in {{\cal E}_{\omega}}} \varphi_{ {\pi}_{{\cal E}_{\omega}} }^{ +\Delta }( i ) & = & 
\newobj_{\omega} \left( {\pi}_{{\cal E}_{\omega}} \right) \\
& \geq & \objfunc_{\omega} \left({\cal G}^{\omega} \cup {\cal R}^{{\cal G}^{\omega}} \right) \\
& \geq & (1 - \eps) \cdot (\psi_m - \psi_{m-1}) \ .
\end{eqnarray*}
We prove the first equality and second inequality in Claims~\ref{cl:delayed-completion} and~\ref{cl:marginal-profit}, respectively. To understand the first inequality, recall that the permutation ${\pi}_{{\cal E}_{\omega}}$ is constructed in step~3 according to Lemma~\ref{lem:reformulation} for the chain ${\cal G}^{\omega} \cup {\cal R}^{{\cal G}^{\omega}}$, which guarantees $\newobj_{\omega} ( {\pi}_{{\cal E}_{\omega}} ) \geq \objfunc_{\omega} ({\cal G}^{\omega} \cup {\cal R}^{{\cal G}^{\omega}} )$.

\begin{claim} \label{cl:delayed-completion}
$\sum_{i \in {{\cal E}_{\omega}}} \varphi_{ {\pi}_{{\cal E}_{\omega}} }^{ +\Delta }( i )  = \newobj_{\omega}({\pi}_{{\cal E}_{\omega}})$. 
\end{claim}
\begin{proof}
Let us use $\varphi^{\omega}_{{\pi}_{{\cal E}_{\omega}}}(i)$ to denote the profit contribution of item $i$ with respect to the permutation ${\pi}_{{\cal E}_{\omega}}$ in the instance $\hat{\cal I}^{\omega}$. In other words, $\varphi^{\omega}_{{\pi}_{{\cal E}_{\omega}}}(i) = \max \{ p_{it} : t \in [T+1] \text{ and } \hat{W}_t \geq C_{{\pi}_{{\cal E}_{\omega}} }( i ) \}$. With this notation, we have $\newobj_{ \omega }({\pi}_{{\cal E}_{\omega}})  = \sum_{i \in {\cal E}_{\omega} } \varphi^{\omega}_{ {\pi}_{{\cal E}_{\omega}} }( i )$, meaning that to prove the desired equality, it remains to show that $\varphi_{ {\pi}_{{\cal E}_{\omega}} }^{ +\Delta }( i ) = \varphi^{\omega}_{ {\pi}_{{\cal E}_{\omega}} }( i )$ for every item $i \in {\cal E}_{\omega}$. To verify this claim, note that
\begin{eqnarray*}
\varphi_{{\pi}_{{\cal E}_{\omega}}}^{+ \Delta}(i) & = & \max \left\{ p_{it} : t \in [T+1] \text{ and } W_t - \Delta  \geq C_{ {\pi}_{{\cal E}_{\omega}} }( i ) \right\} \\
& = & \max \left\{ p_{it} : t \in [T+1] \text{ and } [W_t - \Delta]^+  \geq C_{{\pi}_{{\cal E}_{\omega}} }( i ) \right\}\\
& = & \max \left\{ p_{it} : t \in [T+1] \text{ and } \min \{ [W_t - \Delta]^+, \omega \}  \geq C_{ {\pi}_{{\cal E}_{\omega}} }( i ) \right\}\\
& = & \max \left\{ p_{it} : t \in [T+1] \text{ and } \hat{W}_t \geq C_{ {\pi}_{{\cal E}_{\omega}} }( i ) \right\}\\ 
& = & \varphi^{\omega}_{{\pi}_{{\cal E}_{\omega}}}(i) \ .
\end{eqnarray*}
Here, the second equality holds since $C_{{\pi}_{{\cal E}_{\omega}}}(i)\geq 0$. The third equality is obtained by noting that $C_{{\pi}_{{\cal E}_{\omega}}}(i) \leq w({\cal E}_{\omega}) = w(G^{\omega}_T) + w(R^{{\cal G}^{\omega}}_T) \leq \hat{W}_T \leq \omega$, where the equality follows by definition of ${\cal E}_{\omega}$ and the second inequality is implied by the feasibility of ${\cal G}^{\omega} \cup {\cal R}^{{\cal G}^{\omega}}$ for the instance $\hat{\cal I}^{\omega}$. The last two equalities follow from the definitions of $\hat{W}_t$ and $\varphi^{\omega}_{{\pi}_{{\cal E}_{\omega}}}(i)$. 
\end{proof}

\begin{claim}\label{cl:marginal-profit}
$\objfunc_{\omega}( {\cal G}^{\omega} \cup {\cal R}^{{\cal G}^{\omega}} ) \geq (1 - \eps) \cdot (\psi_m - \psi_{m-1})$. 
\end{claim}
\begin{proof}
We begin by noting that since $({\cal E}^*, \pi^*_{{\cal E}^*}) \in \myextra[ \MyAbove{ (m, \psi_m, {\cal Q}_{>m}) }{ (m-1, \psi_{m-1}, {\cal Q}_{>m-1}) } ]$, this item set and permutation necessarily satisfy constraint~1, which informs us that ${\cal E}^* \subseteq ({\cal C}_m \uplus {\cal Q}_{>m}) \setminus {\cal Q}_{>m-1}$ and ${\cal Q}_{>m} \setminus {\cal Q}_{>m-1} \subseteq {\cal E}^*$. As a result, recalling that the collection of items introduced by the chain ${\cal S}_*$ is precisely ${\cal E}^*$, it follows that the latter chain can be expressed as ${\cal S}_* = {\cal S}_*|_{{\cal Q}_{>m} \setminus {\cal Q}_{>m-1}} \cup {\cal S}_*|_{{\cal C}_m \setminus {\cal Q}_{>m-1}}$. We remind the reader that, based on the terminology of Section~\ref{sec:qptas-one}, the first term ${\cal S}_*|_{{\cal Q}_{>m} \setminus {\cal Q}_{>m-1}}$ is the restriction of ${\cal S}_*$ to the items in ${\cal Q}_{>m} \setminus {\cal Q}_{>m-1}$, whereas the second term ${\cal S}_*|_{{\cal C}_m \setminus {\cal Q}_{>m-1}}$ is its restriction to ${\cal C}_m \setminus {\cal Q}_{>m-1}$. 

The crucial observation is that, since the chain ${\cal S}_*$ introduces all items in ${\cal Q}_{>m} \setminus {\cal Q}_{>m-1}$, its restriction ${\cal G}_* = {\cal S}_*|_{{\cal Q}_{>m} \setminus {\cal Q}_{>m-1}}$ is necessarily considered in step~2 of our algorithm; moreover, ${\cal S}_*|_{{\cal C}_m \setminus {\cal Q}_{>m-1}}$ constitutes a feasible chain for the residual instance $\hat{\cal I}^{\omega,-{\cal G}_*}$, by Lemma~\ref{lem:residual-chain}. As such, the corresponding chain ${\cal R}^{{\cal G}_*}$ we compute for the latter instance is guaranteed to have a profit of $\objfunc_{ \omega }( {\cal R}^{{\cal G}_*} ) \geq (1 - \eps) \cdot \objfunc_{ \omega }( {\cal S}_*|_{{\cal C}_m \setminus {\cal Q}_{>m-1}} )$. Consequently, since the chain ${\cal G}^{\omega}$ is the one  maximizing $\objfunc_{ \omega }({\cal G}^{\omega}) + \objfunc_{ \omega }( {\cal R}^{{\cal G}^{\omega}} )$ over all chains considered in step~2, we conclude that ${\cal G}^{\omega} \cup {\cal R}^{{\cal G}^{\omega}}$ is a feasible chain for $\hat{\cal I}^{\omega}$ with a profit of
\begin{eqnarray*}
\objfunc_{\omega} \left({\cal G}^{\omega} \cup {\cal R}^{{\cal G}^{\omega}} \right) & = & \objfunc_{\omega} \left({\cal G}^{\omega} \right) + \objfunc_{\omega} \left( {\cal R}^{{\cal G}^{\omega} } \right) \\ 
& \geq & \objfunc_{\omega} \left({\cal G}_* \right) + \objfunc_{\omega} \left( {\cal R}^{{\cal G}_*} \right)  \\
& \geq & \objfunc_{\omega} \left( {\cal S}_*|_{{\cal Q}_{>m} \setminus {\cal Q}_{>m-1}} \right) + (1 - \eps) \cdot \objfunc_{\omega} \left({\cal S}_*|_{{\cal C}_m \setminus {\cal Q}_{>m-1}} \right)  \\
& \geq & (1 - \eps) \cdot \objfunc_{\omega} \left({\cal S}_* \right)  \\
& = & (1 - \eps) \cdot \sum_{i \in {\cal E}^*} \varphi_{ \pi^*_{{\cal E}^*} }^{ \rightsquigarrow }( i ) \label{eq:shifted-opt} \\
& \geq & (1 - \eps) \cdot \left(\psi_m - \psi_{m-1} \right)  \ .
\end{eqnarray*}
Here, the first and second equalities follow from Lemma~\ref{lem:union-chain} and Claim~\ref{cl:opt-is-feasible}, respectively. The last inequality holds since $({\cal E}^*, \pi^*_{{\cal E}^*}) \in \myextra[ \MyAbove{ (m, \psi_m, {\cal Q}_{>m}) }{ (m-1, \psi_{m-1}, {\cal Q}_{>m-1}) } ]$ by definition, and hence, constraint~2 is necessarily satisfied.
\end{proof}

\subsection{Proof of Lemma~\ref{lem:approximate-F}} \label{app:proof_lem_approximate_F}

We prove the lemma by induction on $m$.

\paragraph{Base case: $\bs{m=0}$.} In this case, for any state with $F(0, \psi_0, {\cal Q}_{>0}) \leq W_T$, we actually have $\hat{F}(0, \psi_0, {\cal Q}_{>0}) = F(0, \psi_0, {\cal Q}_{>0})$, by the way terminal states of $\hat{F}$ are handled. In addition, letting $\hat{\pi}_{\hat{S}_0}$ be the permutation of $\hat{S}_0 = {\cal Q}_{>0}$ that attains $\hat{F}(0, \psi_0, {\cal Q}_{>0})$, it follows that $\hat{S}_0 \subseteq {\cal C}_{[1,0]}  \uplus {\cal Q}_{>0}$, ${\cal Q}_{>0} \subseteq \hat{S}_0$, and $\newobj(\hat{\pi}_{\hat{S}_0}) \geq \psi_0$, again by definition.

\paragraph{General case: $\bs{m \geq 1}$.} Let $(m, \psi_m, {\cal Q}_{>m})$ be a state for which $F(m, \psi_m, {\cal Q}_{>m}) \leq W_T$. We first show that $\hat{F}(m, \psi_m, {\cal Q}_{>m}) \leq F(m, \psi_m, {\cal Q}_{>m})$. To this end, recall that the function value $\hat{F}(m, \psi_m, {\cal Q}_{>m})$ is determined by minimizing $\hat{F}(m-1, \psi_{m-1}, {\cal Q}_{>m-1}) + w( \hat{\cal E} )$ over all conceivable states $(m-1, \psi_{m-1}, {\cal Q}_{>m-1})$, where the item set $\hat{\cal E}$ and its permutation $\hat{\pi}_{\hat{\cal E}} : \hat{\cal E} \to [|\hat{\cal E}|]$ are obtained by instantiating Lemma~\ref{lem:approx_rec_eqn} with $\Delta = \hat{F}(m-1, \psi_{m-1}, {\cal Q}_{>m-1})$ and satisfy $(\hat{\cal E}, \hat{\pi}_{\hat{\cal E}}) \in \myextra_{ \eps,\Delta }[ \MyAbove{ (m, \psi_m, {\cal Q}_{>m}) }{ (m-1, \psi_{m-1}, {\cal Q}_{>m-1}) } ]$. Therefore, specifically for the state $(m-1, \psi_{m-1}^*, {\cal Q}_{>m-1}^*) = \mybest(m, \psi_m, {\cal Q}_{>m})$, we have $\Delta = \hat{F}(m-1, \psi_{m-1}^*, {\cal Q}_{>m-1}^*) \leq F(m-1, \psi_{m-1}^*, {\cal Q}_{>m-1}^*)$ by the induction hypothesis. In turn, our auxiliary procedure computes a corresponding item set and permutation $(\hat{\cal E}^*, \hat{\pi}_{\hat{\cal E}}^*) \in \myextra_{ \eps,\Delta }[ \MyAbove{ (m, \psi_m, {\cal Q}_{>m}) }{ (m-1, \psi_{m-1}^*, {\cal Q}_{>m-1}^*) } ]$ with total weight $w(\hat{\cal E}^*) \leq F(m, \psi_m, {\cal Q}_{>m}) - F(m-1, \psi_{m-1}^*, {\cal Q}_{>m-1}^*)$, as guaranteed by Lemma~\ref{lem:approx_rec_eqn}. Consequently, 
\begin{eqnarray*}
\hat{F} \left( m, {\psi}_m, {\cal Q}_{>m} \right) & \leq & \hat{F} \left( m-1, \psi_{m-1}^*, {\cal Q}_{>m-1}^* \right) + w \left( \hat{\cal E}^* \right) \\
& \leq & F \left( m-1, \psi_{m-1}^*, {\cal Q}_{>m-1}^* \right) \\
& & \mbox{} + \left( F \left( m, \psi_m, {\cal Q}_{>m} \right) - F \left( m-1, \psi_{m-1}^*, {\cal Q}_{>m-1}^* \right) \right) \\
& = & F(m, \psi_m, {\cal Q}_{>m}) \ ,
\end{eqnarray*}
which is precisely the required upper bound on $\hat{F}(m, \psi_m, {\cal Q}_{>m})$.

Next, we show that $\hat{F}(m, \psi_m, {\cal Q}_{>m})$ is attained by an item set $\hat{S}_m$ and a permutation $\hat{\pi}_{\hat{S}_m}$ satisfying $\hat{S}_m \subseteq {\cal C}_{[1,m]} \uplus {\cal Q}_{>m}$, ${\cal Q}_{>m} \subseteq \hat{S}_m$, and $\newobj(\hat{\pi}_{\hat{S}_m}) \geq (1 - \eps) \cdot \psi_m$. For this purpose, let 
$(m-1, \psi_{m-1}, {\cal Q}_{>m-1})$, $\hat{\cal E}$, and $\hat{\pi}_{\hat{\cal E}}$ be the conceivable state, item set, and permutation at which $\hat{F}(m, \psi_m, {\cal Q}_{>m}) =\hat{F}(m-1, \psi_{m-1}, {\cal Q}_{>m-1}) + w( \hat{\cal E} )$ is attained, meaning in particular that ${\cal Q}_{>m-1} \setminus {\cal C}_m \subseteq {\cal Q}_{>m}$ by definition of conceivable states, and that $(\hat{\cal E}, \hat{\pi}_{\hat{\cal E}}) \in \myextra_{ \eps,\Delta }[ \MyAbove{ (m, \psi_m, {\cal Q}_{>m}) }{ (m-1, \psi_{m-1}, {\cal Q}_{>m-1}) } ]$ by the way general states of $\hat{F}$ are handled. We proceed by observing that, by the induction hypothesis, $\hat{F}(m-1, \psi_{m-1}, {\cal Q}_{>m-1})$ is attained by an item set $\hat{S}_{m-1}$ and a permutation $\hat{\pi}_{\hat{S}_{m-1}}$ satisfying $\hat{S}_{m-1} \subseteq {\cal C}_{[1,{m-1}]} \uplus {\cal Q}_{>{m-1}}$, ${\cal Q}_{>{m-1}} \subseteq \hat{S}_{m-1}$, and $\newobj(\hat{\pi}_{\hat{S}_{m-1}}) \geq (1 - \eps) \cdot \psi_{m-1}$. With these ingredients, let us define the item set $\hat{S}_m$ and permutation  $\hat{\pi}_{\hat{S}_m}$ as follows:
\begin{itemize}
    \item The item set $\hat{S}_m$ is given by $\hat{S}_m = \hat{S}_{m-1} \uplus \hat{\cal E}$. To understand why $\hat{S}_{m-1}$ and $\hat{\cal E}$ are disjoint, recall that $(\hat{\cal E}, \hat{\pi}_{\hat{\cal E}}) \in \myextra_{ \eps,\Delta }[ \MyAbove{ (m, \psi_m, {\cal Q}_{>m}) }{ (m-1, \psi_{m-1}, {\cal Q}_{>m-1}) } ]$, which implies by constraint~1 that $\hat{\cal E} \subseteq ({\cal C}_m \uplus {\cal Q}_{>m}) \setminus {\cal Q}_{>m-1} \subseteq {\cal C}_{[m,M]} \setminus {\cal Q}_{>m-1}$; however, $\hat{S}_{m-1} \subseteq {\cal C}_{[1,{m-1}]} \uplus {\cal Q}_{>{m-1}}$ by the induction hypothesis. These observations allow us to concurrently argue that $\hat{S}_m \subseteq {\cal C}_{[1,m]} \uplus {\cal Q}_{>m}$ as required, since $\hat{\cal E} \subseteq ({\cal C}_m \uplus {\cal Q}_{>m}) \setminus {\cal Q}_{>m-1} \subseteq {\cal C}_{[1,m]} \uplus {\cal Q}_{>m}$ and since 
    \begin{eqnarray*}
    \hat{S}_{m-1} & \subseteq & {\cal C}_{[1,{m-1}]} \uplus {\cal Q}_{>{m-1}} \\
    & \subseteq & {\cal C}_{[1,m]} \uplus ({\cal Q}_{>{m-1}} \setminus {\cal C}_m) \\
    & \subseteq & {\cal C}_{[1,m]} \uplus {\cal Q}_{>m} \ ,
    \end{eqnarray*}
    where the last inclusion follows by noting that ${\cal Q}_{>{m-1}} \setminus {\cal C}_m \subseteq {\cal Q}_{>m}$ due to state $(m-1, \psi_{m-1}, {\cal Q}_{>m-1})$ being conceivable. In addition, 
    \begin{eqnarray*}
    {\cal Q}_{>m} & \subseteq & {\cal Q}_{>m-1} \uplus ({\cal Q}_{>m} \setminus {\cal Q}_{>m-1}) \\
    & \subseteq & \hat{S}_{m-1} \uplus \hat{\cal E} \\
    & = & \hat{S}_m \ ,
    \end{eqnarray*}
    where the second inclusion holds since ${\cal Q}_{>{m-1}} \subseteq \hat{S}_{m-1}$ by the induction hypothesis and since ${\cal Q}_{>m} \setminus {\cal Q}_{>m-1} \subseteq \hat{\cal E}$, again by constraint~1. 

\item To define the permutation $\hat{\pi}_{\hat{S}_m} : \hat{S}_m \to [| \hat{S}_m |]$, we simply append $\hat{\pi}_{\hat{\cal E}}$ to $\hat{\pi}_{\hat{S}_{m-1}}$. As a result, we obtain a profit of
\begin{eqnarray*}
    \newobj \left( \hat{\pi}_{\hat{S}_m} \right) & = & \newobj \left( \hat{\pi}_{\hat{S}_{m-1}} \right) + \sum_{i \in \hat{\cal E}}  \varphi_{ \hat{\pi}_{\hat{\cal E}} }^{ +w( \hat{S}_{m-1} ) }( i ) \\
    & = & \newobj \left( \hat{\pi}_{\hat{S}_{m-1}} \right) + \sum_{i \in \hat{\cal E}}  \varphi_{ \hat{\pi}_{\hat{\cal E}} }^{ +\Delta }( i ) \\
    & \geq & (1 - \eps) \cdot {\psi}_{m-1} + (1 - \eps) \cdot (\psi_m - {\psi}_{m-1}) \\
    & = & (1 - \eps) \cdot {\psi}_m \ .
\end{eqnarray*}
Here, the second equality holds since $w( \hat{S}_{m-1} ) = \hat{F}(m-1, \psi_{m-1}, {\cal Q}_{>m-1}) = \Delta$. To understand the inequality above, note that $\newobj(\hat{\pi}_{\hat{S}_{m-1}}) \geq (1 - \eps) \cdot {\psi}_{m-1}$ by the inductive hypothesis, and in addition, $\sum_{i \in \hat{\cal E}}  \varphi_{ \hat{\pi}_{\hat{\cal E}} }^{ +\Delta }( i ) \geq (1 - \eps) \cdot (\psi_m - {\psi}_{m-1})$, since $(\hat{\cal E}, \hat{\pi}_{\hat{\cal E}}) \in \myextra_{ \eps,\Delta }[ \MyAbove{ (m, \psi_m, {\cal Q}_{>m}) }{ (m-1, \psi_{m-1}, {\cal Q}_{>m-1}) } ]$ implies that constraint~2 is $(\eps, \Delta)$-satisfied.
\end{itemize}

\end{document}